\documentclass[11pt,a4paper]{article}

\usepackage[a4paper,top=2.5cm,bottom=2.5cm,left=3cm,right=3cm]{geometry}
\usepackage{microtype}
\usepackage{setspace}
\usepackage{empheq}

\usepackage[dvipsnames]{xcolor}
\usepackage[
	colorlinks=true,
	linkcolor=RoyalBlue,
	citecolor=OliveGreen,
	urlcolor=RoyalBlue
]{hyperref}

\usepackage{appendix}

\usepackage[utf8]{inputenc}
\usepackage[T1]{fontenc}
\usepackage{lmodern}
\usepackage[english]{babel}

\usepackage{amsmath, amssymb, amsfonts, amsthm}
\usepackage{mathtools}
\usepackage{mathrsfs}
\usepackage{dsfont}

\usepackage{thmtools}

\usepackage{microtype}
\usepackage{bm}

\usepackage{graphicx}
\usepackage{tikz}
\usepackage{tikz-cd}
\usepackage{pgfplots}
\pgfplotsset{compat=1.18}

\usepackage{booktabs}
\usepackage{array}

\usepackage{xcolor}

\usepackage{float}
\usepackage{placeins}

\usepackage{tcolorbox}

\usepackage{enumitem}

\usepackage[capitalise,noabbrev]{cleveref}

\usepackage{bbm}

\usepackage{appendix}

\numberwithin{equation}{section}

\theoremstyle{plain}
\newtheorem{theorem}{Theorem}[section]
\newtheorem{lemma}[theorem]{Lemma}
\newtheorem{proposition}[theorem]{Proposition}
\newtheorem{corollary}[theorem]{Corollary}
\newtheorem{question}[theorem]{Question}

\theoremstyle{definition}
\newtheorem{definition}[theorem]{Definition}
\newtheorem{assumption}[theorem]{Assumption}
\newtheorem{remark}[theorem]{Remark}
\newtheorem{problem}[theorem]{Problem}

\theoremstyle{remark}
\newtheorem{example}[theorem]{Example}
\newtheorem{notation}[theorem]{Notation}

\newcommand{\f}[2]{\frac{#1}{#2}}

\newcommand{\lp}{\left(}
\newcommand{\rp}{\right)}

\newcommand{\R}{\mathbb{R}}
\newcommand{\Z}{\mathbb{Z}}
\newcommand{\C}{\mathbb{C}}
\newcommand{\N}{\mathbb{N}}

\newcommand{\bP}{\mathbb{P}}

\newcommand{\bS}{\mathbb{S}}

\newcommand{\cA}{\mathcal{A}}

\newcommand{\cC}{\mathcal{C}}
\newcommand{\cD}{\mathcal{D}}

\newcommand{\cF}{\mathcal{F}}

\newcommand{\cK}{\mathcal{K}}
\newcommand{\cL}{\mathcal{L}}
\newcommand{\cM}{\mathcal{M}}
\newcommand{\cN}{\mathcal{N}}

\newcommand{\cR}{\mathcal{R}}
\newcommand{\cS}{\mathcal{S}}
\newcommand{\cT}{\mathcal{T}}
\newcommand{\cU}{\mathcal{U}}

\newcommand{\fL}{\mathfrak{L}}

\newcommand{\fS}{\mathfrak{S}}

\DeclareMathOperator{\tr}{Tr}

\DeclareMathOperator{\sgn}{sgn}

\newcommand\1{{\ensuremath {\mathds 1} }}

\renewcommand{\geq}{\geqslant}

\renewcommand{\hat}{\widehat}

\newcommand{\rd}{\mathrm{d}}
\newcommand{\dx}{\rd x}

\setlist[enumerate,1]{label=(\roman*)}
\newcommand{\proofstep}[1]{\par\smallskip\noindent\emph{#1}}
\usepackage{csquotes}
\usepackage{refcount}
\usepackage[shortcuts]{extdash}
\usepackage[
	backend=biber,
	style=alphabetic,
	sorting=none,
	maxbibnames=99,
	url=false,
	doi=true,
	isbn=false,
]{biblatex}

\title{V-Representability Beyond Singular Times in Time-Dependent Density Functional Theory: The Hubbard Dimer Case}

\author{
	\'Eric Canc\`es\thanks{CERMICS, CNRS, \'Ecole des Ponts -- Institut Polytechnique de
	Paris and Inria Paris, 6-8 avenue Blaise Pascal, Cité Descartes, 77455 Marne-la-Vallée, France}\\
	\small\texttt{eric.cances@enpc.fr}
	\and
	Th\'eo Duez\footnotemark[1]\\
	\small\texttt{theo.duez@enpc.fr}
}
\date{\today}

\begin{document}

\maketitle

\begin{abstract}
We study the global $V$-representability problem in time-dependent density functional theory (TDDFT) for the two-particle Hubbard dimer model in the spin-singlet sector. In particular, we investigate whether the time-dependent density of the interacting system obtained for a given one-body external potential $u$ (the tunable control variable) can be reproduced by a non-interacting Kohn--Sham system beyond the singular times at which the usual inversion formulas break down. We deal with both the conventional formulation of TDDFT based on some time-dependent variational principle and the recently introduced geometric formulation of TDDFT (Phys. Rev. Lett. 136 (2026) 256401, Phys. Rev. A 113 (2026) 062222). For both principles, we prove global-in-time $V$-representability, show that singular times induce non-uniqueness of the time-dependent effective potential, and classify the set of solutions.
\end{abstract}

\tableofcontents
\section{Introduction}

Density functional theory (DFT) is one of the cornerstones of modern quantum chemistry and condensed matter physics. Its central idea~\cite{HK64,Levy79,Lieb83} is to replace the many-body wavefunction $\psi(x_1,\dots,x_N)$, whose number of variables grows linearly with the number of electrons $N$, by the associated one-body density $\rho_\psi(x) = N \int |\psi(x,x_2,\dots,x_N)|^2 \, \dx_2 \cdots \dx_N$, a function of the position variable of a single particle, as the main variable of the problem. In practice, the overwhelming success of DFT rests on the Kohn--Sham scheme~\cite{KS65}, whose principle consists in reproducing the ground-state density $\rho_0$ of the interacting physical system under consideration through an auxiliary system of fictitious non-interacting electrons subjected to an effective local potential $v$, called the Kohn--Sham potential. The potential $v$ is not known explicitly, and is formally the  derivative of some well-defined functional of the density, but with no known explicit expression. In fact, $v$ is even known not to exist for some particular systems for which this functional is not differentiable at $\rho_0$. The practical success of DFT stems from the fact that computationally tractable and accurate enough (for most applications) approximations of this functional have been developed over the past 60 years. Thanks to such approximations, an intractable many-body problem can be turned into a tractable one-body problem. This approach is very successful and earned Walter Kohn the 1998 Nobel Prize in Chemistry.

The theoretical foundation of DFT relies on Hohenberg--Kohn theorems~\cite{HK64} and Levy's constrained search principle~\cite{Levy79}. We refer to the seminal work by Lieb~\cite{Lieb83} on the mathematical formulation of DFT, and to Garrigue~\cite{Garrigue2018} and Lammert~\cite{Lammert2018} for more refined results on the Hohenberg--Kohn theorems. The Kohn--Sham scheme leads to an additional theoretical question, namely the existence and uniqueness (up to an additive constant) of a local potential $v$ for which the non-interacting system subjected to the potential $v$ has the same ground-state density as the interacting system under consideration. This problem was coined the non-interacting $V$-representability problem~\cite{Levy82}, and was mathematically investigated by several authors~\cite{Lieb83,EnglischEnglisch1983,Lammert2006,Lammert2007}. Let us also mention the Moreau--Yosida regularization approach introduced by Kvaal and collaborators~\cite{KvaalEkstromTealeHelgaker2014} allowing one to somehow bypass the $V$-representability problem and introduce a regularized bijective potential-to-density mapping (see also~\cite{PenzHerbstHelgakerLaestadius2026} for a recent review on this approach, and references therein). We also refer to~\cite{CF23} for other recent mathematical contributions to DFT and focused reviews on several aspects of DFT.

Let us recall that Kohn--Sham DFT is a ground-state theory and cannot be directly applied to the computation of electronic excitation energies and optical spectra, nor to the investigation of out-of-equilibrium quantum dynamics, a fast-growing topic with applications in the emerging fields of attosecond physics and chemistry (2023 Nobel Prize in Physics~\cite{Nobel23}) and quantum technologies. Time-dependent density functional theory (TDDFT)~\cite{Ull11,Mar+12} aims at filling this gap. In the same spirit as for static DFT, practical implementations of TDDFT rely on a time-dependent Kohn--Sham scheme, and the replacement of an interacting quantum dynamics by a non-interacting one reproducing the time-dependent density of the interacting system. The construction of approximate time-dependent Kohn--Sham potentials $v$ leading to accurate approximations of out-of-equilibrium electronic dynamics is a very challenging problem. Despite many attempts to construct better approximations, the vast majority of TDDFT computations are done using adiabatic approximations ~\cite{LM23,MBW02}, in which memory effects are neglected: the time-dependent Kohn--Sham potential at time $t$ is approximated by approximations of the static Kohn--Sham potential associated with the density at time $t$. In particular, the basic adiabatic local-density approximation (ALDA) introduced in the early TDDFT papers is still commonly used. Adiabatic approximations perform reasonably well close to equilibrium, in the linear response regime. As a matter of fact, (adiabatic) linear-response TDDFT is one of the state-of-the-art methods to compute excitation energies and optical spectra of atoms and molecules as it often provides reasonable approximations at low computational cost for a large class of molecular systems~\cite{Ull25}. In strongly out-of-equilibrium regimes, adiabatic approximations completely break down~\cite{Tea+22,LM23,FFTAKR13,FM14}. The construction of approximate nonadiabatic time-dependent Kohn--Sham potentials able to simulate e.g. charge-transfer excitations~\cite{Fuks16}, high-harmonic generation, or attosecond physics experiments, remains a major challenge.

Motivated by the practical limitations of adiabatic approximations, a new formulation of TDDFT based on a \emph{geometric principle} (GP) was recently introduced~\cite{geoPRA,geoPRL}. In the GP formulation, the \emph{real} effective time-dependent local potential $v$ of conventional TDDFT, whose construction follows from some time-dependent \emph{variational principle} (VP), is replaced by a \emph{purely imaginary} effective time-dependent local potential $iw$. This term is chosen so that the Kohn--Sham dynamics reproducing the target density remains as close as possible, in some geometric sense, to the interacting one. The so-obtained GP dynamics differs markedly from the VP dynamics of conventional TDDFT. One motivation for further investigating the former is the prospect of building approximation strategies for $w$ beyond the standard adiabatic framework.

The non-interacting $V$-representability problem introduced above in the static case can be formulated as follows in the time-dependent case: given a prescribed interacting density trajectory, does there exist a non-interacting Kohn--Sham dynamics with real (VP) or purely imaginary (GP) time-dependent potential reproducing the interacting density at all times? If so, what can be said about the uniqueness and regularity properties of this potential? From a control-theoretic viewpoint, this asks whether a prescribed output -- the interacting density trajectory -- can be tracked by a suitable time-dependent control -- the effective potential $v$ for VP or $iw$ for GP -- of the non-interacting dynamics. This connection with control theory has already been exploited numerically in local-control schemes~\cite{NRvL13}. The $V$-representability problem remains largely open, and has recently regained attention~\cite{RPvL15,FLLO}.

The Runge--Gross theorem~\cite{RG84}, the time-dependent analogue of the Hohenberg--Kohn theorems, provides partial answers to the uniqueness question for conventional TDDFT. It guarantees that, under some assumptions, if two effective potentials $v_1$ and $v_2$ lead to the same one-body density, then they are equal up to a time-dependent constant.   This original uniqueness result~\cite{RG84}, as well as the existence result in \cite{vanLeeuwen99}, relied on strong (unrealistic for applications) time-analyticity assumptions~\cite{MTWB10,YMB12,YB13,FLLO}. A breakthrough was recently made in~\cite{Lauritsen2026a,Lauritsen2026b}, where the short-time existence and uniqueness (up to a time-dependent constant) of a non-interacting time-dependent potential $v$ reproducing the time-dependent density of the interacting system was proved for an $N$-electron system on a three-dimensional torus with Coulomb interactions. Importantly, no time-analyticity assumption is required. The assumptions are that the time-dependent density $\rho(t,x)$ is $C^2$ in time, real-analytic in space, and bounded away from zero.

For both VP and GP principles, the effective potential solves at each time a linear system, obtained by differentiating the density constraint in time: the \emph{van Leeuwen equation}~\cite{vanLeeuwen99} for VP, and a \emph{modified continuity equation} for GP~\cite{geoPRA}. When the associated linear map is invertible, this yields local existence and uniqueness of the time-dependent effective potential ($v$ for VP, $w$ for GP). For finite systems (finite lattice models)~\cite{FT12,Tok11} or for $N$ electrons on a torus with Coulombic interactions (and under suitable assumptions, see above)~\cite{Lauritsen2026a,Lauritsen2026b}, this yields \emph{local-in-time} non-interacting $V$-representability. The difficulty is that this map may degenerate at \emph{singular times}, where the usual inversion formulas break down and it is a priori unclear whether the non-interacting dynamics can be continued. On finite lattices, such degeneracies are not artifacts: they reflect in particular (but not only) a genuine constraint on how fast the density may vary~\cite{Baer08,LU08}. This loss of local solvability is an obstruction to \emph{global-in-time} non-interacting $V$-representability.

To investigate these questions, we focus on the Hubbard dimer, one of the simplest interacting quantum systems retaining non-trivial correlation effects, for which explicit expressions are known. Introduced as a minimal model of strongly-correlated electrons, it describes two interacting electrons on a two-site lattice, and has become a paradigmatic test case for both DFT and TDDFT, see e.g.~\cite{CFSB15, CFMB18}. In particular, it is a reference model for which the limitations of adiabatic approximations have been thoroughly analyzed \cite{FFTAKR13,FM14}. Restricting the dynamics to the spin-singlet sector reduces the interacting quantum dynamics to a dynamical system on the unit sphere of $\C^3$ (or in fact precisely the complex projective space $\C{\rm P}^2$), in which the external one-body time-dependent potential $u$ plays the role of a control variable. In~\cite{FT12,FFTAKR13,FM14b,RPvL15}, the existence of singular times for well-chosen time-dependent external potentials was proved, but the question of whether the non-interacting dynamics can be continued across the singularity was left open.

In this work, we answer this question affirmatively for the two-particle Hubbard dimer in the spin-singlet sector, for both the variational and the geometric principles. To the best of our knowledge, this is the first result showing that non-interacting TDDFT dynamics can be continued across singular times. We show that singular times do occur for well-chosen external potentials because of the controllability properties of the interacting Hubbard dimer (see \cref{appendix:controllability}). They do not obstruct global-in-time non-interacting $V$-representability, but lead to non-uniqueness of the effective potential $v$ (for VP) or $iw$ (for GP) and corresponding non-interacting trajectories.

The article is organized as follows. In \cref{sec:hubbard}, we detail the formulation of the interacting dynamics of the Hubbard dimer model and the corresponding non-interacting dynamics on the Bloch sphere, and state our main result (\cref{thm:NIVIN}): any interacting dynamics is non-interacting $V$-representable for both VP and GP \emph{globally in time}. In \cref{sec:VP}, we provide a complete characterization of the set of VP-representable densities for the Hubbard dimer in a suitable functional setting, and for a given VP-representable density, we classify the (in general non-unique) non-interacting trajectories and effective potentials VP-representing this density (\cref{thm:VP-characterization}). In \cref{sec:GP}, we state similar results for GP and GP (\cref{thm:GP-characterization}). In \cref{sec:comparison}, we compare the VP and GP representability results, and show that any VP-representable density is GP-representable, but that the converse is not true. In \cref{sec:interacting}, we focus on interacting densities: we prove our main result (\cref{thm:NIVIN}), show that the non-interacting trajectories cross the singular sets at finitely many times, and describe all of them together with bounds on the corresponding effective potentials (\cref{thm:int-representability}). In \cref{sec:integrability}, we show that the geometric potentials reproducing an interacting density are in general not integrable at singular times, which justifies the functional setting chosen for GP. In \cref{sec:numerics}, we illustrate these results on explicit trajectories obtained with a bang-bang external potential, for which trajectories can be computed explicitly. In  \cref{sec:proofs}, we state and prove some intermediate technical results and provide the proofs of the other results. In \cref{appendix:derivation}, we derive the dynamical system describing the time-evolution of the Hubbard dimer subjected to an external potential in the two-particle spin-singlet sector of the Fock space from the usual formulation of the Hubbard model in second quantization. In \cref{appendix:controllability}, we prove the controllability of this dynamical system.\medskip

\section{Statement of the problem} \label{sec:hubbard}

In this section, we recall the mathematical formulation of the interacting Hubbard dimer model in the two-particle spin-singlet sector, introduce the $V$-representability problems, for both the variational (VP) and geometric (GP) formulations of TDDFT, and present the main results of this article. The models are written directly as low-dimensional dynamical systems and their derivations from the second-quantized many-body formalism are postponed until \cref{appendix:derivation}.

\subsection{Formulation of the interacting and non-interacting models}\label{subsec:interacting}

The Hubbard dimer model describes interacting electrons on a two-site lattice, each carrying a spin-$1/2$ degree of freedom. Restricting to the two-particle spin-singlet sector of the Fock space -- i.e. to states with two electrons and zero total spin --, which is preserved by the dynamics, reduces the state space to a three-dimensional complex vector space, which we identify with $\C^3$ once and for all. Following the axioms of quantum mechanics, a physical state is a unit vector
\begin{equation*}
	\Psi = \begin{pmatrix} z_1 \\ z_2 \\ z_3 \end{pmatrix} \in \C^3, \qquad \|\Psi\|^2 = |z_1|^2+|z_2|^2+|z_3|^2 = 1 ,
\end{equation*}
defined up to a global phase, two vectors differing by a factor $e^{i\vartheta}$, $\vartheta \in \R$,  representing the same state. In the chosen basis, the three coordinates correspond respectively to both electrons occupying site~$1$, both occupying site~$2$, and one electron on each site. In the following, for each physical state $\Psi$, we denote by $P_\Psi = \Psi\Psi^* \in \C^{3\times 3}$ the associated rank-one projector.

\medskip

The quantities of interest in TDDFT are the on-site density $\rho=(\rho_1,\rho_2)^T \in [0,2]^2$ and the on-site current $j=(j_1,j_2)^T \in \R^2$. For a normalized wavefunction $\Psi = (z_1,z_2,z_3)^T \in \C^3$, they are defined by
\begin{align}
	\text{(on-site density)} && \rho^\Psi_i &:= 2|z_i|^2 + |z_3|^2, \label{eq:density}\\[0.35cm]
	\text{(on-site current)} && j^\Psi_i &:= 2\sqrt{2}\,\tau\, (-1)^{i+1} \, \Im\lp z_3(\overline{z_1}-\overline{z_2}) \rp. \label{eq:curent}
\end{align}
Note that $\rho^\Psi_i=2$ means that both electrons are on site $i$, and that $\rho^\Psi_1+\rho^\Psi_2=2$ for any normalized state $\Psi$, so that the density carries a single scalar degree of freedom.

\medskip

In mean-field models such as (static or time-dependent) Hartree--Fock or Kohn--Sham models, admissible states are Slater states generated by Slater determinants. In the simple case of the two-particle spin-singlet Hubbard dimer, the set of Slater state wavefunctions takes the simple form
\begin{equation}\label{eq:def_mks}
	\cM^{\rm Slat} := \left\{ \Phi=(z_1,z_2,z_3)^T \in \C^3 \;\middle|\; \|\Phi\|=1, \; z_3^2=2z_1z_2 \right\}.
\end{equation}
It is customary in quantum chemistry to denote generic many-body wavefunctions by the Greek letter $\Psi$, and Slater state wavefunctions by $\Phi$. Taking the gauge invariance $\Phi \mapsto e^{i\vartheta} \Phi$ into account, the set of Slater states of the two-particle spin-singlet Hubbard dimer is isometric to the Bloch sphere (see \cref{subsec:bloch} below).

\medskip

The dynamics of the interacting systems on the time interval $[0,T]$, with $T > 0$, reads as
\begin{equation}\label{eq:interacting_dynamics}
	\begin{cases}
		i\,\dot\Psi^u(t) = H_u (t)\, \Psi^u, & t \in (0,T], \\
		\Psi^u(0) = \Psi_0^{\rm int} ,
	\end{cases}
\end{equation}
where
\begin{equation*}
	H_{u}(t) := K + U + u(t)V
\end{equation*}
is the time-dependent Hamiltonian constructed from the three Hermitian matrices of $\R^{3\times 3}$
\begin{equation*}
	K := \begin{pmatrix} 0 & 0 & -\sqrt2\,\tau \\ 0 & 0 & -\sqrt2\,\tau \\
		-\sqrt2\,\tau & -\sqrt2\,\tau & 0 \end{pmatrix}, \quad U := \begin{pmatrix} U_1 & 0 & 0 \\ 0 & U_2 & 0 \\ 0 & 0 & 0 \end{pmatrix}, \quad V := \begin{pmatrix} 2 & 0 & 0 \\ 0 & -2 & 0 \\ 0 & 0 & 0 \end{pmatrix},
\end{equation*}
respectively describing electrons hopping between the two sites -- kinetic energy component of the Hamiltonian --, on-site electron-electron repulsion, and the site potential imbalance. The function $u\in L^1(0,T;\R)$ is the amplitude of the time-dependent potential with which the system can be controlled (when the system is at rest, $u$ is equal to a constant $u_0 \in \R$). We consider controls $u$ in $L^1$, in order to accommodate the discontinuous or rapidly varying external electromagnetic fields arising in quantum control and out-of-equilibrium dynamics. Here $\tau > 0$ is the hopping parameter, and $U_1, U_2 > 0$ are the interaction strengths. In practice, the initial condition $\Psi_0^{\rm int}$ is often the ground-state wavefunction of the interacting system at rest (i.e. of the matrix $K+U+u_0V$), but our results hold for any normalized $\Psi_0^{\rm int} \in \C^3$.
\medskip

For our analysis, it is convenient to embed~\eqref{eq:interacting_dynamics} into the family of (possibly nonlinear) Schr\"odinger dynamics
\begin{equation}\label{eq:schrodinger}
	\begin{cases}
		i\,\dot\Psi(t) = H^\lambda_{v,w}(t,\Psi)\, \Psi, & t \in (0,T], \\
		\Psi(0) = \Psi_0 ,
	\end{cases}
\end{equation}
where $H^\lambda_{v,w}(t,\Psi)\in \C^{3\times 3}$ is the Hamiltonian defined by
\begin{equation}\label{eq:ham_general}
	H^\lambda_{v,w}(t,\Psi) := K + \lambda U + v(t)V - iw(t)[V,P_\Psi],
\end{equation}
with $\lambda\ge0$, $v$ and $w$ real-valued functions of time, and $\Psi_0$ a normalized vector of $\C^3$. This family indeed encompasses the interacting dynamics~\eqref{eq:interacting_dynamics}, as well as the non-interacting variational principle (VP), geometric principle (GP), and corrected geometric principle (GPc) dynamics:
\begin{equation*}
	\renewcommand{\arraystretch}{1.4}
	\begin{array}{c@{\qquad \qquad}c@{\qquad \qquad}c}
	\text{Dynamics} & \text{Hamiltonian} & \text{Initial condition} \\[1mm]
	\eqref{eq:interacting_dynamics} & H^1_{u,0}(t) \quad \text{(interacting, linear)}   & \Psi_0=\Psi_0^{\rm int}\\[1mm]
		\text{VP} &  H^0_{v,0}(t)\quad \text{(non-interacting, linear)} &\Psi_0=\Phi_0 \in \cM^{\rm Slat} \\[1mm]
		\text{GP} & H^0_{0,w}(t,\Psi) \quad \text{(non-interacting, nonlinear)} & \Psi_0=\Phi_0 \in \cM^{\rm Slat} \\[1mm]
		\text{GPc} &   H^0_{v_c,w_c}(t,\Psi) \quad \text{(non-interacting, nonlinear)} & \Psi_0=\Phi_0 \in \cM^{\rm Slat}
	\end{array}
\end{equation*}

\medskip

The interacting dynamics \eqref{eq:interacting_dynamics} describes the physical system. The density associated with $\Psi^u$ will be denoted by $\rho^u$ to simplify the notation. It will serve as the target density throughout the paper. The VP, GP, and GPc dynamics are auxiliary non-interacting dynamics ($\lambda=0$), whose purpose is to reproduce this target density at all times by suitable choices of the initial conditions, and the potentials $v$, $w$, $w_c$ (in GPc, the potential $v_c$ is supposed to be given). As the VP, GP, and GPc initial conditions are chosen in the set of Slater states, the VP, GP, and GPc wavefunctions remain Slater states at all times (this is a well-known property of non-interacting dynamics, see also \cref{prop:slater_set}).
The GPc dynamics is a variant of GP in which the real potential $v_c$ is given. In practice, this potential may be equal to $u$, the external potential of the physical system, or to any approximation of the exact VP potential $v$ (for instance its adiabatic approximation), the approximation error being corrected by the geometric control $w_c$, hence the “c” in GPc.\medskip

In this paper, we are interested in the existence of non-interacting trajectories reproducing interacting densities (see \cref{subsec:TDDFT-V} for the precise formulation of the problems), and we require these trajectories to belong to $W^{1,1}(0,T;\C^3)$. Let us recall that $W^{1,1}(0,T;\C^d)$, for $d\in\N^*$, is the space of absolutely continuous functions defined on $[0,T]$ and with values in $\C^d$. From an existence point of view, one can then ask what is the largest class of potentials compatible with a given trajectory in $W^{1,1}$. To answer this question, we use the following general notion of solution. Let $L^0(0,T;\R)$ denote the space of real-valued Borel functions on $(0,T)$, in which, as usual, functions taking the same values almost everywhere are identified. For $\lambda\ge0$ and $v,w\in L^0(0,T;\R)$, a function $\Psi\in W^{1,1}(0,T;\C^3)$ is a solution to~\eqref{eq:schrodinger} if $\Psi$ is absolutely continuous, $\Psi(0)=\Psi_0$, and the ODE in~\eqref{eq:schrodinger} holds in $L^1(0,T;\C^3)$. This requires the RHS to be in $L^1(0,T;\C^3)$ which, since $K+\lambda U$ is bounded and the vectors $vV\Psi$ and $-iw[V,P_\Psi]\Psi$ are orthogonal for $\Re\langle\cdot,\cdot\rangle$, amounts to
\begin{equation}
	\label{eq:integrability-vw}
	v\,\|V\Psi\|\in L^1(0,T;\R) \qquad\text{and}\qquad w\,\|[V,P_\Psi]\Psi\|\in L^1(0,T;\R).
\end{equation}
For the Slater trajectories of the non-interacting dynamics, one has $\sqrt 2 \le \|V\Phi\|\le 2$ and $\|[V,P_\Phi]\Phi\|=\sqrt{2\rho_1^\Phi\rho_2^\Phi}$, so that these conditions read
\begin{equation*}
	v\in L^1(0,T;\R) \qquad\text{and}\qquad w\sqrt{\rho_1^\Phi\rho_2^\Phi}\in L^1(0,T;\R).
\end{equation*}
This yields two different functional settings for the two principles. For the physical dynamics~\eqref{eq:interacting_dynamics}, we choose $u\in L^1(0,T;\R)$, as explained above. For VP, the above condition shows that $v\in L^1(0,T;\R)$ is necessary, and we also assume $v_c\in L^1(0,T;\R)$ for GPc. For potentials in $L^1$, \eqref{eq:schrodinger} is a Carath\'eodory--Cauchy problem, whose unique global solution is in $W^{1,1}$ (see e.g.~\cite{Walter1998}). For GP and GPc, the natural class is that of measurable $w$ (resp. $w_c$) with $w\sqrt{\rho_1\rho_2}\in L^1(0,T;\R)$. This larger class is needed: the geometric potential reproducing an interacting density is in general not integrable near the times at which a site density vanishes (see \cref{sec:integrability}), even for smooth external potentials $u\in C^\infty_{\rm c}(0,T;\R)$. The price to pay is that this weighted condition depends on the trajectory itself, so that a given effective potential $w$ or $w_c$ need not generate a unique trajectory, and only the uniqueness of the density trajectory is guaranteed by construction.

\medskip

Let us start with a first useful result, concerning the preservation of the set of Slater states.
\begin{proposition}\label{prop:slater_set}
	Let $\lambda \ge 0$, $v\in L^1(0,T;\R)$, $w\in L^0(0,T;\R)$, let $\Psi \in W^{1,1}(0,T;\C^3)$ be a solution of \eqref{eq:schrodinger}, and let
	\begin{equation*}
	\cT^{\rm Slat}:=\left\{t\in[0,T]\;\middle|\;\Psi(t)\in \cM^{\rm Slat} \right\}
	\end{equation*}
	be the set of times at which $\Psi(t)$ is a Slater state wavefunction.
	\begin{enumerate}
		\item \textup{Non-interacting case, $w\in L^1$:} if $\lambda=0$, $w\in L^1(0,T;\R)$ and $\Psi_0=\Phi_0 \in \cM^{\rm Slat}$, then $\cT^{\rm Slat} = [0,T]$.
		\item \textup{Non-interacting case, $w\in L^0$:} if $\lambda=0$ and $\Psi_0=\Phi_0 \in \cM^{\rm Slat}$, then $[0,T_*]\subset\cT^{\rm Slat}$, where
		\begin{equation*}
			T_*:=\inf\left\{t\in[0,T]\;\middle|\;\Psi(t)\in\C (1,0,0)^T \cup\C (0,1,0)^T \right\}
		\end{equation*}
		($T_*:=T$ if this set is empty) is the first time the state $\Psi(t)$ becomes one of the two Slater states represented by a pole on the Bloch sphere.
		\item \textup{Interacting case:} if $\lambda > 0$ and $w=0$, then $\cT^{\rm Slat}$ is discrete and therefore finite.
	\end{enumerate}
\end{proposition}
The proof of \Cref{prop:slater_set} is given in \Cref{ssec:proof-slater}.

By point (i), the VP-trajectories, as well as the GP- and GPc-trajectories with integrable potentials, which start from $\Phi_0\in\cM^{\rm Slat}$, remain Slater states at all times. By point (ii), this remains true for measurable geometric potentials until the first time at which the trajectory reaches a pole, but not beyond in general, since the geometric potential may fail to be integrable near the poles (see \cref{sec:integrability}). For this reason, the non-interacting problems are formulated directly on the Bloch sphere in \cref{subsec:bloch}, that is, we only consider Slater trajectories. Point (iii) will be useful later to justify that non-interacting dynamics cross singular sets at isolated times when reproducing the interacting density. \medskip

The second result, which only requires the effective potentials to be measurable, is concerned with the time-regularity of the density $\rho^\Psi(t)$ and current $j^\Psi(t)$ when $\Psi$ is solution to~\eqref{eq:schrodinger}.

\begin{proposition}\label{prop:reg_density}
	Let $\lambda \ge 0$, $v,w\in L^0(0,T;\R)$, and $\Psi_0 \in \C^3$ be such that $\|\Psi_0\|=1$. Let $\Psi=(z_1,z_2,z_3)^T\in W^{1,1}(0,T;\C^3)$ be a solution to \eqref{eq:schrodinger}. Then
	\begin{enumerate}
		\item $\rho^\Psi, j^\Psi\in W^{1,1}(0,T;\R^2)$, with $\rho_1^\Psi,\rho_2^\Psi \ge 0$, $\rho_1^\Psi+\rho_2^\Psi=2$ and $j_1^\Psi+j_2^\Psi=0$;
		\item we have the \textup{modified continuity equation}
		\begin{equation}\label{eq:continuity}
			\dot \rho_1^\Psi + j_1^\Psi = -4w\lp \rho_1^\Psi\rho_2^\Psi-|z_3|^2\rp \qquad \text{a.e. on } (0,T),
		\end{equation}
		and this equation reduces to $\dot \rho_1^\Phi + j_1^\Phi = -2w\rho_1^\Phi\rho_2^\Phi$ when $\Psi(t)=\Phi(t)\in\cM^{\rm Slat}$ for all $t \in [0,T]$ (non-interacting dynamics with $\Psi_0 =\Phi_0\in \cM^{\rm Slat}$);
		\item if $w=0$, then $\rho^\Psi\in W^{2,1}(0,T;\R^2)\hookrightarrow C^1([0,T];\R^2)$, the continuity equation
		\begin{equation}\label{eq:continuity-w0}
			\dot\rho_1^\Psi+j^\Psi_1=0
		\end{equation}
		holds in $W^{1,1}(0,T;\R)$, and the \textup{van Leeuwen equation}
		\begin{equation}\label{eq:vanleeuwen}
			\ddot\rho^\Psi_1 + 4\tau^2(\rho^\Psi_1-1) = -2\sqrt2\,\tau\Big( \lp 2v + \lambda U_1 \rp \Re\lp \overline{z_3}\,z_1 \rp + \lp 2v - \lambda U_2 \rp \Re\lp \overline{z_3}\,z_2 \rp\Big)
		\end{equation}
		holds a.e. on $(0,T)$.
	\end{enumerate}
\end{proposition}
\begin{proof}
	The proof of \Cref{prop:reg_density} is given in \Cref{ssec:proof-reg-density}.
\end{proof}
\subsection{Bloch representation of non-interacting dynamics}\label{subsec:bloch}

In the simple case of the two-particle spin-singlet Hubbard dimer, the manifold of Slater states (i.e. $\cM^{\rm Slat}$ quotiented by the gauge invariance $\Phi \mapsto e^{i\vartheta}\Phi$) is isometric to the Bloch sphere $\bS^2$ (see \cref{appendix:derivation} for details) through the map
\begin{align*}
	X := \begin{pmatrix} x_1 \\ x_2 \\ x_3 \end{pmatrix} = \begin{pmatrix} \sin\theta\cos\phi \\ \sin\theta \sin\phi \\ \cos\theta \end{pmatrix} \quad & \mapsto \quad \varphi_X := \begin{pmatrix} \varphi_{X,1} \\ \varphi_{X,2} \end{pmatrix} := \begin{pmatrix} \cos\f\theta 2 \\ e^{i\phi} \sin\f\theta 2 \end{pmatrix} \\ & \mapsto \quad \Phi_X := \begin{pmatrix} \varphi_{X,1}^2 \\ \varphi_{X,2}^2 \\ \sqrt2\,\varphi_{X,1} \varphi_{X,2} \end{pmatrix} = \begin{pmatrix} \f{1+\cos\theta}2 \\ \f{1-\cos\theta}2 e^{2i\phi} \\ \f{\sin\theta}{\sqrt 2} e^{i\phi} \end{pmatrix} ,
\end{align*}
where $\varphi_X$ is the ``molecular orbital'' occupied by the two electrons. The density $\rho^X=(\rho^X_1,\rho^X_2)^T$ and current $j^X=(j^X_1,j^X_2)^T$ of a state $\Phi_X$ are therefore given by
\begin{align*}
	& \rho_1^X= 2 |\varphi_{X,1}^2|^2 + |\sqrt2\,\varphi_{X,1} \varphi_{X,2}|^2 = 2 |\varphi_{X,1}|^2 = 2 \cos^2 \f{\theta}2 = \cos\theta +1 = 1+x_3,\\
	& \rho_2^X = 2-\rho_1^X = 1-x_3,\\
	& j^X_1= 2\sqrt2\,\tau\,\Im\lp \sqrt2\,\varphi_{X,1}\varphi_{X,2}\lp \overline{\varphi_{X,1}}^2-\overline{\varphi_{X,2}}^2\rp\rp = 4\tau\,\varphi_{X,1}\,\Im\lp\varphi_{X,2}\rp \\
	& \phantom{j^X_1}= 4\tau\cos\f\theta2\sin\f\theta2\sin\phi = 2\tau \sin\theta\sin\phi = 2\tau x_2,\\
	& j^X_2 = -j^X_1.
\end{align*}
The North Pole $P_+:=(0,0,1)^T$ of the Bloch sphere thus corresponds to the state with both electrons on site~$1$, the South Pole $P_-:=(0,0,-1)^T$ to the state with both electrons on site~$2$, and the equator to states with equal occupations $\rho_1 = \rho_2= 1$ (see \cref{fig:bloch_sphere}).
\begin{figure}[h!]
	\centering
	\includegraphics[width=0.65\linewidth]{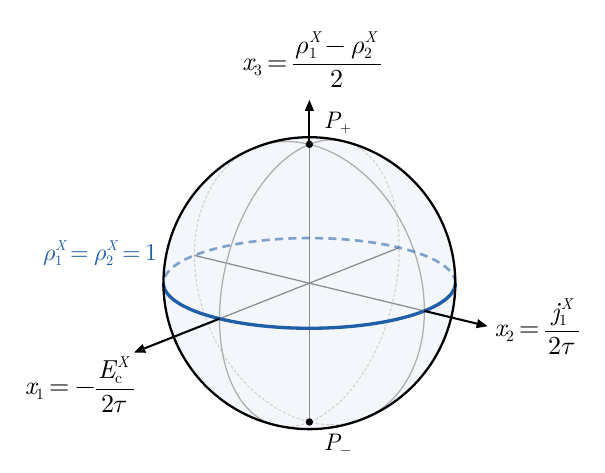}
	\caption{
		Bloch-sphere representation of the Slater states of the two-particle spin-singlet Hubbard model.
		The coordinates $x_1$, $x_2$, and $x_3$ are respectively related
		to the kinetic energy, the current, and the density imbalance.
		The north and South Poles correspond to complete localization of the two electrons on
		sites $1$ and $2$, respectively, while the equator corresponds to states such that
		$\rho_1=\rho_2=1$.
	}
	\label{fig:bloch_sphere}
\end{figure}

The Bloch representation is not only a convenient parametrization of the non-interacting dynamics: each component of the Bloch vector $X$ can be given a direct physical meaning:
\begin{align}
	x_1 &= - \f{E_{\rm c}^X}{2\tau} && \text{(negative of the normalized kinetic energy)}, \label{eq:interpretation_x1}\\
	x_2 &= \f{j_1^X}{2\tau} && \text{(normalized current from site $1$ to site $2$)}, \label{eq:interpretation_x2} \\
	x_3 &= \rho_1^X-1= \f{\rho_1^X-\rho^X_2}2 && \text{(normalized site density difference)}. \label{eq:interpretation_x3}
\end{align}
where $E_{\rm c}^X:=\langle \Phi_X | K | \Phi_X\rangle$ denotes the kinetic energy of the Slater state $\Phi_X$. This makes the physical meaning of a point on the Bloch sphere transparent.\medskip

In the Bloch representation, the non-interacting VP, GP, and GPc
dynamics take particularly simple forms. They are given by (see \cref{appendix:derivation} for details)
\begin{align}
	&\text{(VP)}  &&\dot X^{\rm VP} = \lp -2\tau e_1+2v\,e_3\rp \times X^{\rm VP}, \quad X^{\rm VP}(0) = X_0, \label{eq:vp_X}\\[2mm]
	&\text{(GP)}  &&\dot X^{\rm GP} = \lp -2\tau e_1 +2w\lp e_3\times X^{\rm GP}\rp \rp \times X^{\rm GP}, \quad X^{\rm GP}(0)  = X_0, \label{eq:gp_X}\\[2mm]
	&\text{(GPc)}  &&\dot X^{\rm GPc}_{v_c} = \lp -2\tau e_1+2v_c\,e_3 +2w_c\lp e_3\times X^{\rm GPc}_{v_c}\rp \rp \times X^{\rm GPc}_{v_c}, \;  X^{\rm GPc}_{v_c}(0) = X_0, \label{eq:gpc_X}
\end{align}
where $(e_1,e_2,e_3)$ denotes the canonical basis of the three-dimensional space $\R^3$ in which the unit sphere $\bS^2$ is embedded,  $X_0\in\bS^2$ the Bloch vector of the initial Slater state $\Phi_0$, i.e. $\Phi_0=\Phi_{X_0}$ up to an irrelevant global phase.
Note that the trajectories of the dynamics~\eqref{eq:vp_X},~\eqref{eq:gp_X}, and~\eqref{eq:gpc_X} with initial data $X_0\in\bS^2$ stay on the Bloch sphere $\bS^2$. In accordance with the functional setting of \cref{subsec:interacting}, we require these trajectories to belong to $W^{1,1}(0,T;\bS^2)$, the space of absolutely continuous functions on $[0,T]$ with values in $\bS^2$.

\subsection{Formulations of the TDDFT \texorpdfstring{$V$}{V}-representability problems}
\label{subsec:TDDFT-V}

The \textbf{non-interacting $V$-representability problem for general densities} for VP (resp. GP) can be formulated as follows: given a time-dependent density $\rho=(\rho_1,\rho_2)^T \in W^{1,1}(0,T;\R^2)$ such that $\rho_1(t), \rho_2(t) \ge 0$ and $\rho_1(t)+\rho_2(t)=2$ for all $t \in [0,T]$, and an initial condition $X_0 \in \bS^2$ such that $\rho^{X_0}=\rho(0)$, does there exist a function $v \in L^1(0,T;\R)$ (resp. a function $w \in L^0(0,T;\R)$, see \cref{subsec:interacting} for the choice of these spaces) such that the solution to~\eqref{eq:vp_X} (resp. a solution to~\eqref{eq:gp_X}) in $W^{1,1}(0,T;\bS^2)$ satisfies
\begin{equation*}
	\forall t\in[0,T], \quad  \rho^{X^{\rm VP}(t)}=\rho(t) \qquad \lp  \text{resp. }  \rho^{X^{\rm GP}(t)}=\rho(t)\rp  \; ?
\end{equation*}
For GPc, the formulation is similar, except that $v_c$ is given in $L^1(0,T;\R)$ and that we seek a function $w_c \in L^0(0,T;\R)$ such that $\rho^{X^{\rm GPc}_{v_c}(t)}=\rho(t)$ for all $t \in [0,T]$ where $X^{\rm GPc}_{v_c}$ is a solution to~\eqref{eq:gpc_X}. In view of~\eqref{eq:interpretation_x3}, the VP, GP, and GPc non-interacting $V$-representability problems can be reformulated as
\begin{align}
	&\text{(VP)}&& \text{find }v \in L^1(0,T;\R) \text{ such that the unique solution to \eqref{eq:vp_X} in $W^{1,1}(0,T;\bS^2)$}  \nonumber \\
	& && \qquad \text{satisfies }  1+X^{\rm VP}\cdot e_3 = \rho_1, \label{prob:VP} \\[2mm]
	&\text{(GP)}&& \text{find } w\in L^0(0,T;\R)  \text{ and } X^{\rm GP} \in W^{1,1}(0,T;\bS^2) \text{ satisfying \eqref{eq:gp_X} and } \nonumber \\
	& && \qquad  1+X^{\rm GP}\cdot e_3 = \rho_1, \nonumber\\[2mm]
	&\text{(GPc)} && \text{find } w_c\in L^0(0,T;\R)  \text{ and } X^{\rm GPc}_{v_c} \in W^{1,1}(0,T;\bS^2) \text{ satisfying \eqref{eq:gpc_X} and } \nonumber \\
	& && \qquad
		1+X^{\rm GPc}_{v_c}\cdot e_3 = \rho_1. \label{prob:GPc}
\end{align}

Given $X_0 \in \bS^2$ and $v_c \in L^1(0,T;\R)$, we define the set of VP/GP/GPc non-interacting $V$-representable densities arising from $X_0$ as, respectively,
\begin{align*}
	\cR_{\rm VP}(X_0) &:= \Big\{ \rho^{X^{\rm VP}} \,\Big|\, X^{\rm VP}\in W^{1,1}(0,T;\bS^2) \text{ solves~\eqref{eq:vp_X} with } v \in L^1 \Big\},\\
	\cR_{\rm GP}(X_0) &:= \Big\{ \rho^{X^{\rm GP}} \,\Big|\, X^{\rm GP}\in W^{1,1}(0,T;\bS^2) \text{ solves~\eqref{eq:gp_X} with } w \in L^0 \Big\},\\
	\cR_{\rm GPc}(X_0,v_c) &:= \Big\{ \rho^{X^{\rm GPc}_{v_c}} \,\Big|\, X^{\rm GPc}_{v_c}\in W^{1,1}(0,T;\bS^2) \text{ solves~\eqref{eq:gpc_X} with } w_c \in L^0 \Big\},
\end{align*}
where $L^1=L^1(0,T;\R)$ and $L^0=L^0(0,T;\R)$.
Note that
\begin{equation*}
	\cR_{\rm GP}(X_0)=\cR_{\rm GPc}(X_0,0).
\end{equation*}

Of particular physical interest is the \textbf{non-interacting $V$-representability problem for interacting densities}, in which the target density $\rho$ is the density $\rho^u$ generated by the interacting Hubbard dynamics~\eqref{eq:interacting_dynamics}. Given an initial condition $\Psi_0^{\rm int} \in \C^3$ such that $\|\Psi_0^{\rm int}\|=1$, we define the set of interacting $V$-representable densities arising from $\Psi_0^{\rm int}$ as
\begin{equation*}
	\cR_{\rm int}(\Psi_0^{\rm int}):= \left\{ \rho^{\Psi^u}, \; \Psi^u \in W^{1,1}(0,T;\C^3) \text{ sol. to } \eqref{eq:interacting_dynamics} \text{ for some } u \in L^1(0,T;\R)\right\}.
\end{equation*}
The question is then whether $\cR_{\rm int}(\Psi_0^{\rm int})\subset\cR_{\rm VP}(X_0)$ (resp. $\cR_{\rm int}(\Psi_0^{\rm int})\subset\cR_{\rm GPc}(X_0,v_c)$), and for which initial conditions $X_0\in\bS^2$.

\subsection{Main results}

The main practical result of this article is the following. It shows in particular that any interacting $V$-representable density is non-interacting $V$-representable for both VP and GPc, for any $v_c \in L^1(0,T;\R)$.

\begin{theorem}[Non-interacting $V$-representability of interacting $V$-representable densities] \label{thm:NIVIN} Let $\Psi_0^{\rm int} \in \C^3$ be such that $\|\Psi_0^{\rm int}\|=1$, $X_0 \in \bS^2$, and $v_c \in L^1(0,T;\R)$. Then:
\begin{enumerate}
	\item \emph{VP-representability.}
	\begin{equation*}
		\lp  \cR_{\rm int}(\Psi_0^{\rm int}) \subset \cR_{\rm VP}(X_0) \rp  \quad \Longleftrightarrow \quad \lp  \rho^{X_0}=\rho^{\Psi_0^{\rm int}}  \text{ and } j^{X_0} = j^{\Psi_0^{\rm int}}\rp .
	\end{equation*}
	\item \emph{GPc-representability.}
	\begin{equation*}
		\lp  \cR_{\rm int}(\Psi_0^{\rm int}) \subset \cR_{\rm GPc}(X_0,v_c) \rp  \quad \Longleftrightarrow \quad \lp \rho^{X_0}= \rho^{\Psi_0^{\rm int}}  \rp .
	\end{equation*}
	\item \emph{Compatible initial condition.} There exists $X_0' \in \bS^2$ such that $\rho^{X_0'}=\rho^{\Psi_0^{\rm int}}$ and $j^{X_0'} = j^{\Psi_0^{\rm int}}$.
\end{enumerate}
\end{theorem}
\begin{proof}
	The proof of \Cref{thm:NIVIN} is given in \Cref{sec:interacting}, as a consequence of \cref{thm:VP-characterization,thm:GP-characterization} and of the properties of interacting densities established there.
\end{proof}

\medskip

The other main results are detailed in \cref{thm:VP-characterization,thm:GP-characterization} below, providing complete solutions to the non-interacting VP/GP/GPc $V$-representability problem for general densities. More precisely, we give a complete characterization of the sets $\cR_{\rm VP}(X_0)$ and $\cR_{\rm GPc}(X_0,v_c)$ (we show in particular that $\cR_{\rm GPc}(X_0,v_c)=\cR_{\rm GP}(X_0)$ for all $v_c \in L^1(0,T;\R)$), and we specify
\begin{itemize}
	\item for all $\rho \in \cR_{\rm VP}(X_0)$, the (in general non-unique) potentials $v \in L^1(0,T;\R)$ for which~\eqref{prob:VP} holds;
	\item for all $v_c \in L^1(0,T;\R)$ and all $\rho \in \cR_{\rm GPc}(X_0,v_c)$, the (in general non-unique) potentials $w_c \in L^0(0,T;\R)$ \emph{and} trajectories $X^{\rm GPc}_{v_c} \in W^{1,1}(0,T;\bS^2)$ for which~\eqref{prob:GPc} holds.
\end{itemize}

In \cref{sec:comparison}, we highlight the main differences between the two principles, and we show that VP-representability implies GP-representability, whereas the converse is false, that is $\cR_{\rm VP}(X_0)\subsetneq\cR_{\rm GP}(X_0)$ (\cref{cor:VP-in-GP}). In \cref{sec:interacting}, we prove \cref{thm:NIVIN} and describe all the non-interacting trajectories and effective potentials reproducing an interacting density (\cref{thm:int-representability}). In \cref{sec:integrability}, we justify the functional setting for GP: the GP potentials reproducing an interacting density are in general not integrable, even for smooth external potentials.

\section{The VP-representability problem}
\label{sec:VP}

Let $X_0\in\bS^2$ and $\rho=(\rho_1,\rho_2)^T\in W^{1,1}(0,T;\R^2)$ be such that $\rho_1,\rho_2\ge0$ and $\rho_1+\rho_2=2$ on $[0,T]$. Assume that $\rho$ is VP-representable, i.e. that there exists $v\in L^1(0,T;\R)$ such that the solution $X^{\rm VP}\in W^{1,1}(0,T;\bS^2)$ to~\eqref{eq:vp_X} satisfies~\eqref{prob:VP}. In the canonical coordinate system, using the notation $X^{\rm VP}(t)=(x_1(t),x_2(t),x_3(t))^T$, the equations in problem~\eqref{prob:VP} read
\begin{equation}\label{eq:TDVP_coord}
	\dot x_1 = -2vx_2, \quad \dot x_2 = 2vx_1 + 2\tau x_3, \quad \dot x_3=-2\tau x_2, \quad x_j(0)=X_0 \cdot e_j, \quad x_3(t)=\rho_1(t) -1.
\end{equation}
As already noticed in point (iii) of \cref{prop:reg_density} (with $\lambda=0$ and $w=0$), the requirement $X^{\rm VP} \in W^{1,1}(0,T;\bS^2)$ implies that $\rho$ must be in $W^{2,1}(0,T;\R^2)$, since $\dot\rho_1=\dot x_3=-2\tau x_2\in W^{1,1}(0,T;\R)$. Together with the constraints $x_1(t)^2+x_2(t)^2+x_3(t)^2=1$ and $x_3(t)=\rho_1(t) -1$, this implies that
\begin{equation*}
	\cD_\rho:=4\tau^2\rho_1\rho_2 -\dot\rho_1^2 \in W^{1,1}(0,T;\R_+),
\end{equation*}
and
\begin{equation*}
	x_1= \sgn(x_1)\f{\sqrt{\cD_\rho}}{2\tau}, \qquad x_2=-\f{\dot\rho_1}{2\tau}, \qquad x_3=\rho_1-1,
\end{equation*}
with the convention that $\sgn(0)=0$.
Using again the requirement $X^{\rm VP} \in W^{1,1}(0,T;\bS^2)$ as well as the initial condition, we obtain that $\rho$ and $X_0$ must also be such that $\sqrt{\cD_\rho}\in W^{1,1}(0,T;\R_+)$, $x_2(0)=-\dot \rho_1(0)/2\tau$ and $x_3(0)=\rho_1(0)-1$ (the normalization condition then fixes $|x_1(0)|$, and its sign is fixed by the choice of $\eta_1$ below). The potential $v$ can then be formally obtained from the second equation in~\eqref{eq:TDVP_coord} and reads
\begin{equation}
	\label{eq:vVP}
	v(t) = - \sgn(x_1) \f{\cN_\rho(t)}{2\sqrt{\cD_\rho(t)}} \quad \text{where} \quad \cN_\rho:=\ddot\rho_1+4\tau^2(\rho_1-1).
\end{equation}
In view of the VP-representability conditions $\rho \in W^{2,1}(0,T;\R^2)$ and $\sqrt{\cD_\rho} \in W^{1,1}(0,T;\R_+)$ obtained above, the potential $v$ is in $L^1$ on any interval on which the non-negative continuous function $\cD_\rho$ is bounded away from $0$. As a consequence, the singular times for the VP-representability problem are the times $t_*$ at which $\cD_\rho(t_*)=0$. Interestingly, the set of singular times
\begin{equation*}
	\cS_\rho := \left\{ t_*\in[0,T] \;\middle|\; \cD_\rho(t_*)=0 \right\} \qquad \text{(VP-representability singular set)}
\end{equation*}
has a very simple interpretation in terms of the VP dynamics. Indeed, using~\eqref{eq:TDVP_coord}, we get
\begin{equation}
	\label{eq:D-Bloch}
	\cD_\rho = 4\tau^2\rho_1\rho_2 -\dot\rho_1^2 = 4\tau^2 (1+x_3)(1-x_3)-4\tau^2x_2^2 = 4\tau^2 (1-x_3^2-x_2^2)=4\tau^2 x_1^2.
\end{equation}
The singular times for VP thus are the times $t_*$ at which the VP-trajectory $X^{\rm VP}$ reaches the great circle
\begin{equation*}
	\cC_s:=\left\{X=(x_1,x_2,x_3)^T \in\bS^2 \;\middle|\; x_1=0\right\},
\end{equation*}
which is, by \eqref{eq:interpretation_x1}, the set of Slater states with zero kinetic energy. Note that a VP-trajectory may remain on $\cC_s$ on a time-interval of positive length, so that $\cS_\rho$ may have a positive Lebesgue measure (see \cref{rem:Cs}). We will see in \cref{sec:interacting} that this cannot happen for interacting densities.

\begin{figure}[t]
	\centering
	\includegraphics[width=0.65\linewidth]{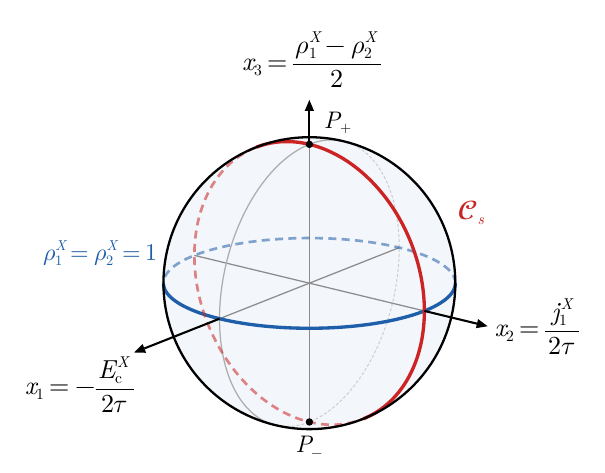}
	\caption{
		Bloch-sphere representation of the singular states: the red great circle $\cC_s=\{x_1=0\}$, characterized by a vanishing kinetic energy, is the set of singular states for the VP dynamics. The two poles $P_\pm$, which belong to $\cC_s$, are the singular states for the GP and GPc dynamics.
	}
	\label{fig:singular_set}
\end{figure}

\medskip

Let us introduce the set
\begin{equation*}
	\Omega_\rho:=\left\{t\in[0,T] \;\middle|\; \cD_\rho(t) \neq 0 \right\} =[0,T]\setminus \cS_\rho \qquad \text{(VP-representability regular set)}.
\end{equation*}
Since $\cD_\rho$ must be continuous, $\Omega_\rho$ is relatively open in $[0,T]$ and is therefore the disjoint union of at most countably many connected components (relatively open intervals of $[0,T]$) $(I_k)_{1\le k\le M_{\rm cc}}$ with $M_{\rm cc} \in \N \cup \{+\infty\}$, indexed so that $0\in I_1$ if $0\in\Omega_\rho$ and with the convention that $M_{\rm cc}=0$ if $\Omega_\rho=\emptyset$. Note that these intervals cannot be ordered (with respect to the order of the real line) in general, since $\cS_\rho$ may be a Cantor-like set.
We thus have
\begin{equation*}
	\Omega_\rho = \bigcup_{k = 1}^{M_{\rm cc}} I_k.
\end{equation*}
On each $I_k$, the sign of $x_1$ is constant and can be denoted by $\eta_k$. If $\sgn(x_1(0))\neq 0$, then $\eta_1 = \sgn(x_1(0))$, otherwise $\eta_1\in\{-1,1\}$ is not fixed by the initial condition. In addition, the condition $v \in L^1(0,T;\R)$ implies, together with~\eqref{eq:vVP}, that $\1_{\Omega_\rho} \f{\cN_\rho}{\sqrt{\cD_\rho}} \in L^1(0,T;\R)$, while the second and third equations in \eqref{eq:TDVP_coord} give $\cN_\rho=-4\tau vx_1$ a.e. on $(0,T)$, so that $\cN_\rho=0$ a.e. on $\cS_\rho$, where $x_1$ vanishes. Finally, the potential $v$ is uniquely determined a.e. by the trajectory: on $\Omega_\rho$ by~\eqref{eq:vVP}, and on $\cS_\rho$, since $x_1$ vanishes on $\cS_\rho$, so that $\dot x_1=0$ a.e. on $\cS_\rho$, and the first equation in~\eqref{eq:TDVP_coord} gives $vx_2=0$ a.e. on $\cS_\rho$. As $x_2^2=\rho_1\rho_2$ on $\cS_\rho$, we get $v=0$ a.e. on $\cS_\rho\setminus\cT_\rho$, where $\cT_\rho:=\{\rho_1\rho_2=0\}$. Besides, $|\cT_\rho|=0$ by \cref{lem:poles}. Hence $v=-\f12\,\sgn(x_1)\1_{\Omega_\rho}\cN_\rho/\sqrt{\cD_\rho}$ a.e. on $(0,T)$. Conversely, these computations can be reversed, which yields the sufficiency of the conditions below (see \cref{ssec:proof-VP}).

\medskip

We are now in position to state the main result of this section.

\begin{theorem}[The VP-representability problem]
	\label{thm:VP-characterization}
	Let $X_0\in\bS^2$.
	\begin{enumerate}
		\item \emph{Characterization of the set of VP-representable densities.}
	We have
	\begin{equation}
		\label{eq:RVP-characterization}
			\rho\in\cR_{\rm VP}(X_0) \iff \left\{
			\begin{aligned}
				&\rho\in W^{2,1}(0,T;\R^2), \\ & \rho_1+\rho_2=2 \text{ and } \rho_1,\rho_2\geq0  \text{ on } [0,T], \\
				&\rho_1(0)=1+X_0\cdot e_3, \quad \dot\rho_1(0)=-2\tau X_0\cdot e_2, \\
			& \cD_\rho:=4\tau^2\rho_1\rho_2 -\dot\rho_1^2 \geq 0 \text{ on } [0,T], \\
			& \cN_\rho:=\ddot\rho_1+4\tau^2(\rho_1-1)=0 \ \text{a.e. on } \cS_\rho:= \cD_\rho^{-1}(\{0\}), \\
				&\displaystyle \1_{\{\cD_\rho>0\}} \f{\cN_\rho}{\sqrt{\cD_\rho}} \in L^1(0,T;\R).
			\end{aligned}
			\right.
	\end{equation}
	\item \emph{Classification of admissible VP-trajectories and potentials.}
	Let $\rho \in \cR_{\rm VP}(X_0)$.  Consider the decomposition of the relatively open subset $\Omega_\rho := \cD_\rho^{-1}((0,+\infty)) = [0,T] \setminus \cS_\rho$ into the at most countable disjoint union of relatively open intervals
	\begin{equation*}
		\Omega_\rho = \bigcup_{k=1}^{M_{\rm cc}} I_k, \qquad M_{\rm cc} \in \N \cup \{+\infty\},
	\end{equation*}
		with $0\in I_1$ if $0\in\Omega_\rho$ (i.e. if $X_0\cdot e_1\neq0$). Then, the $W^{1,1}(0,T;\bS^2)$ VP-trajectories
	reproducing the density $\rho$ are in one-to-one correspondence with the set of parameters
	\begin{equation*}
			\fS_\rho(X_0) = \left\{ \eta =(\eta_k)_{1\le k\le M_{\rm cc}} \in \{-1,1\}^{M_{\rm cc}} \;\middle|\; \eta_1 = \sgn(X_0\cdot e_1) \text{ if } X_0\cdot e_1 \neq 0  \right\},
	\end{equation*}
	through
	\begin{equation}
		\label{eq:VP-fiber-X}
		X^{\rm VP}_\eta(t) =
		\begin{pmatrix}
			\displaystyle \sigma_\eta(t)\f{\sqrt{\cD_\rho(t)}}{2\tau} \\[1.5mm]
			\displaystyle -\f{\dot\rho_1(t)}{2\tau} \\[1.5mm]
			\rho_1(t)-1
		\end{pmatrix}
		\quad \text{with} \quad
		\sigma_\eta(t) =
		\begin{cases}
			\eta_k & \text{if } t\in I_k,\ 1\le k\le M_{\rm cc}, \\
			0 & \text{if } t\in\cS_\rho,
		\end{cases}
	\end{equation}
	and the effective potential generating $X^{\rm VP}_\eta$ is given by
	\begin{equation}
		\label{eq:VP-fiber}
		v_\eta = -\f12\,  \sigma_\eta \1_{\{\cD_\rho > 0\}}\f{\cN_\rho}{\sqrt{\cD_\rho}}.
	\end{equation}
		\end{enumerate}
\end{theorem}
\begin{proof}
	The proof of \Cref{thm:VP-characterization} is given in \Cref{ssec:proof-VP}.
\end{proof}
\begin{remark}
	The condition $\sqrt{\cD_\rho}\in W^{1,1}(0,T;\R_+)$ derived above does not appear in~\eqref{eq:RVP-characterization}, since it follows from $\rho\in W^{2,1}(0,T;\R^2)$, $\cD_\rho\ge0$ and $\1_{\{\cD_\rho>0\}}\cN_\rho/\sqrt{\cD_\rho}\in L^1(0,T;\R)$ by \cref{lem:sqrtD}.
\end{remark}

\section{The GP and GPc-representability problems}
\label{sec:GP}

Let us now turn to the $V$-representability problem for GP and GPc. Since GP is the particular case $v_c=0$ of GPc, we directly treat the GPc-representability problem \eqref{prob:GPc}. Let $X_0\in\bS^2$, $v_c\in L^1(0,T;\R)$ and $\rho=(\rho_1,\rho_2)^T\in W^{1,1}(0,T;\R^2)$ be such that $\rho_1,\rho_2\ge0$ and $\rho_1+\rho_2=2$ on $[0,T]$. Assume that $\rho$ is GPc-representable, i.e. that there exist $w_c\in L^0(0,T;\R)$ and $X^{\rm GPc}_{v_c}\in W^{1,1}(0,T;\bS^2)$ satisfying~\eqref{prob:GPc}.

In the canonical coordinate system, using the notation $X^{\rm GPc}_{v_c}(t)=(y_1(t),y_2(t),y_3(t))^T$, the equations in problem~\eqref{prob:GPc} read
\begin{align}
	\label{eq:TDGP_coord}
	&\dot y_1 = -2v_cy_2+2w_cy_1y_3, \quad \dot y_2 = 2v_cy_1+2\tau y_3 + 2w_cy_2y_3, \quad \dot y_3=-2\tau y_2-2w_c(y_1^2+y_2^2), \\
	& y_j(0)=X_0 \cdot e_j, \quad y_3=\rho_1-1. \nonumber
\end{align}
The constraints $y_1^2+y_2^2+y_3^2=1$ and $y_3=\rho_1-1$ imply
\begin{equation*}
	y_1^2+y_2^2=\rho_1\rho_2, \qquad y_3=\rho_1-1.
\end{equation*}
In contrast with the VP case, the constraint on the density only fixes the colatitude of the GPc-trajectory on the Bloch sphere, and not its longitude. Let us introduce the function $r_\rho \in C^0([0,T];\R_+)$ defined by
\begin{equation*}
	r_\rho:=\sqrt{\rho_1\rho_2},
\end{equation*}
and the sets
\begin{align*}
	\cT_\rho&:=\left\{t_* \in [0,T] \;\middle|\; r_\rho(t_*)=0\right\} && \text{(GP-representability singular set)}, \\
	\widetilde\Omega_\rho&:=[0,T]\setminus\cT_\rho && \text{(GP-representability regular set)}.
\end{align*}
Note that the set $\cT_\rho$ has a very simple geometric interpretation: it is the set of times at which the non-interacting trajectory $X^{\rm GPc}_{v_c}$ reaches one of the poles of the Bloch sphere. The set $\cT_\rho$ is closed since $r_\rho$ is continuous, and an important technical result (\cref{lem:poles}) shows that it has zero Lebesgue measure if $\rho \in \cR_{\rm GPc}(X_0,v_c)$. Hence $\widetilde\Omega_\rho$ is relatively open in $[0,T]$, and is therefore the disjoint union of at most countably many connected components (relatively open intervals of $[0,T]$) $(J_\ell)_{1\le \ell\le N_{\rm cc}}$ with $N_{\rm cc}\in\N^*\cup\{+\infty\}$,
\begin{equation*}
	\widetilde\Omega_\rho = \bigcup_{\ell=1}^{N_{\rm cc}} J_\ell,
\end{equation*}
where the $J_\ell$ are labeled in such a way that $0\in J_1$ if $0\in\widetilde\Omega_\rho$. As in \cref{sec:VP}, these intervals cannot be ordered in general, since $\cT_\rho$ may be a Cantor-like set. On $\widetilde\Omega_\rho$, the function $r_\rho$ is continuous and positive, so that the continuous function $(y_1,y_2)$ admits a continuous lift $\alpha \in W^{1,1}_{\rm loc}(\widetilde\Omega_\rho;\R/(2\pi\Z))$ such that
\begin{equation}
	\label{eq:GP-alpha-reconstruction}
	y_1=-r_\rho\sin\alpha, \qquad y_2=r_\rho\cos\alpha, \qquad y_3=\rho_1-1 \qquad\text{on } \widetilde\Omega_\rho .
\end{equation}

\begin{figure}[t]
	\centering
	\includegraphics[width=0.65\linewidth]{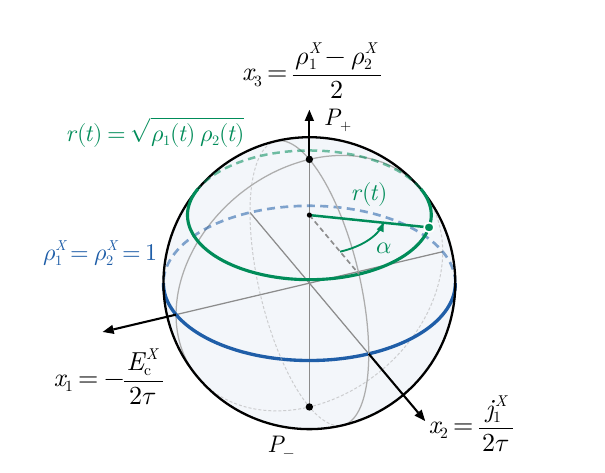}
	\caption{
		Geometric picture of the GPc reconstruction. The density constraint $y_3=\rho_1(t)-1$ fixes the latitude of $X^{\rm GPc}_{v_c}(t)$: the trajectory lies on the horizontal circle (green) of radius $r_\rho(t)=\sqrt{\rho_1(t)\rho_2(t)}$, and only its position on this circle, given by the angle $\alpha$ of \eqref{eq:GP-alpha-reconstruction}, is left free. The circle degenerates into the poles $P_\pm$, where $\rho_1\rho_2=0$ and $\alpha$
		is no longer defined.
	}
	\label{fig:GP_latitude}
\end{figure}
The angle $\alpha$ (see \cref{fig:GP_latitude}) is related to the longitude $\phi$ of $X^{\rm GPc}_{v_c}$ by $\alpha=\phi-\f\pi2$. Using \eqref{eq:TDGP_coord}, we get
\begin{equation*}
	r_\rho^2\,\dot\alpha = y_1\dot y_2-y_2\dot y_1 = 2v_c\,(y_1^2+y_2^2)+2\tau y_1y_3 = 2v_c\,r_\rho^2-2\tau(\rho_1-1)\,r_\rho\sin\alpha \qquad\text{on } \widetilde\Omega_\rho.
\end{equation*}
We infer that $\alpha$ is solution in $W^{1,1}_{\rm loc}(\widetilde\Omega_\rho;\R/(2\pi\Z))$ to the equation
\begin{equation}
	\label{eq:dynamics_alpha}
	\dot\alpha = - 2\tau\f{\rho_1-1}{r_\rho}\sin\alpha + 2v_c,
\end{equation}
which, remarkably, does not involve the geometric control $w_c$.
On each $J_\ell$, the coefficient $(\rho_1-1)/r_\rho$ is continuous, hence locally integrable, so that the right-hand side of \eqref{eq:dynamics_alpha} is Lipschitz in $\alpha$ at all times $t \in J_\ell$, with a locally integrable Lipschitz constant. By  Carath\'eodory's existence and uniqueness theorem, \eqref{eq:dynamics_alpha} admits a unique solution $\alpha\in W^{1,1}_{\rm loc}(\widetilde\Omega_\rho;\R/(2\pi\Z))$ once its value $\alpha(s_\ell)=a_\ell$ at one point $s_\ell\in J_\ell$ for each $1 \le \ell \le N_{\rm cc}$ is prescribed. If $0\in\widetilde\Omega_\rho$, we take $s_1=0$, and the value $a_1=a_0$ is fixed by the initial condition $X_0$ through
\begin{equation}
	\label{eq:GP-alpha-initial}
	X_0\cdot e_1 = -r_\rho(0)\sin a_0 \quad \text{and} \quad X_0\cdot e_2 = r_\rho(0)\cos a_0 \quad \text{if } 0 \in \widetilde\Omega_\rho,
\end{equation}
whereas on every other connected component $J_\ell$, $\ell \ge 2$, of $\widetilde\Omega_\rho$, the value $a_\ell$ can be chosen arbitrarily in $\R/(2\pi\Z)$. If $0\in\cT_\rho$, that is if $X_0\in\{P_+,P_-\}$, no angular value is fixed by the initial condition.

\medskip

Since $y_1=y_2=0$ and $y_3=\pm1$ at the poles, any GPc-trajectory reproducing $\rho$ is thus of the form
\begin{equation}
	\label{eq:GP-full-trajectory}
	X^{\rm GPc}_{v_c}(t)=
	\begin{cases}
		\lp -r_\rho(t)\sin\alpha(t),\ r_\rho(t)\cos\alpha(t),\ \rho_1(t)-1\rp^T & \text{if } t\in J_\ell,\ 1\le \ell\le N_{\rm cc}, \\[1mm]
		\lp 0,\ 0,\ \rho_1(t)-1\rp^T & \text{if } t\in\cT_\rho,
	\end{cases}
\end{equation}
where $\alpha$ solves \eqref{eq:dynamics_alpha} and is determined by its value at one time $s_\ell\in J_\ell$ for each $1 \le \ell \le N_{\rm cc}$, this value being fixed by $X_0$ on $J_1$ if $0\in\widetilde\Omega_\rho$, and free otherwise. The two expressions match at the endpoints of the $J_\ell$ belonging to $\cT_\rho$, so that any function defined by \eqref{eq:GP-full-trajectory} is continuous on $[0,T]$.
The third equation in \eqref{eq:TDGP_coord} is the continuity equation
\begin{equation*}
	\dot\rho_1 = -2\tau y_2-2w_c\rho_1\rho_2.
\end{equation*}
Using \eqref{eq:GP-alpha-reconstruction}, we get
\begin{equation}
	\label{eq:formula_w}
	w_c(t) = -\f{\dot\rho_1(t) +2\tau r_\rho(t)\cos\alpha(t)} {2r_\rho(t)^2} \quad \text{on } \widetilde\Omega_\rho.
\end{equation}
Since the set $\cT_\rho$ is negligible, the above expression completely characterizes $w_c$ in $L^0(0,T;\R)$, for given values of the free parameters $a_\ell$. It remains to find necessary and sufficient conditions on $\rho$ and the free parameters $a_\ell$ for the pair $(w_c,X^{\rm GPc}_{v_c})$ defined by~\eqref{eq:dynamics_alpha}, \eqref{eq:GP-alpha-initial}, \eqref{eq:GP-full-trajectory} and \eqref{eq:formula_w} to be a solution to~\eqref{eq:TDGP_coord} in $L^0(0,T;\R) \times W^{1,1}(0,T;\bS^2)$. Assuming $|\cT_\rho|=0$, these conditions read $w_cy_1y_3, w_cy_2y_3, w_c(y_1^2+y_2^2) \in L^1(0,T;\R)$ (the absolute continuity of $X^{\rm GPc}_{v_c}$ across the poles then follows, see \cref{ssec:proof-GP}), that is, since $y_1^2+y_2^2+y_3^2=1$, $w_c r_\rho \in L^1(0,T;\R)$, or, using~\eqref{eq:formula_w},
\begin{equation*}
	\1_{\widetilde\Omega_\rho} \f{\dot\rho_1}{\sqrt{\rho_1\rho_2}} \in L^1(0,T;\R).
\end{equation*}
Conversely, these computations can be reversed, which yields the sufficiency of the conditions below (see \cref{ssec:proof-GP}).

The above analysis gives rise to the following theorem.

\begin{theorem}[The GPc-representability problem] \label{thm:GP-characterization}
Let $X_0\in\bS^2$ and $v_c\in L^1(0,T;\R)$ (the GP-representability problem corresponds to $v_c=0$).
\begin{enumerate}
	\item \emph{Characterization of the set of GPc-representable densities.} We have
	\begin{equation}
		\label{eq:RGP-characterization}
			\thickmuskip=4mu \medmuskip=3mu
			\rho\in\cR_{\rm GPc}(X_0,v_c) \iff \left\{
			\begin{aligned}
				&\rho\in W^{1,1}(0,T;\R^2), \\
				&\rho_1+\rho_2=2 \text{ and } \rho_1,\rho_2\geq0 \text{ on } [0,T], \\
				&\rho_1(0)=1+X_0\cdot e_3, \\
				&|\cT_\rho|=0 \text{ where } \cT_\rho:=\{t_* \in [0,T] \mid \rho_1(t_*)\rho_2(t_*)=0\}, \\
				&\displaystyle \1_{\{\rho_1\rho_2>0\}} \f{\dot\rho_1}{\sqrt{\rho_1\rho_2}} \in L^1(0,T;\R).
			\end{aligned}
			\right.
	\end{equation}
	In particular,  the set $\cR_{\rm GPc}(X_0,v_c)$ is independent of $v_c$ and is equal to the set of GP-representable densities $\cR_{\rm GP}(X_0)$.
\item \emph{Classification of admissible GPc-trajectories and potentials.} Let $\rho \in \cR_{\rm GPc}(X_0,v_c)=\cR_{\rm GP}(X_0)$, and consider the decomposition
	\begin{equation*}
		\widetilde\Omega_\rho:= \left\{t \in [0,T] \;\middle|\; \rho_1(t)\rho_2(t)>0\right\}= \bigcup_{\ell=1}^{N_{\rm cc}} J_\ell, \qquad N_{\rm cc} \in \N^* \cup \{+\infty\}
	\end{equation*}
	of the relatively open subset $\widetilde\Omega_\rho$ of $[0,T]$ into its connected components, with $0\in J_1$ if $0\in\widetilde\Omega_\rho$ (i.e. if $X_0\notin\{P_+,P_-\}$), and set $r_\rho:=\sqrt{\rho_1\rho_2}$. Fix $s_\ell\in J_\ell$ for every $1\le \ell\le N_{\rm cc}$, with $s_1=0$ if $X_0\notin\{P_+,P_-\}$, and, in that case, let $a_0\in\R/(2\pi\Z)$ be the unique angle defined by \eqref{eq:GP-alpha-initial}. Then the GPc-trajectories reproducing $\rho$ are in one-to-one correspondence with the set of parameters
	\begin{equation*}
			\cA_\rho(X_0) = \left\{a=(a_\ell)_{1 \le \ell \le N_{\rm cc}} \in \lp\R/(2\pi\Z)\rp^{N_{\rm cc}} \;\middle|\; a_1 = a_0 \text{ if } X_0\notin\{P_+,P_-\} \right\},
	\end{equation*}
	through
	\begin{equation}
		\label{eq:GP-fiber-X}
		X^{\rm GPc}_{v_c,a}(t) =
		\begin{pmatrix}
			\displaystyle -\1_{\{r_\rho>0\}}r_\rho(t)\,\sin\alpha_a(t) \\[1.5mm]
			\displaystyle \1_{\{r_\rho>0\}}r_\rho(t)\,\cos\alpha_a(t) \\[1.5mm]
			\rho_1(t)-1
		\end{pmatrix},
	\end{equation}
	where $\alpha_a$ is the unique solution of
	\begin{equation}\label{eq:angle_gp}
		\dot\alpha_a = - 2\tau\f{\rho_1-1}{r_\rho}\sin\alpha_a + 2v_c \quad \text{on } J_\ell, \qquad \alpha_a(s_\ell)=a_\ell, \qquad 1\le \ell\le N_{\rm cc}.
	\end{equation}
	Moreover, the trajectory $X^{\rm GPc}_{v_c,a}$ is solution to \eqref{eq:gpc_X} with potential $w_c$ given by
	\begin{equation}
		\label{eq:GP-fiber}
		w_{c,a} = -\1_{\{r_\rho>0\}}\f{\dot\rho_1 +2\tau r_\rho\cos\alpha_a} {2r_\rho^2}.
	\end{equation}
\end{enumerate}
\end{theorem}
\begin{proof}
	The proof of \Cref{thm:GP-characterization} is given in \Cref{ssec:proof-GP}.
\end{proof}

\begin{remark}
	Whereas the map $a \mapsto X^{\rm GPc}_{v_c,a}$ is one-to-one, the map $a\mapsto w_{c,a}$ is not in general: two sets of parameters $a\neq a'$ in $\cA_\rho(X_0)$ give the same control $w_{c,a}=w_{c,a'}$ if and only if, for every $1 \le \ell \le N_{\rm cc}$ such that $a_\ell\neq a'_\ell$, one has $a'_\ell=-a_\ell$ and $v_c=0$ a.e. on $J_\ell$. Indeed, $w_{c,a}=w_{c,a'}$ a.e. on $J_\ell$ means $\cos\alpha_a=\cos\alpha_{a'}$ on $J_\ell$. If $\alpha_{a'}(t_0)=\alpha_a(t_0)$ for some $t_0\in J_\ell$, then $\alpha_{a'}=\alpha_a$ on $J_\ell$ by uniqueness of the solution to~\eqref{eq:angle_gp}, and $a'_\ell=a_\ell$. Otherwise, $\alpha_{a'}=-\alpha_a$ on $J_\ell$, and $a'_\ell=-a_\ell$. In the latter case, subtracting the two equations defining $\alpha_a$ and $\alpha_{a'}$ respectively on $J_\ell$ gives $v_c=0$ a.e. on $J_\ell$. In particular, if $X^{\rm GPc}_{0,a}$ is an admissible GP-trajectory (case $v_c=0$), then a GP-trajectory $X^{\rm GP}$ corresponds to the same potential $w_{c,a}$ if and only if on each connected component $J_\ell$ of $\widetilde \Omega_\rho$, either $X^{\rm GP}$ coincides with $X^{\rm GPc}_{0,a}$ on $J_\ell$, or is its mirror image with respect to the vertical plane $\{x_1=0\}$ containing $\cC_s$.
\end{remark}

\section{Connections between VP and GP/GPc-representability results}
\label{sec:comparison}

In this section, we compare the non-interacting $V$-representability results obtained in \cref{thm:VP-characterization,thm:GP-characterization}. We show in particular that VP-representability implies GPc-representability.
This statement can be formalized as follows.

\begin{corollary}
	\label{cor:VP-in-GP}
	Let $X_0\in \bS^2$. Then, for every $v_c\in L^1(0,T;\R)$,
	\begin{equation*}
		\cR_{\rm VP}(X_0) \subsetneq \cR_{\rm GP}(X_0)=\cR_{\rm GPc}(X_0,v_c).
	\end{equation*}
	Moreover, for all $\rho\in\cR_{\rm VP}(X_0)$ we have
	\begin{equation}
	\label{eq:speed-limit}
	|\dot\rho_1|\le 2\tau\sqrt{\rho_1\rho_2} \text{ on } [0,T] \quad \text{and} \quad \cT_\rho\subseteq\cS_\rho.
\end{equation}
\end{corollary}

\begin{proof}
	Let $\rho\in\cR_{\rm VP}(X_0)$. Then $\rho\in W^{2,1}(0,T;\R^2)\subset W^{1,1}(0,T;\R^2)$ and $\rho_1(0)=1+X_0\cdot e_3$. The condition $\cD_\rho\geq0$ reads
\begin{equation*}
	|\dot\rho_1|\le 2\tau\sqrt{\rho_1\rho_2} \qquad\text{on }[0,T],
\end{equation*}
(which proves the inequality in~\eqref{eq:speed-limit}) so that $\1_{\{\rho_1\rho_2>0\}}\dot\rho_1/\sqrt{\rho_1\rho_2}$ is bounded by $2\tau$, hence integrable. It also gives $\dot\rho_1=0$ on $\cT_\rho$, which readily implies that $\cT_\rho\subseteq\cS_\rho$. Finally, $\rho=\rho^X$ for some VP-trajectory $X$, and therefore $|\cT_\rho|=0$ by \cref{lem:poles}. By \cref{thm:GP-characterization}, $\rho\in\cR_{\rm GP}(X_0)=\cR_{\rm GPc}(X_0,v_c)$ for every $v_c\in L^1(0,T;\R)$. To show that the inclusion is strict, it suffices to exhibit a density $\rho \in \cR_{\rm GP}(X_0)$ violating the first condition in~\eqref{eq:speed-limit}. For instance, with $\theta_0\in[0,\pi]$ such that $\cos\theta_0=X_0\cdot e_3$ and $\kappa>2\tau$, the density $\rho_1(t):=1+\cos(\kappa t+\theta_0)$ satisfies $\rho_1\rho_2=\sin^2(\kappa t+\theta_0)$ and $|\dot\rho_1|/\sqrt{\rho_1\rho_2}=\kappa$ on $\{\rho_1\rho_2>0\}$, so that it belongs to $\cR_{\rm GP}(X_0)$ by \cref{thm:GP-characterization}, but $|\dot\rho_1|>2\tau\sqrt{\rho_1\rho_2}$ on this set.
\end{proof}

The bound \eqref{eq:speed-limit} has a transparent physical meaning: on a lattice, the density cannot vary arbitrarily fast, because the current carrying the variation is itself bounded by the hopping amplitude. In the GP and GPc dynamics, the geometric control $w_c$ induces a source term in the continuity equation \eqref{eq:continuity}, which removes this restriction.

\medskip

\noindent
In the rest of the section, we consider some $\rho \in \cR_{\rm VP}(X_0) \subset \cR_{\rm GP}(X_0)$ and some $v_c \in L^1(0,T;\R)$.

\medskip

First, we see that on the Bloch sphere, the VP singular states form the great circle $\cC_s=\{x_1=0\}$ of states with zero kinetic energy, whereas the GPc singular states are the two poles $P_\pm$, which lie on $\cC_s$ (see \cref{fig:singular_set}). Geometrically, the inclusion $\cT_\rho\subseteq\cS_\rho$ reflects the fact that $\{P_+,P_-\}\subset\cC_s$. It implies that if there exists a non-singular VP-trajectory starting from the initial condition $X_0$ reproducing $\rho$ on the whole time-interval $[0,T]$, then for \emph{any} $X_0'$ satisfying $1+X_0' \cdot e_3=\rho_1(0)$, there exists a unique GP-trajectory starting from $X_0'$ reproducing $\rho$, and this trajectory does not meet any singularity.

\medskip

Consider now the case when the VP and/or GP-trajectories have singularities.
In both cases, the constraint on the density fixes the colatitude of the trajectory, $x_3=y_3=\rho_1-1$, and the admissible VP- (resp.\ GPc-)trajectories are parametrized by one free parameter per connected component of the non-singular set, the component containing $t=0$ excluded when $X_0$ is not in the singular set. The parameter, however, differs in nature:
\begin{itemize}
	\item in VP, $x_2$ is fixed by the equation $x_2=-\dot \rho_1/(2\tau)$ and only the sign of $x_1$ is free. The set of admissible VP-trajectories is discrete, parametrized by $\eta\in\fS_\rho(X_0)\subset\{-1,1\}^{M_{\rm cc}}$, and the branching occurs when the trajectory hits $\cC_s$;
	\item in GPc, only the colatitude is fixed, and the longitude $\phi=\alpha+\f \pi 2$ is a solution of the first-order equation \eqref{eq:dynamics_alpha}. The set of admissible GPc-trajectories is a continuum, parametrized by $a\in\cA_\rho(X_0)\subset(\R/(2\pi\Z))^{N_{\rm cc}}$, and the branching occurs when the trajectory reaches a pole.
\end{itemize}

The differences between VP and GP/GPc in terms of non-interacting $V$-representability conditions on the one hand, and of admissible non-interacting trajectories and potentials on the other hand, are summarized in \cref{tab:VP-vs-GP} and \cref{tab:VP-vs-GP-2}, respectively.

\begin{table}[t]
	\centering
	\small
	\setlength{\tabcolsep}{4pt}
	\renewcommand{\arraystretch}{1.6}
	\begin{tabular}{l|c|c}
		& VP & GP / GPc \\
		\hline
		Required regularity of $\rho$ & $W^{2,1}(0,T;\R^2)$ & $W^{1,1}(0,T;\R^2)$ \\
		$N$-representability conditions & $\rho_1,\rho_2 \ge 0$, \ $\rho_1+\rho_2=2$ & $\rho_1,\rho_2 \ge 0$, \ $\rho_1+\rho_2=2$ \\
		Pointwise constraint on $\dot\rho$ & {\renewcommand{\arraystretch}{1.1}\begin{tabular}[c]{@{}c@{}}\rule{0pt}{2.6ex}$\cD_\rho\ge 0$, i.e.\\ $|\dot\rho_1|\le2\tau\sqrt{\rho_1\rho_2}$\rule[-1.3ex]{0pt}{0pt}\end{tabular}} & none \\
		Set of singular times & $\cS_\rho=\{\cD_\rho=0\}$ & $\cT_\rho=\{\rho_1\rho_2=0\}$ \\
		Condition on $\cS_\rho$/$\cT_\rho$ & $\cN_\rho=0$ a.e. on $\cS_\rho$ & $|\cT_\rho|=0$ \\
		Integrability condition & $\1_{\{\cD_\rho > 0\}}\cN_\rho/\sqrt{\cD_\rho} \in L^1$ & $\1_{\{\rho_1\rho_2 > 0\}}\dot\rho_1/\sqrt{\rho_1\rho_2}\in L^1$ \\
		Constraints on $X_0$ & {\renewcommand{\arraystretch}{1.1}\begin{tabular}[c]{@{}c@{}}\rule{0pt}{2.6ex}$X_0\cdot e_3=\rho_1(0)-1=1-\rho_2(0)$\\ $X_0\cdot e_2 = -\dot\rho_1(0)/(2\tau)$\rule[-1.3ex]{0pt}{0pt}\end{tabular}} & $X_0\cdot e_3=\rho_1(0)-1=1-\rho_2(0)$ \\
	\end{tabular}
	\caption{Comparison of the VP and GP/GPc $V$-representability conditions. We recall that $\cD_\rho:=  4\tau^2\rho_1\rho_2-\dot\rho_1^2$ and $\cN_\rho:=\ddot\rho_1+4\tau^2(\rho_1-1)$.}
	\label{tab:VP-vs-GP}
\end{table}

\begin{table}[t]
	\centering
	\small
	\setlength{\tabcolsep}{4pt}
	\renewcommand{\arraystretch}{1.6}
	\begin{tabular}{l|c|c}
		& VP & GP / GPc \\
		\hline
		Class of effective potentials & $v\in L^1(0,T;\R)$ & {\renewcommand{\arraystretch}{1.1}\begin{tabular}[c]{@{}c@{}}\rule{0pt}{2.6ex}$w_c \in L^0(0,T;\R)$ with\\ $w_c\sqrt{\rho_1\rho_2}\in L^1(0,T;\R)$\rule[-1.3ex]{0pt}{0pt}\end{tabular}} \\
		Singular set on the Bloch sphere & great circle $\cC_s=\{x_1=0\}$ & poles $P_\pm$ \\
		{\renewcommand{\arraystretch}{1.1}\begin{tabular}[c]{@{}l@{}}\rule{0pt}{2.6ex}Free parameter per regular connected\\ component of admissible NI-trajectories\rule[-1.3ex]{0pt}{0pt}\end{tabular}} & sign $\eta_k\in\{-1,1\}$ & angle $a_\ell\in\R/(2\pi\Z)$ \\
		{\renewcommand{\arraystretch}{1.1}\begin{tabular}[c]{@{}l@{}}\rule{0pt}{2.6ex}Uniqueness of the NI-trajectory\\ for a fixed admissible effective potential\rule[-1.3ex]{0pt}{0pt}\end{tabular}} & yes & no \\
	\end{tabular}
	\caption{Comparison of the properties of the VP and GP/GPc effective potentials and non-interacting (NI) trajectories on the Bloch sphere. The set of singular times is the set of times at which the non-interacting trajectory hits the singular set.}
	\label{tab:VP-vs-GP-2}
\end{table}

\FloatBarrier
\section{Singular time crossing for interacting densities}
\label{sec:interacting}

In this section, we prove \cref{thm:NIVIN}, and we describe all the non-interacting trajectories reproducing an interacting density, together with the corresponding effective potentials. Throughout this section, we fix a normalized $\Psi_0^{\rm int}\in\C^3$ and $u\in L^1(0,T;\R)$, and we denote by $\Psi^u$ the solution to~\eqref{eq:interacting_dynamics} and by $\rho^u$ its density, which belongs to $W^{2,1}(0,T;\R^2)$ by \cref{prop:reg_density}. We use the shorthand notation $\cS^u:=\cS_{\rho^u}$, $\cT^u:=\cT_{\rho^u}$, $\cD^u:=\cD_{\rho^u}$ and $\cN^u:=\cN_{\rho^u}$, and we denote by $M_{\rm cc}$ and $N_{\rm cc}$ the numbers of connected components of $\Omega_{\rho^u}$ and $\widetilde\Omega_{\rho^u}$ (see \cref{thm:VP-characterization,thm:GP-characterization}).

\medskip
We begin with a lemma on the interacting states, which is the key ingredient of this section. For a normalized $\Psi\in\C^3$, we define the kinematic discriminant
\begin{equation*}
	\cD(\Psi):= 4\tau^2 \rho_1^\Psi \rho_2^\Psi - | j_1^\Psi |^2 .
\end{equation*}
By \cref{prop:reg_density}, $\dot\rho^u_1=-j^{\Psi^u}_1$ on $[0,T]$, so that $\cD^u(t)=\cD(\Psi^u(t))$ for every $t\in[0,T]$. A VP singular time of $\rho^u$ is thus a time at which the interacting state $\Psi^u(t)$ satisfies $\cD(\Psi^u(t))=0$. The following lemma characterizes such states.

\begin{lemma}\label{lem:Du}
	Let $\Psi=(z_1,z_2,z_3)^T$ be a normalized vector of $\C^3$. Then
	\begin{equation}\label{eq:ineq_D_improved}
		\cD(\Psi) \ge 8\tau^2\lp \lp\Re\lp \overline{z_3}z_1\rp\rp^2 + \lp \Re\lp \overline{z_3}z_2\rp\rp^2 \rp.
	\end{equation}
	In particular, $\cD(\Psi) \ge 0$, with equality if and only if $\Psi$ belongs to the set
	\begin{equation}\label{eq:singular-state}
		\cK_s := \left\{ e^{i\vartheta} \lp y, -(1-y), i\eta\sqrt{2y(1-y)} \rp^T \;\middle|\; (y,\vartheta,\eta) \in [0,1] \times \R/(2\pi\Z) \times \{-1,1\} \right\} .
	\end{equation}
	Moreover, $\cK_s \subset \cM^{\rm Slat}$, and the Bloch representation of $\cK_s$ is the great circle $\cC_s$.
\end{lemma}
\begin{proof}
	The proof of \Cref{lem:Du} is given in \Cref{ssec:proof-Du}.
\end{proof}

The states at which the inversion of the van Leeuwen equation degenerates are thus Slater states, represented on the Bloch sphere by the VP singular set $\cC_s$.

\medskip
We can now prove the main result of this article.

\begin{proof}[Proof of \cref{thm:NIVIN}]
	Let $\rho^u\in\cR_{\rm int}(\Psi_0^{\rm int})$. By \cref{prop:reg_density}, $\rho^u\in W^{2,1}(0,T;\R^2)$, $\rho^u(0)=\rho^{\Psi_0^{\rm int}}$ and $\dot\rho^u_1(0)=-j_1^{\Psi_0^{\rm int}}$.

\proofstep{Proof of (i).} We check the conditions of \cref{thm:VP-characterization}. By \cref{lem:Du}, $\cD^u=\cD(\Psi^u)\ge0$ on $[0,T]$. Moreover, writing $\Psi^u=(z_1,z_2,z_3)^T$, the Cauchy--Schwarz inequality in $\R^2$ applied to the van Leeuwen equation~\eqref{eq:vanleeuwen}, with $\lambda=1$ and $v=u$, and the improved bound~\eqref{eq:ineq_D_improved} give, a.e. on $(0,T)$,
	\begin{align*}
		|\cN^u| &\le 2\sqrt2\,\tau \sqrt{\lp 2u+U_1 \rp^2 + \lp 2u-U_2 \rp^2}\, \sqrt{\Re\lp \overline{z_3}z_1\rp^2 + \Re\lp \overline{z_3}z_2\rp^2} \\
		&\le \sqrt{\lp 2u+U_1 \rp^2 + \lp 2u-U_2 \rp^2}\,\sqrt{\cD^u},
	\end{align*}
	so that, since $\sqrt{(2u+U_1)^2 + (2u-U_2)^2} \le 2\sqrt2\,|u| + \sqrt{U_1^2+U_2^2}$,
	\begin{equation}
		\label{eq:N-over-D-bound}
		|\cN^u| \le \lp 2\sqrt2\,|u| + \sqrt{U_1^2+U_2^2}\rp\sqrt{\cD^u} \qquad\text{a.e. on } (0,T).
	\end{equation}
	Hence $\cN^u=0$ a.e. on $\cS^u$, and $\1_{\{\cD^u>0\}}\cN^u/\sqrt{\cD^u}\in L^1(0,T;\R)$ since $u\in L^1(0,T;\R)$. The only conditions of \cref{thm:VP-characterization} which are not automatically satisfied are therefore $\rho^u_1(0)=1+X_0\cdot e_3$ and $\dot\rho^u_1(0)=-2\tau X_0\cdot e_2$, that is $\rho_1^{\Psi_0^{\rm int}}=\rho_1^{X_0}$ and $j_1^{\Psi_0^{\rm int}}=j_1^{X_0}$ by~\eqref{eq:interpretation_x2} and~\eqref{eq:interpretation_x3}. Since $\rho_1+\rho_2=2$ and $j_1+j_2=0$, these conditions are equivalent to $\rho^{X_0}=\rho^{\Psi_0^{\rm int}}$ and $j^{X_0}=j^{\Psi_0^{\rm int}}$. They are sufficient for $\rho^u\in\cR_{\rm VP}(X_0)$ for every $u$, and necessary as soon as $\rho^u\in\cR_{\rm VP}(X_0)$ for one $u$.

\proofstep{Proof of (ii).} Let $X_0'\in\bS^2$ be such that $\rho^{X_0'}=\rho^{\Psi_0^{\rm int}}$ and $j^{X_0'}=j^{\Psi_0^{\rm int}}$, which exists by (iii). By point (i) and \cref{cor:VP-in-GP}, $\rho^u\in\cR_{\rm VP}(X_0')\subset\cR_{\rm GP}(X_0')$. Since the conditions in~\eqref{eq:RGP-characterization} only depend on $X_0$ through $X_0\cdot e_3$, we have $\cR_{\rm GP}(X_0')=\cR_{\rm GP}(X_0)=\cR_{\rm GPc}(X_0,v_c)$ as soon as $\rho^{X_0}=\rho^{X_0'}=\rho^{\Psi_0^{\rm int}}$, so that this condition is sufficient. It is also necessary, since $\rho^{X_0}=\rho^u(0)=\rho^{\Psi_0^{\rm int}}$ for $\rho^u\in\cR_{\rm GPc}(X_0,v_c)$.

\proofstep{Proof of (iii).} By \cref{lem:Du}, $\cD(\Psi_0^{\rm int})\ge0$, that is $\big(j_1^{\Psi_0^{\rm int}}/(2\tau)\big)^2\le\rho_1^{\Psi_0^{\rm int}}\rho_2^{\Psi_0^{\rm int}}=1-\big(\rho_1^{\Psi_0^{\rm int}}-1\big)^2$. Hence
\begin{equation*}
	X_0':=\Big(\sqrt{\rho_1^{\Psi_0^{\rm int}}\rho_2^{\Psi_0^{\rm int}}-\big(j_1^{\Psi_0^{\rm int}}/(2\tau)\big)^2},\ \f{j_1^{\Psi_0^{\rm int}}}{2\tau},\ \rho_1^{\Psi_0^{\rm int}}-1\Big)^T\in\bS^2
\end{equation*}
satisfies $\rho^{X_0'}=\rho^{\Psi_0^{\rm int}}$ and $j^{X_0'}=j^{\Psi_0^{\rm int}}$ by~\eqref{eq:interpretation_x2} and~\eqref{eq:interpretation_x3}.
\end{proof}

\medskip
Combining \cref{thm:NIVIN} with \cref{thm:VP-characterization,thm:GP-characterization}, we obtain a complete description of the non-interacting trajectories reproducing $\rho^u$, and bounds on the corresponding effective potentials.

\begin{theorem}[Non-interacting trajectories and potentials reproducing an interacting density]
	\label{thm:int-representability}
	Let $X_0\in\bS^2$ be such that $\rho^{X_0}=\rho^{\Psi_0^{\rm int}}$. The sets $\cS^u$ and $\cT^u$ are finite, so that the numbers $M_{\rm cc}$ and $N_{\rm cc}$ are finite. Moreover:
	\begin{enumerate}
		\item \emph{VP.} If $j^{X_0}=j^{\Psi_0^{\rm int}}$, the VP-trajectories reproducing $\rho^u$ are the $X^{\rm VP}_\eta$ of~\eqref{eq:VP-fiber-X}, with $\eta\in\fS_{\rho^u}(X_0)\simeq\{-1,1\}^{M_{\rm cc}-1}$ if $X_0\cdot e_1\neq0$ and $\fS_{\rho^u}(X_0)\simeq \{-1,1\}^{M_{\rm cc}}$ otherwise. For every $t_*\in\cS^u$, $\Psi^u(t_*)$ is a Slater state with Bloch vector $X^{\rm VP}_\eta(t_*)\in\cC_s$, and the potentials $v_\eta$ satisfy
		\begin{equation}
			\label{eq:bound-vp-int}
			|v_\eta(t)| \le \sqrt2\,|u(t)| + \f{\sqrt{U_1^2+U_2^2}}{2} \qquad\text{for a.e. } t\in[0,T] ,
		\end{equation}
		so that $v_\eta\in L^p(0,T;\R)$ as soon as $u\in L^p(0,T;\R)$, $p\in[1,+\infty]$.
		\item \emph{GPc.} For every $v_c\in L^1(0,T;\R)$, the GPc-trajectories reproducing $\rho^u$ are the $X^{\rm GPc}_{v_c,a}$ of~\eqref{eq:GP-fiber-X}, with $a\in\cA_{\rho^u}(X_0)\simeq(\R/(2\pi\Z))^{N_{\rm cc}-1}$ if $X_0\notin\{P_+,P_-\}$ and $\cA_{\rho^u}(X_0)\simeq(\R/(2\pi\Z))^{N_{\rm cc}}$ otherwise. For every $t_z\in\cT^u$, $\Psi^u(t_z)$ is a Slater state with Bloch vector $X^{\rm GPc}_{v_c,a}(t_z)\in\{P_+,P_-\}$, and the potentials $w_{c,a}$ satisfy
		\begin{equation}
			\label{eq:bound-wc-int}
			|w_{c,a}(t)|\sqrt{\rho^u_1(t)\rho^u_2(t)} \le 2\tau \qquad\text{for a.e. } t\in[0,T] .
		\end{equation}
	\end{enumerate}
\end{theorem}
\begin{proof}
	The proof of \Cref{thm:int-representability} is given in \Cref{ssec:proof-int-rep}.
\end{proof}

\begin{remark}\label{rem:Cs}
	The fact that $\cS^u$ is finite does not hold for general VP-representable densities. For instance, if $X_0\in\cC_s$ and $v=0$, the first equation of~\eqref{eq:TDVP_coord} gives $x_1=0$ on $[0,T]$, so that the density $\rho$ of the solution to~\eqref{eq:vp_X} satisfies $\cS_\rho=[0,T]$. On the other hand, the sets $\cS^u$ and $\cT^u$ can be non-empty. Indeed, by \cref{prop:control} in \cref{appendix:controllability}, for $T$ large enough and any $\Psi_0^{\rm int}$, there exist $u\in L^1(0,T;\R)$ and $t_*\in(0,T]$ such that $\Psi^u(t_*)\in\cK_s$ (resp. $\Psi^u(t_*)=(1,0,0)^T$), so that $t_*\in\cS^u$ (resp. $t_*\in\cT^u$).
\end{remark}

\medskip
Since $\cS^u$ and $\cT^u$ are finite, we order the connected components of $\Omega_{\rho^u}$ and $\widetilde\Omega_{\rho^u}$ increasingly, so that each singular time in $(0,T)$ separates two consecutive components. At $t_*\in\cS^u\cap(0,T)$, between $I_k$ and $I_{k+1}$, the VP-trajectories reach $\cC_s$, and the sign $\eta_{k+1}$ is free: the trajectory crosses $\cC_s$ if $\eta_{k+1}=-\eta_k$, and bounces back if $\eta_{k+1}=\eta_k$ (see \cref{fig:VP-branching}). At $t_z\in\cT^u\cap(0,T)$, between $J_\ell$ and $J_{\ell+1}$, the GPc-trajectories reach a pole, where the angle $\alpha_a$ is not defined, and the parameter $a_{\ell+1}\in\R/(2\pi\Z)$ is free: the trajectory branches into a circle of continuations (see \cref{fig:alpha-branching}). These branchings are illustrated numerically in \cref{sec:numerics} (see \cref{fig:numerics-bloch,fig:numerics-vp}).

\begin{figure}[t]
	\centering
	\includegraphics[width=0.65\linewidth]{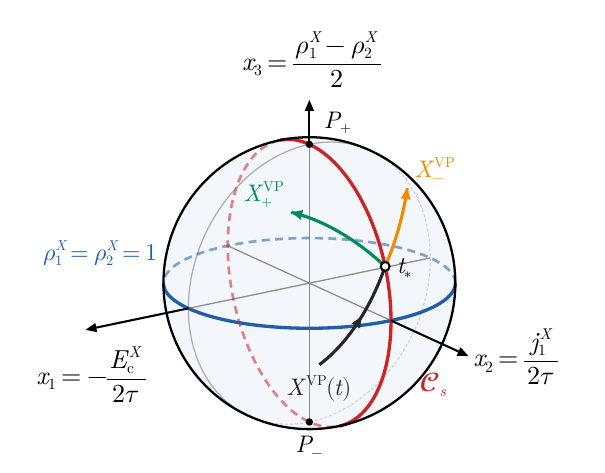}
	\caption{
		Branching of the VP-trajectories reproducing an interacting density. The trajectory $X^{\rm VP}(t)$ reaches the great circle $\cC_s$ (in red, the Bloch representation of $\cK_s$) at a singular time $t_\ast\in\cS^u$. After $t_\ast$, the two continuations $X^{\rm VP}_\pm$, with $x_1>0$ and $x_1<0$, reproduce the same density.
	}
	\label{fig:VP-branching}
\end{figure}

\begin{figure}[t]
	\centering
	\includegraphics[width=0.65\linewidth]{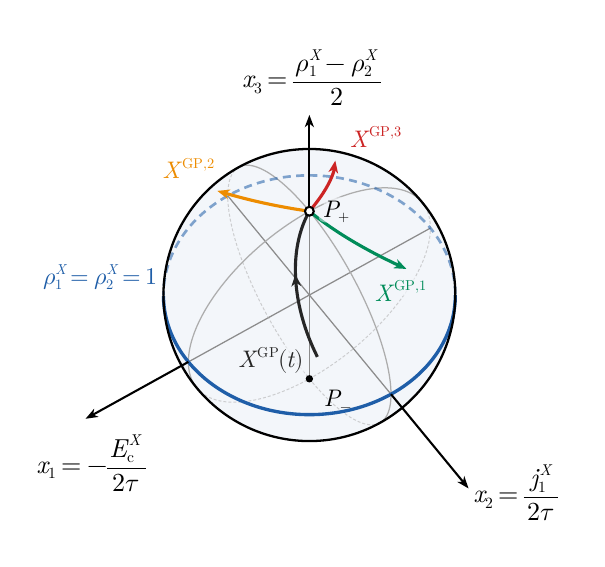}
	\caption{
		Branching of the GP-trajectories ($v_c=0$) reproducing an interacting density at a pole. The trajectory $X^{\rm GP}(t)$ reaches the North Pole $P_+$ at time $t_z$, where the longitude is ill defined. At time $t_z+\epsilon$, $0 < \epsilon \ll 1$, any longitude defines a trajectory reproducing the same density on the time interval $(t_z,t_z')$, where $t_z'$ is the next singular time. For instance, the trajectories $X^{\rm GP,1}$, $X^{\rm GP,2}$ and $X^{\rm GP,3}$ correspond to three values of the free parameter $a_{\ell+1}$ of \cref{thm:GP-characterization}.
	}
	\label{fig:alpha-branching}
\end{figure}

\medskip
The differences between general and interacting densities are summarized in \cref{tab:general-vs-int}.

\begin{table}[t]
	\centering
	\small
	\setlength{\tabcolsep}{4pt}
	\renewcommand{\arraystretch}{1.6}
	\begin{tabular}{l|c|c}
		& general density & interacting density \\
		\hline
		VP singular set $\cS_\rho$ & closed, possibly of positive measure & finite \\
		GP singular set $\cT_\rho$ & closed, of zero measure & finite \\
		Admissible VP-trajectories & $\fS_\rho(X_0)$, possibly uncountable & finite binary tree \\
		Admissible GPc-trajectories & $\cA_\rho(X_0)$, possibly infinite-dimensional & finite-dimensional torus \\
		VP potential & $v\in L^1(0,T;\R)$ & $|v|\le\sqrt2\,|u|+\f12\sqrt{U_1^2+U_2^2}$ \\
		GPc potential & $w_c\sqrt{\rho_1\rho_2}\in L^1(0,T;\R)$ & $|w_c|\sqrt{\rho_1\rho_2}\le2\tau$ \\
	\end{tabular}
	\caption{Comparison of the singular sets, of the admissible non-interacting trajectories, and of the effective potentials, for a general VP- or GPc-representable density and for an interacting density.}
	\label{tab:general-vs-int}
\end{table}

\FloatBarrier
\section{Integrability of the geometric potential}\label{sec:integrability}

In this section, we determine whether the geometric potentials reproducing an interacting density can be chosen integrable. This justifies the functional setting chosen for GP in \cref{subsec:interacting}. Throughout this section, we fix a normalized $\Psi_0^{\rm int}\in\C^3$, $u\in L^1(0,T;\R)$, $X_0\in\bS^2$ such that $\rho^{X_0}=\rho^{\Psi_0^{\rm int}}$, and $v_c\in L^1(0,T;\R)$. We denote by $\Psi^u=(z_1,z_2,z_3)^T$ the interacting dynamics and by $\rho:=\rho^u$ the target density, which belongs to $\cR_{\rm GPc}(X_0,v_c)$ by \cref{thm:NIVIN}, and we use the notation of \cref{thm:GP-characterization}. We also set $Z_i:=\left\{t\in[0,T]\;\middle|\;\rho_i(t)=0\right\}$, $i=1,2$, so that $\cT^u=Z_1\cup Z_2$ is finite by \cref{thm:int-representability}.

For $a\in\cA_\rho(X_0)$, we introduce the longitude $\phi:=\alpha_a+\f\pi2$ of $X^{\rm GPc}_{v_c,a}$ on $\widetilde\Omega_\rho$. By~\eqref{eq:GP-fiber-X}, \eqref{eq:angle_gp} and~\eqref{eq:GP-fiber}, we have on $\widetilde\Omega_\rho$
\begin{equation*}
	X^{\rm GPc}_{v_c,a}=\lp\sqrt{\rho_1\rho_2}\cos\phi,\ \sqrt{\rho_1\rho_2}\sin\phi,\ \rho_1-1\rp^T ,
\end{equation*}
\begin{equation}\label{eq:theta-ode}
	\dot\phi=2v_c+\beta\cos\phi, \qquad \beta:=\f{2\tau\lp\rho_1-1\rp}{\sqrt{\rho_1\rho_2}} ,
\end{equation}
\begin{equation}\label{eq:w-theta}
	w_{c,a}=-\f{\dot\rho_1+2\tau\sqrt{\rho_1\rho_2}\,\sin\phi}{2\rho_1\rho_2} .
\end{equation}
By~\eqref{eq:bound-wc-int}, $w_{c,a}$ is bounded on every compact subset of $\widetilde\Omega_\rho$. The integrability of $w_{c,a}$ is therefore only a question at the singular times $t_z\in\cT^u$.

This requires understanding the behavior of $\rho$, and of the coefficient $\beta$ of~\eqref{eq:theta-ode}, which is singular, near the singular times. In the rest of this section, we fix $t_z\in Z_1$, the case $t_z\in Z_2$ being treated in \cref{rem:Z2} below. Since $\cT^u$ is finite and $\rho_1$ is continuous, we can fix $\delta\in(0,1]$ such that
\begin{equation*}
	\cT^u\cap[t_z-\delta,t_z+\delta]=\{t_z\} \qquad\text{and}\qquad \rho_1\le\tfrac12 \ \text{ on } [t_z-\delta,t_z+\delta]\cap[0,T] .
\end{equation*}
We write $s:=t-t_z$, $I^-:=[t_z-\delta,t_z)\cap[0,T]$, $I^+:=(t_z,t_z+\delta]\cap[0,T]$, and consider the functions $\omega, \hat\omega : [0,\delta]\to \R_+$ defined by
\begin{equation}\label{eq:omega}
	\omega(\sigma):=\int_{[t_z-\sigma,t_z+\sigma]\cap[0,T]}\lp|u|+|v_c|\rp, \qquad \hat\omega(\sigma):=\sigma+\omega(\sigma), \qquad \sigma\in(0,\delta],
\end{equation}
which are non-decreasing and tend to $0$ as $\sigma\to0^+$. The key lemma is the following.

\begin{lemma}\label{lem:a-expansion} Let $t_z \in Z_1$.
	The coefficient $\beta$ is negative on $I^-\cup I^+$, and, as $s\to0$,
	\begin{equation}\label{eq:rho1-refined}
		\dot\rho_1(t_z+s)=4\tau^2s\lp1+O\lp\hat\omega(|s|)\rp\rp ,
		\qquad
		\sqrt{\rho_1\rho_2}\,(t_z+s)=2\tau|s|\lp1+O\lp\hat\omega(|s|)\rp\rp,
	\end{equation}
	\begin{equation}\label{eq:b-refined}
		|\beta(t_z+s)|=\f{1}{|s|}\lp1+O\lp\hat\omega(|s|)\rp\rp .
	\end{equation}
\end{lemma}
\begin{proof}
	The proof of \Cref{lem:a-expansion} is given in \Cref{ssec:proof-a-expansion}.
\end{proof}

Let us first see what is required on $\phi$ for $w_{c,a}$ to be integrable near $t_z$. By~\eqref{eq:w-theta} and \cref{lem:a-expansion}, we have, as $s\to0$,
\begin{equation}\label{eq:bracket}
	w_{c,a}(t_z+s)=-\f{\sgn(s)+\sin\phi(t_z+s)}{2|s|}+O\lp\f{\hat\omega(|s|)}{|s|}\rp .
\end{equation}
Assume first that the remainder in~\eqref{eq:bracket} is integrable near $s=0$. As $\hat\omega(s)/s=1+\omega(s)/s$, this holds under the Dini-type condition
\begin{equation}\label{eq:dini}
	\int_0^\delta\f{\omega(s)}{s}\,\rd s=\int_{[t_z-\delta,t_z+\delta]\cap[0,T]}\lp|u(t)|+|v_c(t)|\rp\log\f{\delta}{|t-t_z|}\,\rd t<\infty ,
\end{equation}
where the equality follows from Fubini's theorem. This condition is satisfied as soon as $u,v_c\in L^r(I^-\cup I^+;\R)$ for some $r>1$, and in particular for bounded, hence for smooth, potentials $u$ and $v_c$. Then, since $1/|s|$ is not integrable at $0$, $w_{c,a}$ can only be integrable on $I^+$ (resp. $I^-$) if the numerator of the leading term tends to $0$, that is if $\sin\phi(t)\to-1$ as $t\to t_z^+$ (resp. $\sin\phi(t)\to1$ as $t\to t_z^-$), that is if $\phi(t)\to-\f\pi2$ as $t\to t_z^+$ (resp. $\phi(t)\to\f\pi2$ as $t\to t_z^-$), fast enough (see \cref{thm:integrability} below). We thus have to determine the behavior of $\phi$ near $t_z$.

To this end, we use the fact that~\eqref{eq:theta-ode} is the projectivization of a linear system. Setting $\zeta:=\tan\lp\f\phi2+\f\pi4\rp$, which defines a homeomorphism from $\R/(2\pi\Z)$ onto $\R\bP^1=\R\cup\{\infty\}$, \eqref{eq:theta-ode} becomes the Riccati equation
\begin{equation}\label{eq:riccati}
	\dot\zeta=\beta\,\zeta+v_c\lp1+\zeta^2\rp ,
\end{equation}
the meridians $\phi=\f\pi2$ and $\phi=-\f\pi2$ corresponding to $\zeta=\infty$ and $\zeta=0$. As for any Riccati equation, $\zeta=p/q$ solves~\eqref{eq:riccati} whenever $\chi=(p,q)^T\in\R^2\setminus\{0\}$ solves
\begin{equation}\label{eq:linear}
	\dot\chi=M\chi, \qquad M:=\begin{pmatrix} \beta/2 & v_c \\ -v_c & -\beta/2\end{pmatrix} .
\end{equation}
On $I^-$ and $I^+$, where $\beta$ is locally integrable, there is a bijection between the solutions of~\eqref{eq:theta-ode} and the solutions of~\eqref{eq:linear} (up to a multiplicative constant). We set
\begin{equation}\label{eq:Lambda}
	\Lambda^+(t):=-\int_t^{t_z+\delta}|\beta| \quad\text{for } t\in I^+, \qquad \Lambda^-(t):=\int_{t_z-\delta}^t|\beta| \quad\text{for } t\in I^-,
\end{equation}
so that $\dot\Lambda^\pm=|\beta|$, and $\Lambda^+(t)\to-\infty$ as $t\to t_z^+$ and $\Lambda^-(t)\to+\infty$ as $t\to t_z^-$ by~\eqref{eq:b-refined}.

\medskip

We start with the simple case $v_c=0$, that is GP. The system~\eqref{eq:linear} is then diagonal, and its solutions on $I^\pm$ are
\begin{equation*}
	\chi=\lp c_1e^{-\Lambda^\pm/2},\ c_2e^{\Lambda^\pm/2}\rp^T, \qquad (c_1,c_2)\in\R^2\setminus\{0\}, \qquad\text{so that}\qquad \zeta=\f{c_1}{c_2}\,e^{-\Lambda^\pm} .
\end{equation*}
These exponentials are explicit up to $\omega$. Indeed, by definition of $\Lambda^\pm$,
\begin{equation*}
	e^{\Lambda^\pm(t_z+s)/2}=\Big(\f{|s|}{\delta}\Big)^{\pm1/2}\exp\Big(\mp\f12\int_{|s|}^\delta\Big(|\beta(t_z\pm\sigma)|-\f1\sigma\Big)\rd\sigma\Big),
\end{equation*}
where, by~\eqref{eq:b-refined}, the integrand is $O(\hat\omega(\sigma)/\sigma)$, hence integrable on $(0,\delta)$ under~\eqref{eq:dini}. Under~\eqref{eq:dini}, we thus have $e^{\Lambda^\pm(t_z+s)/2}=m_\pm(|s|/\delta)^{\pm1/2}(1+o(1))$ as $s\to0^\pm$, with $m_\pm:=\exp(\mp\frac12\int_0^\delta(|\beta(t_z\pm\sigma)|-\frac1\sigma)\,\rd\sigma)\in(0,+\infty)$. In other words, $\chi$ behaves like $(c_1s^{-1/2},c_2s^{1/2})^T$ on $I^+$ and like $(c_1|s|^{1/2},c_2|s|^{-1/2})^T$ on $I^-$. Thus, on $I^+$, $\zeta\to\infty$, that is $\phi\to\f\pi2$, except for $c_1=0$, that is $\phi=-\f\pi2$, and, on $I^-$, $\zeta\to0$, that is $\phi\to-\f\pi2$, except for $c_2=0$, that is $\phi=\f\pi2$. Hence, for GP, only the solutions $\phi=\f\pi2$ on $I^-$ and $\phi=-\f\pi2$ on $I^+$ satisfy the above condition. Every other solution reaches the pole along $\phi=-\f\pi2$, or leaves it along $\phi=\f\pi2$, and the corresponding potential $w_{c,a}$ is of order $1/|s|$ on the corresponding side, hence not integrable.

\medskip

For general $v_c$, the off-diagonal terms of $M$ are integrable, so that this picture persists. This is summarized in the following proposition.

\begin{proposition}[Crossing singularities at the South Pole] \label{prop:dichotomy} Let $t_z \in Z_1$.
	We have the following.
	\begin{enumerate}
		\item \emph{Outgoing side.} Assume $t_z<T$. There is exactly one solution $\phi^\star$ of~\eqref{eq:theta-ode} on $I^+$ such that $\phi^\star(t)\to-\f\pi2$ as $t\to t_z^+$, and it satisfies $|\phi^\star(t_z+s)+\f\pi2|\le4\,\omega(s)$ for $s$ small enough. Every other solution satisfies $\phi(t)\to\f\pi2$ as $t\to t_z^+$.
		\item \emph{Incoming side.} Assume $t_z>0$. There is exactly one solution $\phi_\star$ of~\eqref{eq:theta-ode} on $I^-$ such that $\phi_\star(t)\to\f\pi2$ as $t\to t_z^-$, and it satisfies $|\phi_\star(t_z+s)-\f\pi2|\le4\,\omega(|s|)$ for $|s|$ small enough. Every other solution satisfies $\phi(t)\to-\f\pi2$ as $t\to t_z^-$.
	\end{enumerate}
	The solutions $\phi^\star$ and $\phi_\star$ are called the exceptional outgoing and incoming branches at $t_z$.
\end{proposition}
\begin{proof}
	The proof of \Cref{prop:dichotomy} is given in \Cref{ssec:proof-dichotomy}.
\end{proof}

\begin{remark}[Crossing singularities at the North Pole]\label{rem:Z2}
	If $t_z\in Z_2$, exchanging the two sites, that is $z_1$ and $z_2$, $U_1$ and $U_2$, $u$ and $-u$, and $v_c$ and $-v_c$, maps $\rho$ to $(\rho_2,\rho_1)$ and the Bloch vector $(x_1,x_2,x_3)$ to $(x_1,-x_2,-x_3)$, that is $\phi$ to $-\phi$ and $w_{c,a}$ to $-w_{c,a}$. Hence all the results of this section hold with the meridians $\phi=\f\pi2$ and $\phi=-\f\pi2$ exchanged, that is $\phi^\star(t)\to\f\pi2$ as $t\to t_z^+$ and $\phi_\star(t)\to-\f\pi2$ as $t\to t_z^-$.
\end{remark}

The alternative is illustrated in \cref{fig:branches-south} at the South Pole $P_-$, that is for $t_z\in Z_1$, and in \cref{fig:branches-north} at the North Pole $P_+$, that is for $t_z\in Z_2$.

\begin{figure}[t]
	\centering
	\includegraphics[width=\linewidth]{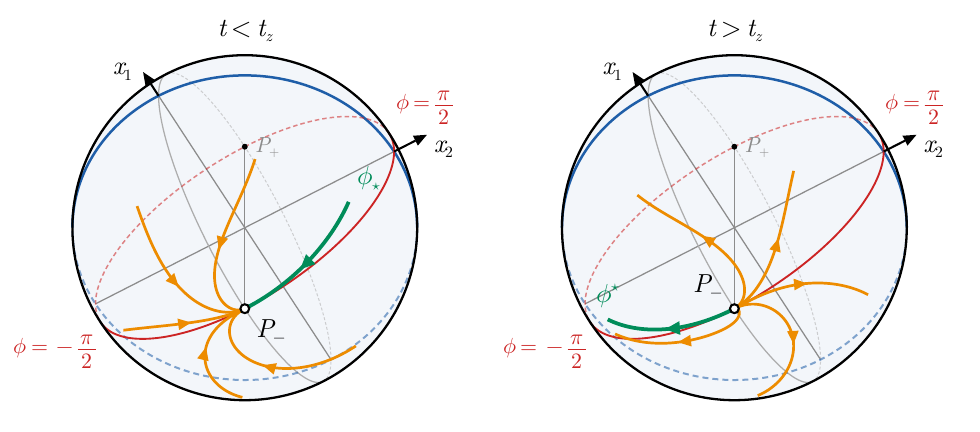}
	\caption{
		The alternative of \cref{prop:dichotomy} at the South Pole $P_-$, that is for $t_z\in Z_1$. Left: before $t_z$, every solution of~\eqref{eq:theta-ode} (orange) reaches $P_-$ along the meridian $\phi=-\f\pi2$, except the exceptional incoming branch $\phi_\star$ (green), which reaches it along $\phi=\f\pi2$. Right: after $t_z$, every solution leaves $P_-$ along the meridian $\phi=\f\pi2$, except the exceptional outgoing branch $\phi^\star$ (green), which leaves it along $\phi=-\f\pi2$. Under~\eqref{eq:dini}, the geometric potential is integrable near $t_z$ only along the exceptional branches.
	}
	\label{fig:branches-south}
\end{figure}

\begin{figure}[t]
	\centering
	\includegraphics[width=\linewidth]{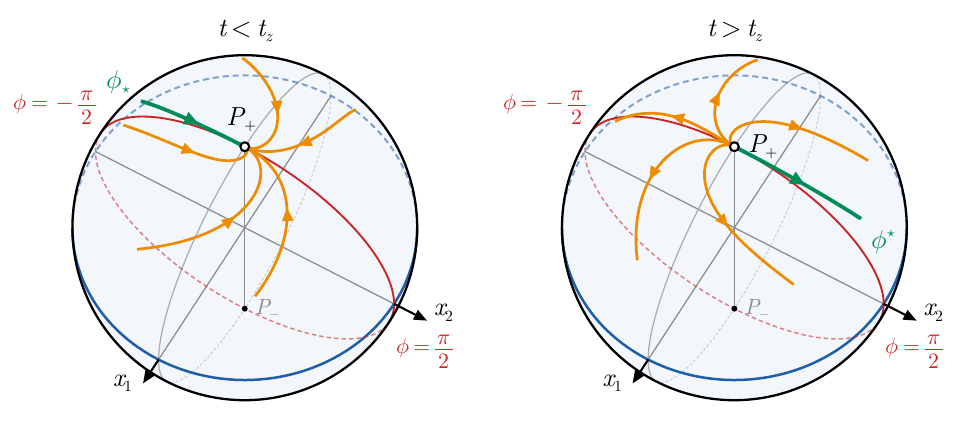}
	\caption{
		The alternative of \cref{rem:Z2} at the North Pole $P_+$, that is for $t_z\in Z_2$, for which the meridians $\phi=\f\pi2$ and $\phi=-\f\pi2$ are exchanged (see \cref{rem:Z2}). Left: before $t_z$, every solution of~\eqref{eq:theta-ode} (orange) reaches $P_+$ along the meridian $\phi=\f\pi2$, except the exceptional incoming branch $\phi_\star$ (green), which reaches it along $\phi=-\f\pi2$. Right: after $t_z$, every solution leaves $P_+$ along the meridian $\phi=-\f\pi2$, except the exceptional outgoing branch $\phi^\star$ (green), which leaves it along $\phi=\f\pi2$.
	}
	\label{fig:branches-north}
\end{figure}

\begin{theorem}[Generic non-integrability of $w$ at a singular time]\label{thm:integrability}
	Let $a\in\cA_\rho(X_0)$. Then, as $s\to0$,
	\begin{equation}\label{eq:w-expansion}
		w_{c,a}(t_z+s)=-\f{\sgn(s)+\sin\phi(t_z+s)}{2|s|}+O\lp\f{\hat\omega(|s|)}{|s|}\rp .
	\end{equation}
	Consequently, on $I^+$ (resp. $I^-$):
	\begin{enumerate}
		\item if $\phi$ is not the exceptional branch, then, as $s\to0^+$ (resp. $s\to0^-$), $\sqrt{\rho_1\rho_2}\,|w_{c,a}|(t_z+s)\to2\tau$ and $|w_{c,a}(t_z+s)|=|s|^{-1}\lp1+o(1)\rp$, so that $w_{c,a}$ is not integrable on $I^+$ (resp. $I^-$),
		\item if $\phi$ is the exceptional branch, then $w_{c,a}(t_z+s)=O\lp\hat\omega(|s|)/|s|\rp$ as $s\to0^+$ (resp. $s\to0^-$), so that $w_{c,a}$ is integrable on $I^+$ (resp. $I^-$) under condition~\eqref{eq:dini}.
	\end{enumerate}
\end{theorem}
\begin{proof}
	The proof of \Cref{thm:integrability} is given in \Cref{ssec:proof-integrability}.
\end{proof}

\begin{remark}
	By \cref{thm:integrability}~(i), the bound~\eqref{eq:bound-wc-int} is sharp.
\end{remark}

We can now answer the question of the integrability of the geometric potential. At each singular time, the outgoing branch can be chosen through the free parameter $a_\ell$ of \cref{thm:GP-characterization}, but the incoming branch is determined by the previous choices and by $X_0$.

\begin{corollary}[Integrable geometric potentials]\label{cor:integrability}
	Assume that~\eqref{eq:dini} holds at every $t_z\in\cT^u$, and let $a\in\cA_\rho(X_0)$. Then $w_{c,a}\in L^1(0,T;\R)$ if and only if, at every $t_z\in\cT^u$, the longitude of $X^{\rm GPc}_{v_c,a}$ is the exceptional incoming branch on $I^-$ if $t_z>0$, and the exceptional outgoing branch on $I^+$ if $t_z<T$. If moreover $u,v_c\in L^\infty(0,T;\R)$, then, under the same condition, $w_{c,a}\in L^\infty(0,T;\R)$. Finally, without assuming~\eqref{eq:dini}, if $0\notin\cT^u$ and $\cT^u\cap(0,T)\neq\emptyset$, then there is at most one $X_0'\in\bS^2$ with $\rho^{X_0'}=\rho(0)$ such that $w_{c,a}\in L^1(0,T;\R)$ for some $a\in\cA_\rho(X_0')$.
\end{corollary}
\begin{proof}
	The proof of \Cref{cor:integrability} is given in \Cref{ssec:proof-cor-integrability}.
\end{proof}

Consequently, in general, the geometric potentials reproducing an interacting density are not integrable as soon as the density has a singular time in $(0,T)$, even for smooth external potentials, in contrast with the VP potentials $v_\eta$, which are integrable by~\eqref{eq:bound-vp-int}.

\section{Numerical illustrations}
\label{sec:numerics}

We illustrate the results of \cref{sec:interacting,sec:integrability} on explicit interacting trajectories, with $\tau=U_1=U_2=1$, $T=4$, and the bang-bang potential $u\in\cU^1_{\rm bb}(0,T)$ (see \cref{prop:control} for the definition) equal to $1$ on $[0,1.3)$, to $-1$ on $[1.3,2.6)$, and to $1$ on $[2.6,4]$. We prescribe the singular state at the singular time $t_*=2$ (for VP) or $t_z=2$ (for GPc), and obtain the initial state $\Psi_0^{\rm int}$ by integrating the interacting dynamics backward in time. The singular states for VP and GPc are respectively $\Psi^u(t_*)=\Phi_Y$ with $Y=(0,-0.96,-0.28)^T\in\cC_s$ and  $\Psi^u(t_z)=e_2$. We consider the GPc dynamics with $v_c=u$. The dynamics is integrated exactly, the Hamiltonian being constant on each interval on which $u$ is constant. For the first trajectory, $\cS^u=\{t_*\}$, and for the second one, $\cT^u=\{t_z\}$.

The trajectories on the Bloch sphere near the singular times are shown in \cref{fig:numerics-bloch}, and the corresponding densities and potentials in \cref{fig:numerics-vp,fig:numerics-gp}.

\begin{figure}[!ht]
	\centering
	\includegraphics[width=\linewidth]{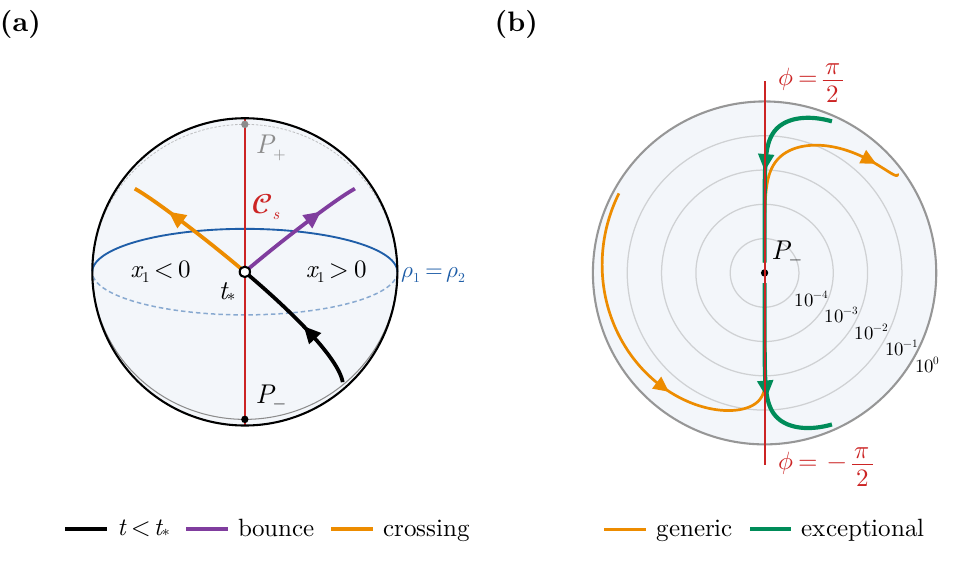}
	\caption{
		Computed non-interacting trajectories near the singular times. (a) VP-trajectories for $|t-t_*|\le0.4$, seen from the singular point $Y\in\cC_s$, so that $\cC_s$ appears as a vertical line: before $t_*$ (black), the trajectory reaches $\cC_s$, and then either bounces back (purple) or crosses it (orange). (b) GPc-trajectories ($v_c=u$) for $|t-t_z|\le0.4$ in a polar chart centered at the south pole $P_-$, with angle the longitude $\phi$ and radius $\log_{10}\sqrt{\rho_1^u\rho_2^u}$ (circles at $10^{-4},\dots,10^{0}$): a generic trajectory (orange) reaches $P_-$ along the meridian $\phi=-\frac\pi2$ and leaves it along $\phi=\frac\pi2$, and the exceptional branches (green) do the opposite.
	}
	\label{fig:numerics-bloch}
\end{figure}

\begin{figure}[!ht]
	\centering
	\includegraphics[width=0.85\linewidth]{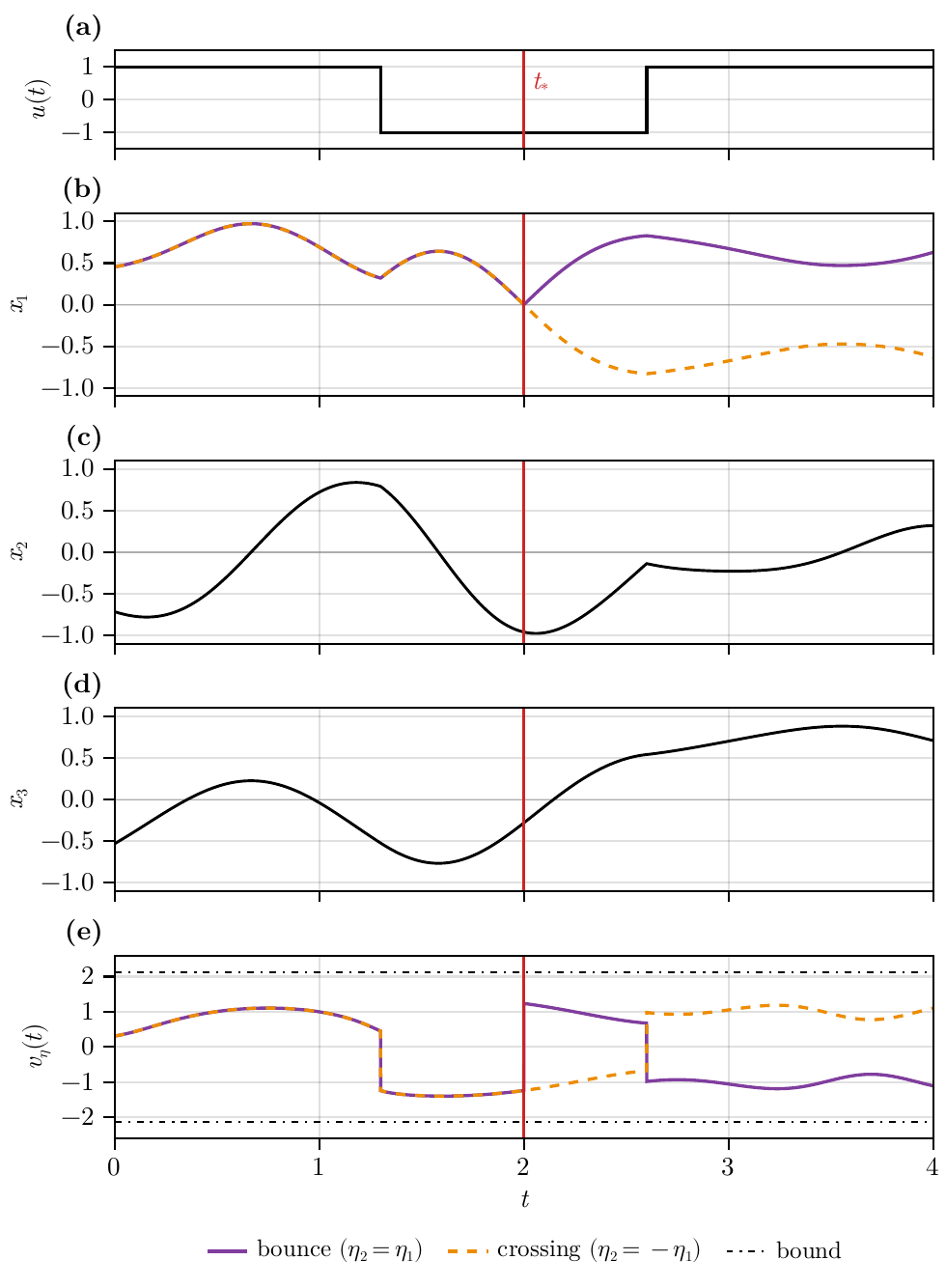}
	\caption{
		VP-trajectories reproducing an interacting density through the VP singular state $\Phi_Y$, $Y\in\cC_s$, at $t_*=2$ (vertical line). (a) Bang-bang potential $u$. (b) $x_1=\pm\sqrt{\cD^u}/(2\tau)$ for the two continuations at $t_*$. It vanishes at $t=t_*$. (c) $x_2=-\dot\rho_1^u/(2\tau)$, which is fixed by the prescribed interacting density. (d) $x_3=\rho_1^u-1$, which is also fixed by the prescribed interacting density. (e) Corresponding VP potentials $v_\eta$. Before $t=t_*$ they coincide, and are opposite after $t_*$. Both satisfy the bound~\eqref{eq:bound-vp-int} (dash-dotted lines). Numerically, we have $\max_{[0,T]}|v_\eta|/(\sqrt2|u|+\frac12 (U_1^2+U_2^2)^{1/2})\approx0.66$.
	}
	\label{fig:numerics-vp}
\end{figure}

\begin{figure}[!ht]
	\centering
	\includegraphics[width=0.8\linewidth]{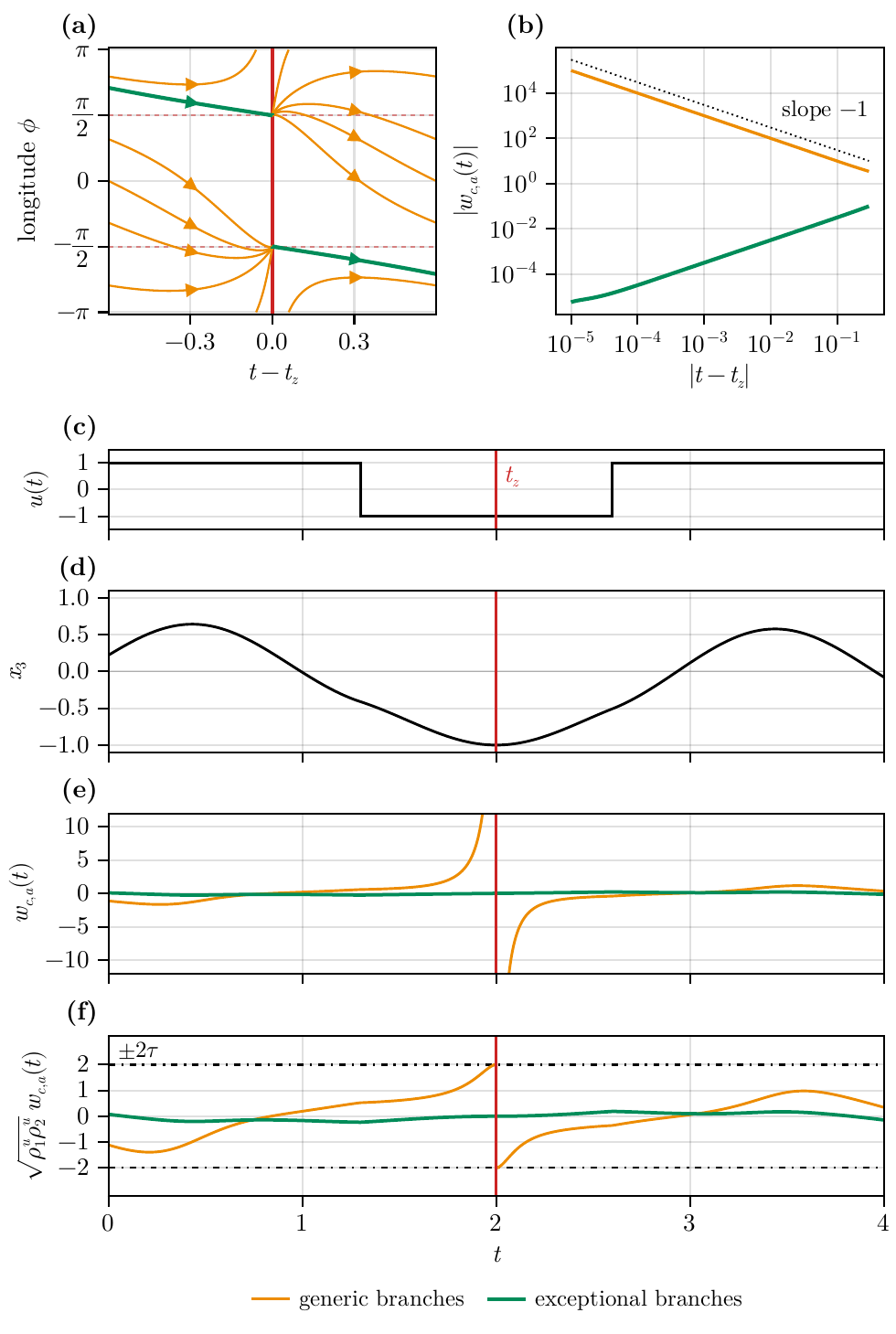}
	\caption{
		GPc-trajectories ($v_c=u$) reproducing the interacting density $\rho^u$ through the south pole $P_-$ at $t_z=2$, illustrating \cref{prop:dichotomy}. (a) Longitude $\phi$ of several generic branches (orange) and of the exceptional branches (green) near $t_z$, the arrows indicating the direction of time; the meridians $\phi=\pm\frac\pi2$ are dashed. (b) Modulus of the corresponding geometric potentials as a function of $|t-t_z|$. Along generic branches, $|w_{c,a}(t)|$ behaves like $1/|t-t_z|$, and is therefore not integrable, whereas it remains bounded along the exceptional branches. (c) Bang-bang potential $u$. (d) $x_3=\rho_1^u-1$. (e) Geometric potential along a generic GPc-trajectory (orange) and along the GPc-trajectory made of the exceptional branches (green). (f) Weighted potential $\sqrt{\rho_1^u\rho_2^u}\,w_{c,a}$ along the same trajectories. All the trajectories satisfy the bound~\eqref{eq:bound-wc-int} (dash-dotted lines) which is attained in the limit $t\to t_z$ for the generic one.
	}
	\label{fig:numerics-gp}
\end{figure}

\FloatBarrier

\section{Proofs}
\label{sec:proofs}

In this section, we provide the proofs of the results stated in \crefrange{sec:hubbard}{sec:integrability}, in the order in which they are stated. We also state and prove some useful intermediate results.

\subsection{Proof of \cref{prop:slater_set}}
\label{ssec:proof-slater}

Let $\lambda \ge 0$, $v\in L^1(0,T;\R)$, $w\in L^0(0,T;\R)$ and $\Psi(t)=(z_1(t),z_2(t),z_3(t))^T\in W^{1,1}(0,T;\C^3)$ be a solution of \eqref{eq:schrodinger} with the Hamiltonian $H^\lambda_{v,w}$ given by \eqref{eq:ham_general}. According to the definition of Slater states \eqref{eq:def_mks}, we define the map
\begin{equation*}
	\Delta^{\rm Slat}(t) := z_3(t)^2 - 2 z_1(t) z_2(t) , \qquad t\in[0,T],
\end{equation*}
which vanishes if and only if $\Psi(t)$ is a Slater state. Since $\langle \Psi, V\Psi\rangle = 2|z_1|^2-2|z_2|^2=\rho_1^\Psi-\rho_2^\Psi$, one has
\begin{equation*}
	[V,P_\Psi]\Psi = V\Psi - \langle \Psi, V\Psi\rangle\Psi = \lp 2\rho_2^\Psi z_1,\, -2\rho_1^\Psi z_2,\, -(\rho_1^\Psi-\rho_2^\Psi) z_3\rp^T,
\end{equation*}
where we used that $\|\Psi(t)\|=1$ for all $t\in[0,T]$ by point (i) of \cref{prop:reg_density}, so that the components of $\Psi$ solve the system
\begin{equation}\label{explicit_schrodinger}
\begin{cases}
\dot z_1(t)
=
-i\lp \lambda U_1+2v(t)\rp z_1(t)
+i\sqrt{2}\,\tau z_3(t)
-2w(t)\rho_2^\Psi(t)z_1(t),
\\[2pt]
\dot z_2(t)
=
-i\lp \lambda U_2-2v(t)\rp z_2(t)
+i\sqrt{2}\,\tau z_3(t)
+2w(t)\rho_1^\Psi(t)z_2(t),
\\[2pt]
\dot z_3(t)
=
i\sqrt{2}\,\tau\lp z_1(t)+z_2(t)\rp
+w(t)\lp \rho_1^\Psi(t)-\rho_2^\Psi(t)\rp z_3(t).
\end{cases}
\end{equation}

Because $z_1,z_2,z_3$ are absolutely continuous, $\Delta^{\rm Slat}$ is also absolutely continuous, and \eqref{explicit_schrodinger} gives for almost every $t\in[0,T]$
\begin{align*}
	\dot\Delta^{\rm Slat}
	&= 2z_3\dot z_3 - 2\dot z_1 z_2 - 2 z_1\dot z_2 \\
	&= 2w\lp\rho_1^\Psi-\rho_2^\Psi\rp z_3^2 + 2i\lambda\lp U_1+U_2\rp z_1z_2 - 4w\lp\rho_1^\Psi-\rho_2^\Psi\rp z_1z_2 ,
\end{align*}
the hopping terms cancelling each other, as well as the contributions of $v$. Writing $2z_1z_2=z_3^2-\Delta^{\rm Slat}$, the terms in $w z_3^2$ also cancel and we obtain
\begin{equation}\label{eq:Delta-ode}
	\dot\Delta^{\rm Slat}(t) = i\lambda\,(U_1+U_2)\, z_3(t)^2 - \kappa(t)\,\Delta^{\rm Slat}(t),
	\quad
	\kappa(t) := i\lambda\,(U_1+U_2) - 2w(t)\lp\rho_1^\Psi(t)-\rho_2^\Psi(t)\rp .
\end{equation}
Note that $\Delta^{\rm Slat}$ obeys a linear equation whose inhomogeneity involves only $z_3$ and vanishes in the non-interacting case. For points (i) and (iii), we assume that $w\in L^1(0,T;\R)$ (and $w=0$ in (iii)), and we define
\begin{equation*}
	\Gamma^{\rm Slat}(t) := \Theta(t)\Delta^{\rm Slat}(t), \qquad \Theta(t) := \exp\lp \int_0^t \kappa(s) \rd s \rp.
\end{equation*}
Since $w\in L^1(0,T)$ and $|\rho_1^\Psi-\rho_2^\Psi|\le 2$, one has $\kappa\in L^1(0,T;\C)$, so that the function $\Theta$ is absolutely continuous and nowhere vanishing, with $|\Theta(t)|=\exp\lp -2\int_0^t w(\rho_1^\Psi-\rho_2^\Psi)\rd s\rp$. $\Gamma^{\rm Slat}$ is also absolutely continuous and \cref{eq:Delta-ode} then yields
\begin{equation}\label{eq:Delta-integrating-factor}
	\dot \Gamma^{\rm Slat}(t) = i\lambda\,(U_1+U_2)\,
	\Theta(t)z_3(t)^2\qquad \text{ for a.e. }\; t\in(0,T).
\end{equation}
The right-hand side being absolutely continuous, $\Gamma^{\rm Slat} \in W^{2,1}(0,T;\C) \hookrightarrow C^1([0,T];\C)$ and \eqref{eq:Delta-integrating-factor} holds for every $t\in[0,T]$. \medskip

\proofstep{Proof of (i).} Assume that $\Psi_0\in \cM^{\rm Slat}$ and $\lambda=0$. Then~\eqref{eq:Delta-ode} reduces to the linear equation $\dot\Delta^{\rm Slat}=-\kappa\,\Delta^{\rm Slat}$, whose coefficient $\kappa$ is integrable on $[0,T]$ since $w\in L^1(0,T;\R)$. By uniqueness of Carath\'eodory solutions, the only solution with $\Delta^{\rm Slat}(0)=0$ is $\Delta^{\rm Slat}=0$. Explicitly, $\Gamma^{\rm Slat}$ solves $\dot \Gamma^{\rm Slat} = 0$ with the initial condition $\Gamma^{\rm Slat}(0)=\Delta^{\rm Slat}(0)= 0$, so that $\Gamma^{\rm Slat}(t)=0$ for every $t\in[0,T]$ and, since $\Theta$ does not vanish, $\Delta^{\rm Slat}(t)=0$, which shows that $\cT^{\rm Slat}=[0,T]$.

\proofstep{Proof of (ii).} Assume that $\Psi_0\in \cM^{\rm Slat}$ and $\lambda=0$, so that $\dot\Delta^{\rm Slat}=2w(\rho_1^\Psi-\rho_2^\Psi)\Delta^{\rm Slat}$ a.e. by~\eqref{eq:Delta-ode}. In contrast with (i), the coefficient $w$ is only measurable, and the uniqueness argument of (i) only applies on the intervals on which $w$ is integrable, which is the case away from the poles in view of the first two equations in~\eqref{explicit_schrodinger} since $\Psi =(z_1,z_2,z_3)^T\in W^{1,1}(0,T;\C^3)$ by assumption. Let $t_1:=\sup\left\{t\in[0,T_*]\;\middle|\;\Delta^{\rm Slat}=0\text{ on }[0,t]\right\}$, and assume that $t_1<T_*$. By continuity, $\Psi(t_1)$ is a Slater state, which is not represented by a pole since $t_1<T_*$. As the eigenvectors of $V$ are the Slater states $e_1=(1,0,0)^T$ and $e_2=(0,1,0)^T$ represented by the poles and the vector $(0,0,1)^T \notin\cM^{\rm Slat}$, $\Psi(t_1)$ is not an eigenvector of $V$, so that $\|[V,P_{\Psi(t_1)}]\Psi(t_1)\|>0$. By continuity, there are $c,\delta>0$ such that $\|[V,P_\Psi]\Psi\|\ge c$ on $[t_1,t_1+\delta]$. Since $w\|[V,P_\Psi]\Psi\|\in L^1(0,T;\R)$ by~\eqref{eq:integrability-vw}, we get $w\in L^1(t_1,t_1+\delta)$, and the argument of (i) on $[t_1,t_1+\delta]$, with $\Delta^{\rm Slat}(t_1)=0$, gives $\Delta^{\rm Slat}=0$ on $[t_1,t_1+\delta]$, which contradicts the definition of $t_1$. Hence $t_1=T_*$, and $\Delta^{\rm Slat}=0$ on $[0,T_*]$ by continuity.

\proofstep{Proof of (iii).} Assume now that $\lambda>0$ and $w=0$, so that $\lambda(U_1+U_2)>0$ and $\Theta(t)=e^{i\lambda(U_1+U_2)t}$, and that $\Psi_0\in \cM^{\rm Slat}$. We show that for $t>0$ small enough $\Delta^{\rm Slat}(t)\neq 0$. We have $\Gamma^{\rm Slat}(0)= 0$ and $\dot \Gamma^{\rm Slat}(0)=i\lambda(U_1+U_2)z_3(0)^2$. If this last quantity is nonzero, the conclusion is immediate. Assume therefore $\dot \Gamma^{\rm Slat}(0)= 0$. Since $\lambda(U_1+U_2)>0$, this implies $z_3(0)=0$. Together with $\Delta^{\rm Slat}(0)=0$, we obtain $z_1(0)z_2(0)=0$. By symmetry we may assume $|z_1(0)| = 1$ and $z_2(0) = 0$. As $w=0$, the third equation of \eqref{explicit_schrodinger} reads $\dot z_3=i\sqrt2\,\tau(z_1+z_2)$, and $z_1+z_2=z_1(0)+o(1)$, so that
\begin{equation*}
	z_3(t) = i\sqrt{2}\,\tau z_1(0)\,t+o(t) \qquad\text{as }t\to0^+.
\end{equation*}
Integrating \eqref{eq:Delta-integrating-factor} from $0$ to $t$, and using $\Theta(t)=1+o(1)$, gives
\begin{equation*}
	\Gamma^{\rm Slat}(t) = i\lambda\,(U_1+U_2) \int_0^t \Theta(s)z_3(s)^2\,\rd s = -\f{2i}{3}\, \tau^2\lambda(U_1+U_2)z_1(0)^2\,t^3 +o(t^3).
\end{equation*}
Since $|\Theta(t)|=1$, it follows that
\begin{equation*}
	|\Delta^{\rm Slat}(t)| = \f{2}{3}\,\tau^2\lambda(U_1+U_2)\,t^3 +o(t^3).
\end{equation*}
Hence $\Delta^{\rm Slat}(t)\neq0$ for all sufficiently small
$t>0$. Translating the result to each time $t^*\in[0,T]$ such that $\Psi(t^*)\in \cM^{\rm Slat}$ (replacing $t$ by $t-t^*$ in the definition of $\Theta$), and arguing in the same way backward in time, we obtain that every point of $\cT^{\rm Slat}$ is isolated. As $\cT^{\rm Slat}$ is closed in the compact $[0,T]$, it is finite, which concludes the proof.

\subsection{Proof of \cref{prop:reg_density}}
\label{ssec:proof-reg-density}

\proofstep{Proof of (i).} Since $K+\lambda U+v(t)V$ is Hermitian and $\langle\Psi,[V,P_\Psi]\Psi\rangle=\|\Psi\|^2\langle\Psi,V\Psi\rangle-\langle\Psi,V\Psi\rangle\|\Psi\|^2=0$, we have $\f{\rd}{\rd t}\|\Psi\|^2=2\Re\langle\Psi,\dot\Psi\rangle=0$ a.e., so that $\|\Psi(t)\|=1$ for all $t\in[0,T]$, and therefore $\rho_1^\Psi+\rho_2^\Psi=2\|\Psi\|^2=2$. Point (i) then comes directly from the definitions \eqref{eq:density} and \eqref{eq:curent}, since absolutely continuous functions form an algebra.

\proofstep{Proof of (ii).} Since $\rho_1^\Psi=2|z_1|^2+|z_3|^2$, we obtain by using \eqref{explicit_schrodinger} a.e. on $(0,T)$,
\begin{equation*}
	\dot\rho_1^\Psi = 4\Re(\overline{z_1}\dot z_1)+2\Re(\overline{z_3}\dot z_3) = -j_1^\Psi - 8w\rho_2^\Psi|z_1|^2+2w(\rho_1^\Psi-\rho_2^\Psi)|z_3|^2 = -j_1^\Psi-4w(\rho_1^\Psi\rho_2^\Psi-|z_3|^2),
\end{equation*}
where we used $2|z_1|^2=\rho_1^\Psi-|z_3|^2$ and $\rho_1^\Psi+\rho_2^\Psi=2$. On a Slater state, $z_3^2=2z_1z_2$ gives $|z_3|^4=4|z_1|^2|z_2|^2$, so that $\rho_1^\Psi\rho_2^\Psi=4|z_1|^2|z_2|^2+2|z_3|^2\lp1-|z_3|^2\rp+|z_3|^4=2|z_3|^2$, whence the reduced form of \eqref{eq:continuity}.

\proofstep{Proof of (iii).} When $w=0$, $\dot\rho_1^\Psi=-j_1^\Psi\in W^{1,1}(0,T;\R)$, so that $\rho^\Psi\in W^{2,1}(0,T;\R^2)\hookrightarrow C^1([0,T];\R^2)$, and the continuity equation~\eqref{eq:continuity-w0}, whose two sides belong to $W^{1,1}(0,T;\R)$ and agree a.e., holds in $W^{1,1}(0,T;\R)$, hence everywhere on $[0,T]$. Finally, the van Leeuwen equation \eqref{eq:vanleeuwen} is obtained by differentiating $\dot\rho_1^\Psi=-j_1^\Psi$ and using \eqref{explicit_schrodinger} with $w=0$.

\subsection{Proof of \cref{thm:VP-characterization}}
\label{ssec:proof-VP}

We first state and prove two auxiliary lemmas. The first one, which is also used in \cref{sec:VP,sec:GP,sec:comparison} and in the proof of \cref{thm:GP-characterization}, shows that the non-interacting trajectories spend no time at the poles.

\begin{lemma}\label{lem:poles}
	Let $X\in W^{1,1}(0,T;\bS^2)$ be
	\begin{itemize}
		\item either a solution to \eqref{eq:vp_X} with $v\in L^1(0,T;\R)$;
		\item or a solution to \eqref{eq:gp_X} with $w\in L^0(0,T;\R)$ such that $w|e_3\times X|\in L^1(0,T;\R)$;
		\item or a solution to \eqref{eq:gpc_X} with $v_c\in L^1(0,T;\R)$ and $w_c\in L^0(0,T;\R)$ such that $w_c|e_3\times X|\in L^1(0,T;\R)$.
	\end{itemize}
	Then $|\cT_{\rho^X}|=0$.
\end{lemma}
\begin{proof}[Proof of \cref{lem:poles}]
	We have $\cT_{\rho^X}=\left\{t\in[0,T]\;\middle|\;X(t)\in\{P_+,P_-\}\right\}$. On $\cT_{\rho^X}$, one has $x_1=x_2=0$ and $x_3=\pm1$. Since $x_2\in W^{1,1}(0,T)$ vanishes on $\cT_{\rho^X}$, one has $\dot x_2=0$ a.e. on $\cT_{\rho^X}$. On the other hand, the second component of \eqref{eq:vp_X}, \eqref{eq:gp_X} and \eqref{eq:gpc_X} reads $\dot x_2=2vx_1+2\tau x_3+2wx_2x_3$ (with $w=0$ for VP and $v=0$ for GP), which gives $\dot x_2=2\tau x_3=\pm2\tau\neq0$ a.e. on $\cT_{\rho^X}$. Hence $|\cT_{\rho^X}|=0$.
\end{proof}

The second one provides the regularity of $\sqrt{\cD_\rho}$ needed for the reconstruction of the VP-trajectory.

\begin{lemma}
	\label{lem:sqrtD}
	Let $\rho\in W^{2,1}(0,T;\R^2)$ be such that $\rho_1+\rho_2=2$, $\cD_\rho\ge0$ and $\1_{\{\cD_\rho>0\}}\cN_\rho/\sqrt{\cD_\rho}$ belongs to $L^1(0,T;\R)$. Then $\sqrt{\cD_\rho}\in W^{1,1}(0,T;\R)$ and
	\begin{equation*}
		\f{\rd}{\rd t}\sqrt{\cD_\rho} = -\1_{\{\cD_\rho>0\}}\,\dot\rho_1\,\f{\cN_\rho}{\sqrt{\cD_\rho}} \qquad \text{a.e. on } (0,T).
	\end{equation*}
\end{lemma}

\begin{proof}
	Since $\rho_1+\rho_2=2$, one has $\dot\cD_\rho=4\tau^2\dot\rho_1(\rho_2-\rho_1)-2\dot\rho_1\ddot\rho_1=-2\dot\rho_1\cN_\rho$ a.e. Hence $f:=\sqrt{\cD_\rho}$ is continuous on $[0,T]$, vanishes on $\cS_\rho$, and is locally absolutely continuous on the relatively open set $\{\cD_\rho>0\}$, with derivative $g:=-\1_{\{\cD_\rho>0\}}\dot\rho_1\cN_\rho/\sqrt{\cD_\rho}$. As $\dot\rho_1$ is bounded, $g\in L^1(0,T)$. Let $s<t$. If $[s,t]\subset\{\cD_\rho>0\}$, then $|f(t)-f(s)|\le\int_s^t|g|$. Otherwise, denoting by $s'\le t'$ the smallest and largest points of $\cS_\rho\cap[s,t]$, $f(s')=f(t')=0$ gives $|f(t)-f(s)|\le\int_s^{s'}|g|+\int_{t'}^t|g|\le\int_s^t|g|$. Thus $f$ is absolutely continuous, that is $f\in W^{1,1}(0,T;\R)$, and $\dot f=g$ a.e., since $\dot f=0$ a.e. on $\cS_\rho$ because $f$ vanishes there.
\end{proof}

\begin{proof}[Proof of \cref{thm:VP-characterization}]
	The necessity of the conditions in \eqref{eq:RVP-characterization}, as well as the fact that any VP-trajectory reproducing $\rho$ is of the form \eqref{eq:VP-fiber-X} for some $\eta\in\fS_\rho(X_0)$ with potential given a.e. by \eqref{eq:VP-fiber}, have been established in \cref{sec:VP}. Conversely, let $\rho$ satisfy the conditions in \eqref{eq:RVP-characterization}, let $\eta\in\fS_\rho(X_0)$, and let $X_\eta:=X^{\rm VP}_\eta=(x_1,x_2,x_3)^T$ and $v_\eta$ be defined by \eqref{eq:VP-fiber-X} and \eqref{eq:VP-fiber}. By the last condition in \eqref{eq:RVP-characterization}, $v_\eta\in L^1(0,T;\R)$. By \cref{lem:sqrtD}, $\sqrt{\cD_\rho}\in W^{1,1}(0,T;\R)$. The function $f:=\sigma_\eta\sqrt{\cD_\rho}$ coincides with $\pm\sqrt{\cD_\rho}$ on each $I_k$ and vanishes on $\cS_\rho$, so that the argument of the proof of \cref{lem:sqrtD}, with $g$ replaced by $|\f{\rd}{\rd t}\sqrt{\cD_\rho}|$, shows that $f\in W^{1,1}(0,T;\R)$. Since $|X_\eta|^2=(\cD_\rho+\dot\rho_1^2)/(4\tau^2)+(\rho_1-1)^2=\rho_1\rho_2+(\rho_1-1)^2=1$, we get $X_\eta\in W^{1,1}(0,T;\bS^2)$, and $X_\eta(0)=X_0$ by the conditions on $\rho(0)$ and $\dot\rho_1(0)$ and the choice of $\eta_1$. The third equation of \eqref{eq:TDVP_coord} holds by construction. On $\Omega_\rho$, the second one is equivalent to $\cN_\rho=-4\tau v_\eta x_1$, which holds by definition of $v_\eta$, and the first one follows from the last two since, the vector field of~\eqref{eq:vp_X} being tangent to the sphere,
	\begin{equation*}
		x_1(\dot x_1+2v_\eta x_2)=-x_2(\dot x_2-2v_\eta x_1-2\tau x_3)-x_3(\dot x_3+2\tau x_2)
	\end{equation*}
	and $x_1\neq0$ on $\Omega_\rho$. On $\cS_\rho$, one has $x_1=0$ and $v_\eta=0$, so that the first equation reduces to $\dot x_1=0$, which holds a.e. on $\cS_\rho$ since $x_1$ vanishes there, and the second one reduces to $\ddot\rho_1=-4\tau^2(\rho_1-1)$, that is $\cN_\rho=0$, which holds a.e. on $\cS_\rho$ by assumption. Therefore $X_\eta$ is the solution of \eqref{eq:vp_X} with potential $v_\eta$, and $\rho\in\cR_{\rm VP}(X_0)$. Finally, $X_\eta\neq X_{\eta'}$ for $\eta\neq\eta'$, since $x_1$ has opposite signs on any $I_k$ with $\eta_k\neq\eta'_k$, where $\cD_\rho>0$.
\end{proof}

\subsection{Proof of \cref{thm:GP-characterization}}
\label{ssec:proof-GP}
If $\rho\in\cR_{\rm GPc}(X_0,v_c)$, then $\rho=\rho^X$ for some GPc-trajectory $X$, so that $|\cT_\rho|=0$ by \cref{lem:poles}. The necessity of the other conditions in \eqref{eq:RGP-characterization}, as well as the fact that any GPc-trajectory reproducing $\rho$ is of the form \eqref{eq:GP-fiber-X} for some $a\in\cA_\rho(X_0)$, with potential given a.e. by \eqref{eq:GP-fiber}, have been established in \cref{sec:GP}. Conversely, let $\rho$ satisfy the conditions in \eqref{eq:RGP-characterization}, let $a\in\cA_\rho(X_0)$, and let $X:=X^{\rm GPc}_{v_c,a}=(y_1,y_2,y_3)^T$ and $w_c:=w_{c,a}$ be defined by \eqref{eq:GP-fiber-X} and \eqref{eq:GP-fiber}. On each $J_\ell$, $X$ is parametrized by $(r_\rho,\alpha_a)$, so that $X$ is locally absolutely continuous on $\widetilde\Omega_\rho$. Let $E_j:=\dot y_j-F_j$, where $F_j$ denotes the right-hand side of the $j$-th equation in~\eqref{eq:TDGP_coord} with $w_c=w_{c,a}$. Since the vector field in~\eqref{eq:TDGP_coord} is tangent to the sphere and $X$ takes values in $\bS^2$, $y_1E_1+y_2E_2+y_3E_3=0$ a.e. on $\widetilde\Omega_\rho$. Moreover, $E_3=0$ by~\eqref{eq:GP-fiber}, and $y_1E_2-y_2E_1=r_\rho^2\dot\alpha_a-2v_cr_\rho^2-2\tau y_1y_3=0$ by~\eqref{eq:angle_gp}. As $r_\rho>0$ on $\widetilde\Omega_\rho$, we get $E_1=E_2=0$, so that $X$ solves \eqref{eq:TDGP_coord} a.e. on $\widetilde\Omega_\rho$, hence a.e. on $(0,T)$ since $|\cT_\rho|=0$. Moreover, $|w_c|r_\rho\le \f{|\dot\rho_1|}{2r_\rho}+\tau$ on $\widetilde\Omega_\rho$, so that $w_cr_\rho\in L^1(0,T;\R)$ by the last condition in \eqref{eq:RGP-characterization}. It remains to check that $X\in W^{1,1}(0,T;\bS^2)$ across the poles. Since $y_3=\rho_1-1\in W^{1,1}(0,T;\R)$, it suffices to consider $y_1$ and $y_2$. On $\widetilde\Omega_\rho$, $\dot r_\rho=-(\rho_1-1)\dot\rho_1/r_\rho$ and $r_\rho|\dot\alpha_a|\le2r_\rho|v_c|+2\tau$ by \eqref{eq:angle_gp}, so that $|\dot y_1|+|\dot y_2|\le g:=2\1_{\widetilde\Omega_\rho}|\dot\rho_1|/r_\rho+4|v_c|+4\tau$, and $g\in L^1(0,T)$ by the last condition in~\eqref{eq:RGP-characterization} and since $v_c\in L^1(0,T;\R)$. Let $f\in\{y_1,y_2\}$ and $s<t$. If $[s,t]\subset\widetilde\Omega_\rho$, then $|f(t)-f(s)|\le\int_s^tg$. Otherwise, denoting by $s'\le t'$ the smallest and largest points of $\cT_\rho\cap[s,t]$, $f(s')=f(t')=0$ gives $|f(t)-f(s)|\le\int_s^{s'}g+\int_{t'}^tg\le\int_s^tg$. Thus $y_1,y_2\in W^{1,1}(0,T;\R)$. Moreover, $X(0)=X_0$, by~\eqref{eq:GP-alpha-initial} if $0\in\widetilde\Omega_\rho$, and since $X_0=(0,0,\rho_1(0)-1)^T=X(0)$ if $0\in\cT_\rho$. Hence $(w_c,X)$ solves \eqref{prob:GPc}, so that $\rho\in\cR_{\rm GPc}(X_0,v_c)$. Since the conditions in \eqref{eq:RGP-characterization} do not depend on $v_c$, $\cR_{\rm GPc}(X_0,v_c)=\cR_{\rm GP}(X_0)$. Finally, $X^{\rm GPc}_{v_c,a}\neq X^{\rm GPc}_{v_c,a'}$ for $a\neq a'$, since $\alpha_a(s_\ell)\neq\alpha_{a'}(s_\ell)$ and $r_\rho(s_\ell)>0$ for any $\ell$ with $a_\ell\neq a'_\ell$.

\subsection{Proof of \cref{lem:Du}}
\label{ssec:proof-Du}
Let $\Psi=(z_1,z_2,z_3)^T$ be a normalized vector of $\C^3$. We set $\xi_1:=\overline{z_3}z_1$ and $\xi_2:=\overline{z_3}z_2$ in the following. One can successively obtain the following lower bound using~\eqref{eq:density} and~\eqref{eq:curent}:
	\begin{align*}
		\f{\cD(\Psi)}{4\tau^2}
		&= \rho_1^\Psi \rho_2^\Psi - 2\Im(\xi_2-\xi_1)^2 \\
		&= \rho_1^\Psi \rho_2^\Psi - 2|\xi_1-\xi_2|^2 + 2\Re(\xi_2-\xi_1)^2 \\
		&= \lp 2|z_1|^2+|z_3|^2\rp \lp 2|z_2|^2+|z_3|^2\rp
			- 2|z_3|^2|z_2-z_1|^2 + 2\Re(\xi_2-\xi_1)^2 \\
		&= |z_3|^4 + 4|z_1|^2|z_2|^2 + 4\Re(\xi_1\overline{\xi_2})
			+ 2\Re(\xi_2-\xi_1)^2 \\
		&= |z_3|^4 + 4|z_1|^2|z_2|^2 + 2\Re(\xi_1)^2 + 2\Re(\xi_2)^2
			+ 4\Re(\xi_1\overline{\xi_2}) - 4\Re(\xi_1)\Re(\xi_2) \\
		&= |z_3|^4 + 4|z_1|^2|z_2|^2 + 2\Re(\xi_1)^2 + 2\Re(\xi_2)^2
			+ 4\Im(\xi_1)\Im(\xi_2) \\
		&\ge |z_3|^4 + 4|z_1|^2|z_2|^2 + 2\Re(\xi_1)^2 + 2\Re(\xi_2)^2
			- 4|\xi_1||\xi_2| \\
		&= |z_3|^4 + 4|z_1|^2|z_2|^2 - 4|z_3|^2|z_1||z_2|
			+ 2\Re(\xi_1)^2 + 2\Re(\xi_2)^2 \\
		&= \lp |z_3|^2 - 2|z_1||z_2| \rp ^2
			+ 2\Re(\xi_1)^2 + 2\Re(\xi_2)^2 ,
	\end{align*}
whence~\eqref{eq:ineq_D_improved}, and in particular $\cD(\Psi)\ge0$. Equality enforces $\Re(\xi_1)=\Re(\xi_2)=0$, with $\Im(\xi_1)\Im(\xi_2)\le0$ and $|z_3|^2=2|z_1||z_2|$. If $z_3=0$, the last condition gives $z_1z_2=0$, that is $\Psi\in\{e^{i\vartheta}(1,0,0)^T,e^{i\vartheta}(0,1,0)^T\}$. If $z_3\neq0$, it gives $z_1,z_2\neq0$, and these conditions are equivalent to $\arg z_1=\arg z_3\pm\tfrac\pi2$ and $\arg z_2=\arg z_3\mp\tfrac\pi2$ together with the modulus condition. Parametrizing $|z_1|=y$ gives the form of singular states given in~\eqref{eq:singular-state}, the case $z_3=0$ corresponding to $y\in\{0,1\}$. Finally, for $\Psi\in\cK_s$, one has $z_3^2=-2y(1-y)e^{2i\vartheta}=2z_1z_2$, so that $\Psi\in\cM^{\rm Slat}$. Hence $\cK_s=\left\{\Phi_X\;\middle|\;X\in\bS^2,\ \cD(\Phi_X)=0\right\}$ up to a global phase, and~\eqref{eq:D-Bloch} gives $\cD(\Phi_X)=4\tau^2x_1^2$, so that the Bloch representation of $\cK_s$ is $\cC_s=\{x_1=0\}$.

\subsection{Proof of \cref{thm:int-representability}}
\label{ssec:proof-int-rep}
\proofstep{Singular sets.} Since $\cD^u(t)=\cD(\Psi^u(t))$, \cref{lem:Du} gives $\Psi^u(t)\in\cK_s\subset\cM^{\rm Slat}$ for every $t\in\cS^u$, so that $\cS^u$ is finite by point (iii) of \cref{prop:slater_set}. On $\cT^u$, $\cD^u=-(\dot\rho^u_1)^2\le0$, hence $\cD^u=0$ since $\cD^u\ge0$ by \cref{lem:Du}, so that $\cT^u\subseteq\cS^u$ is finite. As a consequence, $M_{\rm cc}$ and $N_{\rm cc}$ are finite.

\proofstep{Proof of (i).} Assume that $j^{X_0}=j^{\Psi_0^{\rm int}}$. By \cref{thm:NIVIN}~(i), $\rho^u\in\cR_{\rm VP}(X_0)$, so that the description of the VP-trajectories follows from \cref{thm:VP-characterization}. Let $t_*\in\cS^u$. By \cref{lem:Du}, $\Psi^u(t_*)=e^{i\vartheta}\Phi_Y$ for some $\vartheta\in\R$ and $Y\in\cC_s$. Then $\rho^Y=\rho^u(t_*)$ and $j^Y_1=j_1^{\Psi^u(t_*)}=-\dot\rho^u_1(t_*)$, so that $Y=\lp0,-\dot\rho^u_1(t_*)/(2\tau),\rho^u_1(t_*)-1\rp^T$ by~\eqref{eq:interpretation_x2} and~\eqref{eq:interpretation_x3}. Since $\cD^u(t_*)=0$, this is $X^{\rm VP}_\eta(t_*)$ by~\eqref{eq:VP-fiber-X}. Finally, by \eqref{eq:VP-fiber} and \eqref{eq:N-over-D-bound}, $|v_\eta|=|\cN^u|/(2\sqrt{\cD^u})\le\sqrt2\,|u|+\f12\sqrt{U_1^2+U_2^2}$ a.e. on $\{\cD^u>0\}$, and $v_\eta=0$ on $\cS^u$, whence \eqref{eq:bound-vp-int}.

\proofstep{Proof of (ii).} By \cref{thm:NIVIN}~(ii), $\rho^u\in\cR_{\rm GPc}(X_0,v_c)$, so that the description of the GPc-trajectories follows from \cref{thm:GP-characterization}. Let $t_z\in\cT^u$, for instance with $\rho^u_1(t_z)=0$. Then $2|z_1(t_z)|^2+|z_3(t_z)|^2=0$, so that $\Psi^u(t_z)=e^{i\vartheta}(0,1,0)^T=e^{i\vartheta'}\Phi_{P_-}$ for some $\vartheta,\vartheta'\in\R$, and $X^{\rm GPc}_{v_c,a}(t_z)=P_-$ by~\eqref{eq:GP-fiber-X}. The case $\rho^u_2(t_z)=0$ is similar. Finally, by \eqref{eq:GP-fiber}, $|w_{c,a}|\sqrt{\rho^u_1\rho^u_2}\le\f{|\dot\rho^u_1|}{2\sqrt{\rho^u_1\rho^u_2}}+\tau$ a.e. on $\{\rho^u_1\rho^u_2>0\}$, and $|\dot\rho^u_1|\le2\tau\sqrt{\rho^u_1\rho^u_2}$ since $\cD^u\ge0$ by \cref{lem:Du}, whence \eqref{eq:bound-wc-int}.

\subsection{Proof of \cref{lem:a-expansion}}
\label{ssec:proof-a-expansion}
On $I^-\cup I^+$, we have $\rho_1\le\f12$, so that $\beta<0$.

Since $\rho_1(t_z)=2|z_1(t_z)|^2+|z_3(t_z)|^2=0$ and $\|\Psi^u(t_z)\|=1$, we have $z_1(t_z)=z_3(t_z)=0$ and $|z_2(t_z)|=1$. By~\eqref{explicit_schrodinger} with $\lambda=1$, $v=u$ and $w=0$,
\begin{equation*}
	\dot z_1=-i\lp U_1+2u\rp z_1+i\sqrt2\,\tau z_3, \quad \dot z_2=-i\lp U_2-2u\rp z_2+i\sqrt2\,\tau z_3, \quad \dot z_3=i\sqrt2\,\tau\lp z_1+z_2\rp .
\end{equation*}

\proofstep{First bounds.} Since $|z_1|+|z_2|\le2$, the third equation gives $|\dot z_3|\le2\sqrt2\,\tau$, hence, as $z_3(t_z)=0$, $|z_3(t_z+s)|\le2\sqrt2\,\tau|s|$. As $z_1(t_z)=0$, the Duhamel formula for the first equation gives
\begin{equation*}
	z_1(t)=i\sqrt2\,\tau\int_{t_z}^{t}e^{-i\int_\sigma^t\lp U_1+2u\rp}z_3(\sigma)\,\rd\sigma,
\end{equation*}
so that $|z_1(t_z+s)|\le\sqrt2\,\tau\big|\int_{t_z}^{t_z+s}|z_3|\big|\le2\tau^2s^2$. Finally, the second equation gives $|\dot z_2|\le U_2+2|u|+\sqrt2\,\tau|z_3|$, so that, by~\eqref{eq:omega}, $|z_2(t_z+s)-z_2(t_z)|\le U_2|s|+2\,\omega(|s|)+2\tau^2s^2=O(\hat\omega(|s|))$.

\proofstep{Expansion of $z_3$ and $\rho_1$.} By the above bounds, $\dot z_3(t_z+\sigma)-i\sqrt2\,\tau z_2(t_z)=i\sqrt2\,\tau(z_1(t_z+\sigma)+z_2(t_z+\sigma)-z_2(t_z))=O(\hat\omega(|\sigma|))$. Integrating over $\sigma$ between $0$ and $s$, as $z_3(t_z)=0$, and using $\hat\omega(|\sigma|)\le\hat\omega(|s|)$ for $|\sigma|\le|s|$, we get $z_3(t_z+s)=i\sqrt2\,\tau z_2(t_z)\,s+O(|s|\,\hat\omega(|s|))$. Since $|z_2(t_z)|=1$, $|z_1(t_z+s)|^2=O(s^4)$ and $s^2\le\hat\omega(|s|)$, this gives $\rho_1(t_z+s)=2|z_1(t_z+s)|^2+|z_3(t_z+s)|^2=2\tau^2s^2(1+O(\hat\omega(|s|)))$.

\proofstep{Expansion of $\dot\rho_1$.} By \cref{prop:reg_density} and~\eqref{eq:curent}, $\dot\rho_1=-j_1^{\Psi^u}=-2\sqrt2\,\tau\,\Im\lp z_3(\overline{z_1}-\overline{z_2})\rp$. By the above expansions and $|z_2(t_z)|=1$,
\begin{align*}
	\lp z_3\lp\overline{z_1}-\overline{z_2}\rp\rp(t_z+s)
	&=\lp i\sqrt2\,\tau z_2(t_z)\,s+O\lp|s|\hat\omega(|s|)\rp\rp\lp-\overline{z_2(t_z)}+O\lp\hat\omega(|s|)\rp\rp \\
	&=-i\sqrt2\,\tau s+O\lp|s|\hat\omega(|s|)\rp ,
\end{align*}
so that $\dot\rho_1(t_z+s)=4\tau^2s\lp1+O(\hat\omega(|s|))\rp$, which is the first expansion in~\eqref{eq:rho1-refined}.

\proofstep{Expansions of $\sqrt{\rho_1\rho_2}$ and $\beta$.} Since $\rho_2(t_z+s)=2-\rho_1(t_z+s)=2\lp1+O(s^2)\rp$, the expansion of $\rho_1$ gives $\sqrt{\rho_1\rho_2}\,(t_z+s)=2\tau|s|\lp1+O(\hat\omega(|s|))\rp$, which is the second expansion in~\eqref{eq:rho1-refined}. Together with $1-\rho_1(t_z+s)=1+O(s^2)$, this gives $|\beta(t_z+s)|=2\tau(1-\rho_1)/\sqrt{\rho_1\rho_2}\,(t_z+s)=|s|^{-1}(1+O(\hat\omega(|s|)))$, which is~\eqref{eq:b-refined}.
\subsection{Proof of \cref{prop:dichotomy}}
\label{ssec:proof-dichotomy}
Recall that $\Lambda^+(t)=-\int_t^{t_z+\delta}|\beta|$ for $t\in I^+$ and $\Lambda^-(t)=\int_{t_z-\delta}^t|\beta|$ for $t\in I^-$ (see~\eqref{eq:Lambda}), so that $\dot\Lambda^\pm=|\beta|$.

\proofstep{Proof of (i).} The exceptional branch should behave like the solution $c_1=0$ of the case $v_c=0$, that is like $e^{\Lambda^+/2}(0,1)^T$. We therefore write $\chi=e^{\Lambda^+/2}f$, so that, using $\beta=-|\beta|$, \eqref{eq:linear} reads
\begin{equation}\label{eq:f-system}
	\dot f_1=-|\beta|\,f_1+v_c\,f_2, \qquad \dot f_2=-v_c\,f_1 .
\end{equation}
The exceptional branch is obtained by solving~\eqref{eq:f-system} from the singular time $t_z$ with the datum $f(t_z)=(0,1)^T$. Let $s_0\in(0,\delta]$ such that $\omega(s_0)\le\f14$, and set $I_0:=(t_z,t_z+s_0]$. On the space of bounded continuous functions $f:I_0\to\R^2$, endowed with the norm $\|f\|:=\sup_{I_0}\max(|f_1|,|f_2|)$, consider the map
\begin{equation*}
	(\cF f)_1(t):=\int_{t_z}^te^{-(\Lambda^+(t)-\Lambda^+(\sigma))}v_c(\sigma)f_2(\sigma)\,\rd\sigma, \qquad (\cF f)_2(t):=1-\int_{t_z}^tv_c\,f_1 .
\end{equation*}
Since $\Lambda^+$ is non-decreasing, the exponential is bounded by $1$, and $\cF$ is a contraction with constant $\omega(s_0)\le\f14$. Its unique fixed point $f^\star$ satisfies $\|f^\star\|\le\f43$, $|f^\star_1(t_z+s)|\le\f43\omega(s)$ and $|f^\star_2(t_z+s)-1|\le\f43\omega(s)\le\f13$, and solves~\eqref{eq:f-system} a.e. on $I_0$. Extending $f^\star$ to $I^+$ by solving~\eqref{eq:f-system}, we obtain a solution $\chi^\star:=e^{\Lambda^+/2}f^\star$ of~\eqref{eq:linear} on $I^+$ such that $f^\star(t)\to(0,1)^T$ as $t\to t_z^+$ and $|\zeta^\star|=|f^\star_1|/|f^\star_2|\le2\,\omega(s)$ on $I_0$. Hence $\phi^\star:=2\arctan \zeta^\star-\f\pi2$ satisfies $|\phi^\star(t_z+s)+\f\pi2|\le2|\zeta^\star|\le4\,\omega(s)$, so that $\phi^\star(t)\to-\f\pi2$ as $t\to t_z^+$, since $\omega(s)\to0$.

Let now $\chi=(p,q)^T$ be a solution of~\eqref{eq:linear} on $I^+$. For $v_c=0$, every solution with $c_1\neq0$ behaves like $e^{-\Lambda^+/2}(c_1,0)^T$. We therefore write $\chi=e^{-\Lambda^+/2}g$, so that \eqref{eq:linear} reads
\begin{equation}\label{eq:g-system}
	\dot g_1=v_c\,g_2, \qquad \dot g_2=-v_c\,g_1+|\beta|\,g_2 .
\end{equation}
Integrating~\eqref{eq:g-system} backward from $t_+:=t_z+\delta$, where $\Lambda^+(t_+)=0$, gives
\begin{equation*}
	g_2(t)=e^{\Lambda^+(t)}g_2(t_+)+\int_t^{t_+}e^{\Lambda^+(t)-\Lambda^+(\sigma)}v_c(\sigma)g_1(\sigma)\,\rd\sigma, \qquad g_1(t)=g_1(t_+)-\int_t^{t_+}v_c\,g_2 .
\end{equation*}
The exponential being bounded by $1$, Grönwall's lemma shows that $g$ is bounded on $I^+$. As $v_c$ is integrable and the exponential and $g$ are bounded, these integral equations show that $g\in W^{1,1}(I^+)$. Hence $g$ is continuous up to $t_z$, with $g_1(t)\to\gamma$ for some $\gamma\in\R$ and, by dominated convergence, $g_2(t)\to0$ as $t\to t_z^+$, since $\Lambda^+(t)\to-\infty$. As $\zeta=p/q=g_1/g_2$, we get $\zeta\to\infty$, that is $\phi\to\f\pi2$, as soon as $\gamma\neq0$, so that it remains to determine the solutions for which $\gamma=0$. To identify $\gamma$, we use the Wronskian $W(\chi,\widetilde\chi):=p\widetilde q-q\widetilde p$ of two solutions $\chi=(p,q)^T$ and $\widetilde\chi=(\widetilde p,\widetilde q)^T$ of~\eqref{eq:linear}, which satisfies
\begin{equation*}
	\f{\rd}{\rd t}W(\chi,\widetilde\chi)=\lp\operatorname{tr}M\rp W(\chi,\widetilde\chi)=0 \qquad\text{a.e.},
\end{equation*}
since $M$ is traceless, so that $W(\chi,\widetilde\chi)$ is constant in time, and vanishes if and only if $\chi$ and $\widetilde\chi$ are collinear (that is, if they define the same solution $\zeta$ of~\eqref{eq:riccati}, hence the same longitude $\phi$). As $W(\chi,\chi^\star)=g_1f^\star_2-g_2f^\star_1$ tends to $\gamma$, we get $\gamma=W(\chi,\chi^\star)$. Hence $\gamma=0$ if and only if $\chi$ is collinear to $\chi^\star$, and every other solution satisfies $\phi\to\f\pi2$. Hence $\phi^\star$ is the only solution of~\eqref{eq:theta-ode} on $I^+$ tending to $-\f\pi2$, and every other solution tends to $\f\pi2$.

\proofstep{Proof of (ii).} The proof is similar, on $I^-$ and with the roles of $p$ and $q$ exchanged.
\subsection{Proof of \cref{thm:integrability}}
\label{ssec:proof-integrability}
The expansion~\eqref{eq:w-expansion} is~\eqref{eq:bracket}.

\proofstep{Proof of (i).} If $\phi$ is not the exceptional branch, then $|\sgn(s)+\sin\phi(t_z+s)|\to2$ by \cref{prop:dichotomy}. As $\hat\omega(|s|)\to0$, \eqref{eq:bracket} gives $|w_{c,a}(t_z+s)|=|s|^{-1}\lp1+o(1)\rp$, which is not integrable near $s=0$, and, by~\eqref{eq:rho1-refined}, $\sqrt{\rho_1\rho_2}\,|w_{c,a}|(t_z+s)\to2\tau$.

\proofstep{Proof of (ii).} If $\phi$ is the exceptional branch, then $|\sgn(s)+\sin\phi|=1-\cos\lp\phi+\sgn(s)\f\pi2\rp\le\f12\lp\phi+\sgn(s)\f\pi2\rp^2\le8\,\omega(|s|)^2$ by \cref{prop:dichotomy}, so that $w_{c,a}(t_z+s)=O\lp\hat\omega(|s|)/|s|\rp$ by~\eqref{eq:bracket}. Since $\int_0^\delta\hat\omega(s)/s\,\rd s=\delta+\int_0^\delta\omega(s)/s\,\rd s$, $w_{c,a}$ is integrable on $I^+$ (resp. $I^-$) under~\eqref{eq:dini}.
\subsection{Proof of \cref{cor:integrability}}
\label{ssec:proof-cor-integrability}
The function $w_{c,a}$ is bounded on every compact subset of $\widetilde\Omega_\rho$ by~\eqref{eq:bound-wc-int}, and $\cT^u$ is finite. Hence $w_{c,a}\in L^1(0,T;\R)$ if and only if $w_{c,a}$ is integrable on $I^-$ and $I^+$ at every $t_z\in\cT^u$, and the first assertion follows from \cref{thm:integrability}. If moreover $u,v_c\in L^\infty(0,T;\R)$, then $\omega(s)=O(s)$, so that, along the exceptional branches, $w_{c,a}(t_z+s)=O\lp\hat\omega(|s|)/|s|\rp=O(1)$ by \cref{thm:integrability}~(ii), and $w_{c,a}\in L^\infty(0,T;\R)$. Assume now that $0\notin\cT^u$, let $t_1:=\min\cT^u\in(0,T)$, and let $X_0'\in\bS^2$ with $\rho^{X_0'}=\rho(0)$ and $a\in\cA_\rho(X_0')$ be such that $w_{c,a}\in L^1(0,T;\R)$. By \cref{thm:integrability}~(i), the longitude of $X^{\rm GPc}_{v_c,a}$ must coincide on $J_1\cap I^-$, where $J_1=[0,t_1)$, with the exceptional incoming branch $\phi_\star$ at $t_1$, which is unique by \cref{prop:dichotomy}. By \cref{thm:GP-characterization}, this longitude is $\alpha_a+\f\pi2$ on $J_1$, where $\alpha_a$ is the solution of~\eqref{eq:angle_gp} with $\alpha_a(0)=a_0$, and, by uniqueness for~\eqref{eq:angle_gp}, two different values of $a_0$ give two different solutions on $J_1\cap I^-$. Hence there is at most one value of $a_0$ for which this holds, and, since $X_0'$ is determined by $a_0$ through~\eqref{eq:GP-alpha-initial}, at most one such initial condition $X_0'\in \bS^2$.

\section*{Acknowledgment} This work has benefited from French State support managed by ANR under the France 2030 program through the MaQuI CNRS Risky and High-Impact Research program (RI)2 (grant agreement ANR-24-RRII-0001).
We thank Jari van Gog, Asbj\o rn B\ae kgaard Lauritsen, Mathieu Lewin, Jakob Oldenburg and Julien Toulouse for stimulating discussions within this project.

\section*{AI Statement}
Large Language Models (LLMs) were used in the final stage of preparing the manuscript for proofreading. Claude Code was also used to write the Julia code used to generate the figures.

\printbibliography

\appendix

\section{Derivation of the model}\label{appendix:derivation}

The Hubbard dimer model describes a quantum system with two sites occupied by spin-$1/2$ fermions (recall that electrons are spin-$1/2$ fermions).

\medskip

We denote by $\mathcal H_1$ the one-particle state space of the Hubbard dimer, that is the four-dimensional complex vector space
\begin{equation*}
\mathcal H_1 = \C \; |\!\uparrow \! 0\rangle \oplus \C \; |\!\downarrow \! 0\rangle \oplus \C \; |0\!\uparrow  \rangle \oplus \C \; |0 \!\downarrow \rangle,
\end{equation*}
where the four orthonormal vectors $|\!\uparrow \! 0\rangle$, $|\!\downarrow \! 0\rangle$, $|0\!\uparrow  \rangle$ and $|0 \!\downarrow \rangle$ respectively represent one-particle states in which the particle is on the first site with spin up, on the first site with spin down, on the second site with spin up, and on the second site with spin down. Since two fermions cannot occupy the same quantum state by virtue of the Pauli principle, the Hubbard dimer can host between $0$ and $4$ particles. The zero-particle state space is
\begin{equation*}
\mathcal H_0=\C |{\rm vac}\rangle,
\end{equation*}
where $|{\rm vac}\rangle$ is the vacuum (a normalized quantum state with no particle), the $N$-particle state spaces, $2 \le N \le 4$, are given by
\begin{equation*}
\mathcal H_N=\bigwedge^N \mathcal H_1
\end{equation*}
(antisymmetrized tensor product of $N$ copies of the one-particle state space), and the Fock space is defined as
\begin{equation*}
\mathcal F := \mathcal H_0 \oplus  \mathcal H_1 \oplus \mathcal H_2 \oplus \mathcal H_3 \oplus \mathcal H_4.
\end{equation*}
We have $\dim(\mathcal H_N) = \begin{pmatrix} 4 \\ N \end{pmatrix}$ and $\dim(\mathcal F)=2^4$.

\medskip

The creation and annihilation operators $\widehat a_{i\sigma}^\dagger$ and $\widehat a_{i\sigma}$ are the linear operators on the Fock space $\mathcal F$ describing the creation and annihilation of a particle at site $i \in \{1,2\}$ with spin $\sigma \in \{\uparrow,\downarrow\}$. They satisfy the canonical anticommutation relations (CAR)
\begin{equation*}
\{\widehat a_{i\sigma}^\dagger,\widehat a_{j\sigma'}^\dagger\}=0, \qquad \{\widehat a_{i\sigma},\widehat a_{j\sigma'}\}=0, \qquad \{\widehat a_{i\sigma}^\dagger,\widehat a_{j\sigma'}\}=\delta_{ij}\delta_{\sigma\sigma'},
\end{equation*}
and are such that
\begin{equation*}
\widehat a_{1\uparrow}^\dagger |{\rm vac}\rangle=|\!\uparrow \! 0\rangle, \qquad \widehat a_{1\downarrow}^\dagger |{\rm vac}\rangle=|\!\downarrow \! 0\rangle, \qquad \widehat a_{2\uparrow}^\dagger |{\rm vac}\rangle=|0\!\uparrow  \rangle, \qquad  \widehat a_{2\downarrow}^\dagger |{\rm vac}\rangle=|0 \!\downarrow \rangle.
\end{equation*}
The spin-density operators $\widehat n_{j\sigma}$, the density operators $\widehat n_{j}$, and the number operator $\widehat N$ are the linear operators on $\mathcal F$ respectively defined by
$\widehat n_{j\sigma}:=\widehat a_{j\sigma}^\dagger\widehat a_{j\sigma}$, $\widehat n_{j}=\widehat n_{j\uparrow}+\widehat n_{j\downarrow}$, and $\widehat N =\widehat n_{1}+\widehat n_{2}$. The $N$-particle sector is $\widehat N$-invariant and it holds $\widehat N|_{\mathcal H_N}= N I_{\mathcal H_N}$. We finally introduce the spin operators $\widehat S_z$ (total $z$-component), $\widehat S_+$ (total spin raising), $\widehat S_-$ (total spin lowering), and $\widehat S^2$ (total spin square) defined as the linear operators on $\mathcal F$ given by
\begin{equation*}
\widehat S_z=\frac 12 \sum_{j=1}^2 (\widehat n_{j\uparrow} - \widehat n_{j\downarrow}), \enspace
\widehat S_+ = \sum_{j=1}^2 \widehat a_{j\uparrow}^\dagger \widehat a_{j\downarrow}, \enspace
\widehat S_- = \sum_{j=1}^2 \widehat a_{j\downarrow}^\dagger \widehat a_{j\uparrow}, \enspace
\widehat S^2= \widehat S_z^2+\frac 12 \left(\widehat S_+\widehat S_-+\widehat S_-\widehat S_+ \right).
\end{equation*}
Note that
\begin{equation*}
[\widehat N,\widehat S^2]=0, \quad [\widehat N,\widehat S_z]=0, \quad [\widehat S^2,\widehat S_z]=0.
\end{equation*}

The two-site Hubbard Hamiltonian with time-dependent external potential $(u_1(t),u_2(t))$ is the self-adjoint operator on $\mathcal F$ defined by
\begin{equation*}
\widehat H(t) = - \tau \sum_{\sigma \in \{\uparrow,\downarrow\}} (\widehat a_{1\sigma}^\dagger \widehat a_{2\sigma} + \widehat a_{2\sigma}^\dagger \widehat a_{1\sigma}) + u_1(t) \widehat n_1 + u_2(t) \widehat n_2 + U_1 \widehat n_{1\uparrow} \widehat n_{1\downarrow}+ U_2 \widehat n_{2\uparrow} \widehat n_{2\downarrow},
\end{equation*}
where $\tau > 0$ is the hopping parameter describing quantum tunneling between the two sites, $u_j(t) \in \R$ the external potential on site $j$, and $U_j$ the on-site interaction parameter on site~$j$.

\medskip

The Hamiltonian $\widehat H(t)$ conserves both particle number and spin, in the sense that
\begin{equation*}
[\widehat H(t),\widehat N]=0, \qquad [\widehat H(t),\widehat S^2]=0, \qquad [\widehat H(t),\widehat S_z]=0.
\end{equation*}
It follows that the two-particle spin-singlet sector of the Fock space defined as
\begin{equation*}
\mathcal H_{2,0,0}:= \left\{ \Psi \in \mathcal F \; | \; \widehat N \Psi=2\Psi, \; \widehat S^2\Psi=0, \; \widehat S_z\Psi=0 \right\}
\end{equation*}
(the second condition actually implies the third one) is $\widehat H(t)$-invariant. A convenient orthonormal basis of $\mathcal H_{2,0,0}$ is
\begin{equation*}
	|20\rangle := \widehat a_{1\uparrow}^\dagger \widehat a_{1\downarrow}^\dagger |{\rm vac}\rangle, \quad
|02\rangle := \widehat a_{2\uparrow}^\dagger \widehat a_{2\downarrow}^\dagger |{\rm vac}\rangle, \quad
|11\rangle := \frac{\widehat a_{1\uparrow}^\dagger \widehat a_{2\downarrow}^\dagger - \widehat a_{1\downarrow}^\dagger\widehat a_{2\uparrow}^\dagger }{\sqrt 2} |{\rm vac}\rangle.
\end{equation*}
In state $|20\rangle$ (resp. $|02\rangle$), the two particles are on site $1$ (resp. $2$) with opposite spins, while in state $|11\rangle$ there is one particle on each site and the spin state is chosen in such a way that the total spin is zero.

\medskip

As $\mathcal H_{2,0,0}$ is $\widehat H(t)$-invariant for all $t$, it is also invariant by the Schr\"odinger dynamics
\begin{equation}\label{eq:TDSE_Fock}
i \dot\Psi(t) = \widehat  H(t) \Psi(t).
\end{equation}
Restricting ourselves to $\mathcal H_{2,0,0}$, the $16$-dimensional Schr\"odinger dynamics~\eqref{eq:TDSE_Fock} reduces to a $3$-dimensional Schr\"odinger dynamics, which can be written in the basis $(|20\rangle,|02\rangle,|11\rangle)$ as
\begin{equation*}
i \dot\Psi_{2,0,0}(t) = H(t) \Psi_{2,0,0}(t),
\end{equation*}
where the matrix $H(t)$ (representation in the above basis of the restriction of $\widehat H(t)$ to the invariant subspace $\mathcal H_{2,0,0}$) is given by
\begin{equation*}
H(t) := \begin{pmatrix} 2u_1(t)+U_1 & 0 & -\sqrt 2 \tau \\ 0 & 2u_2(t)+U_2 & -\sqrt 2 \tau \\  -\sqrt 2 \tau  &  -\sqrt 2 \tau  & u_1(t)+u_2(t) \end{pmatrix}.
\end{equation*}
Recall that replacing $H(t)$ by $H(t)-E(t)I_3$ amounts to performing a gauge transform allowed by the fundamental laws of quantum mechanics. Instead of working with $H(t)$, we can therefore work with $H(t)-(u_1(t)+u_2(t))I_3$, that is with
\begin{equation*}
H_u(t) := \begin{pmatrix} 2u(t)+U_1 & 0 & -\sqrt 2 \tau \\ 0 & -2u(t)+U_2 & -\sqrt 2 \tau \\  -\sqrt 2 \tau  &  -\sqrt 2 \tau  & 0 \end{pmatrix},
\end{equation*}
where $u(t):=(u_1(t)-u_2(t))/2$ is half the potential offset between the two sites. With this gauge transform, the restriction of the dynamics~\eqref{eq:TDSE_Fock} to the invariant subspace $\mathcal H_{2,0,0}$ reads as~\eqref{eq:interacting_dynamics}.

\medskip

We have just shown that the quantum dynamics for the Hubbard dimer in the two-particle spin-singlet sector reduces to a time-dependent Schr\"odinger equation on $\C^3$. Let us now proceed to the description of the two-particle spin-singlet Slater states. In the general case, a two-particle spin-singlet Slater state is described by a single one-particle orbital, that is a normalized vector $\varphi$ of the spinless one-particle state space, representing the state of the electron pair composed of a spin-up particle and a spin-down particle. In the case of the Hubbard dimer, the spinless one-particle state space is $\C^2$ and a normalized orbital $\varphi = (\varphi_1,\varphi_2)^T \in \C^2$ represents a state in which the density at site $j$ is $2|\varphi_j|^2$. More precisely, the wavefunction $\Phi$ associated to $\varphi$ is
\begin{equation}\label{eq:Phi}
\Phi = \left( \varphi_1 \widehat a_{1\uparrow}^\dagger + \varphi_2 \widehat a_{2\uparrow}^\dagger \right) \left( \varphi_1 \widehat a_{1\downarrow}^\dagger + \varphi_2 \widehat a_{2\downarrow}^\dagger \right) |{\rm vac}\rangle
= \varphi_1^2 |20\rangle + \varphi_2^2 |02\rangle + \sqrt 2 \varphi_1 \varphi_2 |11\rangle,
\end{equation}
where we used the anticommutation relation $\{\widehat a_{2\uparrow}^\dagger,\widehat a_{1\downarrow}^\dagger\}=0$ to get the last term. Taking into account the gauge invariance $\Psi \mapsto e^{i\vartheta}\Psi$, we can always assume that $\varphi_1 \ge 0$, and rewrite $\varphi$ as
\begin{equation}\label{eq:varphi}
\varphi = \begin{pmatrix} \cos \frac \theta 2 \\ e^{i\phi} \sin \frac \theta 2 \end{pmatrix},
\end{equation}
with $\theta \in [0,\pi]$ and $\phi \in \R/(2\pi\Z)$. The writing is unique if $0 < \varphi_1 < 1$ that is if $\theta \in (0,\pi)$. In other words, the set of two-particle spin-singlet Slater states of the Hubbard dimer can be identified with the sphere $\bS^2$, called in this context the Bloch sphere. From a more abstract viewpoint, this result says that the Grassmann manifold ${\rm Gr}_{\C}(1,2)$ can be identified with the sphere $\bS^2$, a well-known result in differential geometry. In summary, if $X = \begin{pmatrix} \sin\theta \cos\phi \\  \sin\theta \sin\phi \\ \cos\theta \end{pmatrix} \in \bS^2$ is a point of the Bloch sphere, the representation in the basis $(|20\rangle,|02\rangle,|11\rangle)$ of the unique two-particle spin-singlet Slater state associated to $X$ via~\eqref{eq:Phi}-\eqref{eq:varphi} is given by
\begin{equation*}
\Phi_X = \begin{pmatrix} \frac{1+\cos\theta}2 \\ \frac{1-\cos\theta}2 e^{2i\phi} \\  \frac{\sin\theta}{\sqrt 2} e^{i\phi}\end{pmatrix}.
\end{equation*}
Let us finally write down the GP Kohn--Sham dynamics introduced in~\cite{geoPRA}. It is convenient to write this dynamics in terms of one-body reduced density matrices (1-RDM). The 1-RDM associated with the two-particle spin-singlet Slater state $\Phi_X$ is the matrix
\begin{equation*}
\gamma_X = \begin{pmatrix} \cos \frac \theta 2 \\ e^{i\phi} \sin \frac \theta 2 \end{pmatrix} \begin{pmatrix} \cos \frac \theta 2 \\ e^{i\phi} \sin \frac \theta 2 \end{pmatrix}^*
= \begin{pmatrix} \cos^2\left( \frac \theta 2 \right) &   e^{-i\phi} \sin \frac \theta 2  \cos \frac \theta 2 \\ e^{i\phi} \sin \frac \theta 2  \cos \frac \theta 2  & \sin^2\left( \frac \theta 2 \right) \end{pmatrix}.
\end{equation*}
The above expression can be rewritten in terms of the Cartesian coordinates of $X$ as
\begin{equation*}
\forall X = \begin{pmatrix} x_1 \\ x_2 \\ x_3 \end{pmatrix} \in \bS^2, \qquad \gamma_X=
 \begin{pmatrix} \frac{1+x_3}2 & \frac{x_1-ix_2}2 \\ \frac{x_1+ix_2}2  & \frac{1-x_3}2 \end{pmatrix} \in {\rm Gr}_{\C}(1,2).
\end{equation*}
The GP Kohn--Sham dynamics then reads, up to the sign convention for $w$ (see~\cite{geoPRA}),
as
\begin{equation}\label{eq:NIdyn-W}
i \dot \gamma(t) =  \left[ {\mathcal K} + v(t) \mathcal V - i [\mathcal W(t), \gamma(t)], \gamma(t)\right] , \quad \gamma(0)=\gamma_{X_0},
\end{equation}
with
\begin{equation*}
{\mathcal K} := \begin{pmatrix} 0 & -\tau \\ - \tau & 0 \end{pmatrix}, \qquad {\mathcal V} := \begin{pmatrix} 1 & 0 \\ 0 & -1 \end{pmatrix}, \qquad {\mathcal W}(t) := \begin{pmatrix} 2w_1(t) & 0 \\ 0 & 2w_2(t) \end{pmatrix},
\end{equation*}
where $w_j(t)$ is the geometric potential at site $j$. Since $\mathcal W(t)$ is only involved through the commutator $[\mathcal W(t),\gamma(t)]$, we can replace $\mathcal W(t)$ with the traceless matrix
\begin{equation*}
\widetilde{\mathcal W}(t) :=  \mathcal W(t)-\frac{{\rm Tr}(\mathcal W(t))}2 I_2
= w(t) \mathcal V , \quad \mbox{with} \quad w(t):=w_1(t)-w_2(t).
\end{equation*}
Finally, the dynamics~\eqref{eq:NIdyn-W} is equivalent to
\begin{equation}\label{eq:NIdyn}
i \dot \gamma(t) =  \left[ {\mathcal K} + v(t) \mathcal V - i w(t) [\mathcal V, \gamma(t)], \gamma(t)\right] , \quad \gamma(0)=\gamma_{X_0}.
\end{equation}
Since $\gamma_{X_0} \in {\rm Gr}_{\C}(1,2)$ and ${\mathcal K} + v(t) \mathcal V - i w(t) [\mathcal V, \gamma(t)]$ is Hermitian, the solution to~\eqref{eq:NIdyn} at time $t$ belongs to
${\rm Gr}_{\C}(1,2)$ at any time $t$, and can thus be written as $\gamma(t)=\gamma_{X(t)}$ with $X(t)\in\bS^2$. We then have
\begin{equation*}
i \begin{pmatrix} \frac{\dot x_3}2 & \frac{\dot x_1-i\dot x_2}2 \\ \frac{\dot x_1+i\dot x_2}2  & -\frac{\dot x_3}2 \end{pmatrix} =
\left[ \begin{pmatrix} v & -\tau - w(x_2+ix_1) \\ - \tau-w (x_2-ix_1)  & -v \end{pmatrix},  \begin{pmatrix} \frac{1+x_3}2 & \frac{x_1-ix_2}2 \\ \frac{x_1+ix_2}2  & \frac{1-x_3}2 \end{pmatrix} \right],
\end{equation*}
or equivalently
\begin{align*}
 & \dot x_1 = -2vx_2 + 2 w x_1x_3, \\
 & \dot x_2 = 2v x_1 +2 \tau x_3 + 2w x_2x_3, \\
   & \dot x_3 = -2 \tau x_2 - 2 w( x_1^2+x_2^2).
\end{align*}
The above dynamics can be rewritten in the more intrinsic form
\begin{equation*}
\dot X(t) = (-2\tau e_1 + 2v(t) e_3 + 2w(t)  (e_3 \times X(t))) \times X(t),
\end{equation*}
which is~\eqref{eq:gpc_X} with $v_c=v$ and $w_c=w$, and~\eqref{eq:gp_X} for $v=0$.

\section{Controllability of the interacting Hubbard dimer}
\label{appendix:controllability}

For completeness, we provide a proof that the interacting Hubbard dimer is pure-state controllable, which implies in particular that there exists $T > 0$ such that for any normalized initial state $\Psi_0^{\rm int} \in \C^3$ and any normalized state $\Psi_*$ belonging to the set of singular states for either VP or GP, there exist $0<t_* \le T$ and $u \in L^1(0,T;\R)$ (we could actually require $u \in C^\infty_{\rm c}(0,T;\R)$) such that $\Psi^u(t_*)=\Psi_*$. Recall that $\Psi_*$ belongs to the set of singular states for VP (resp. GP) if $\Psi_*$ is a Slater state and if the representative of $\Psi_*$ on the Bloch sphere is on the great circle $\cC_s$, that is $\Psi_*\in\cK_s$ by \cref{lem:Du} (resp. is one of the poles $P_\pm$, that is $\Psi_*=(1,0,0)^T$ or $\Psi_*=(0,1,0)^T$ up to a global phase). The proof only uses the fact that the interaction matrix $U$ has a positive trace, so that we state the result for the more general dynamics~\eqref{eq:schrodinger} with $\lambda>0$, $v=u$ and $w=0$, the interacting dynamics~\eqref{eq:interacting_dynamics} corresponding to $\lambda=1$.

\medskip

\begin{proposition}\label{prop:control}
	Let $\lambda>0$ and $M>0$. There exists $T>0$ such that for all normalized states $\Psi_0,\Psi_1\in \C^3$, there exist $0<t_* \le T$ and $u\in \cU_{\rm bb}^M(0,T)$, where $\cU_{\rm bb}^M(0,T)$ is the class of \textit{bang--bang} functions defined by
	\begin{align*}
		\cU^M_{\rm bb}(0,T) := & \Bigg\{ u=\sum_{i=0}^{N}u_{i}\1_{[t_i,t_{i+1})}:[0,T]\to \R, \quad N\in \N, \quad u_0,\dots,u_N \in\{-M,M\},\\ & \qquad \quad 0=t_0<t_1<\dots<t_N<t_{N+1}=T\Bigg\},
	\end{align*}
	such that the unique solution $\Psi\in W^{1,1}(0,T;\C^3)$ to~\eqref{eq:schrodinger} with $v=u$ and $w=0$ satisfies $\Psi(0)=\Psi_0$ and $\Psi(t_*)=\Psi_1$.
\end{proposition}

\begin{proof}
Following the terminology of \cite{albertiniNotionsControllabilityBilinear2003}, \cref{prop:control} is concerned with the pure-state controllability of the system, that is a controllability result at the level of the wavefunction $\Psi(t)$. We can in fact prove a stronger result, namely the operator-controllability of the system, a property which can be formulated as follows. Consider the propagator $\cU^u(t)$ associated with the time-dependent Hamiltonian $H^\lambda_{u,0}(t)=K+\lambda U+u(t)V$, defined by
\begin{equation*}
	\begin{cases}
		i\,\dot \cU^u(t)  = (K+\lambda U+ u(t) V) \,\cU^u(t), & t \in (0,T], \\
		\cU^u(0) = I_3.
	\end{cases}
\end{equation*}
We have that for all $t \in [0,T]$, $\Psi(t)=\cU^u(t)\Psi(0)$. The dynamics defined by $H^\lambda_{u,0}(t)$ is called operator-controllable if there exists $T > 0$ such that for all $U_* \in U(3)$, there exist $0<t_* \le T$ and $u \in L^1(0,T;\R)$ such that $\cU^u(t_*)=U_*$. If we can prove operator-controllability for $u$ in the class $\cU_{\rm bb}^M(0,T)$, then the pure-state controllability result in \cref{prop:control} readily follows. In fact, by \cite[Thm~4]{albertiniNotionsControllabilityBilinear2003}, operator-controllability and pure-state controllability are actually equivalent in our case since the dimension $d=3$ of the system under consideration is odd.

By the standard results in \cite[Thm~7.1]{jurdjevicControlSystemsLie1972}, it follows from the fact that the Lie group $U(3)$ is connected and compact that the system is operator-controllable if and only if the Lie algebra $\fL:=\textrm{Lie}(-i(K+\lambda U),-iV)$ generated by $-i(K+\lambda U)$ and $-iV$ is equal to the full Lie algebra $\mathfrak{u}(3)$. An upper bound of the minimal time $T>0$ for which the system is operator-controllable is then given by \cite[Thm~7.2]{jurdjevicControlSystemsLie1972}. Since
\begin{equation*}
	[-iV,[-i(K+\lambda U),-iV]] = i[V,[K,V]] = -4iK,
\end{equation*}
we first obtain that $-iK\in\fL$, and therefore, since $\lambda\neq0$,
\begin{equation*}
	\fL=\textrm{Lie}(-iK,-iU,-iV).
\end{equation*}

Checking that $\fL = \mathfrak{u}(3)$ can be done by computing several other Lie brackets of $3\times 3$ matrices. To keep computations to a minimum, we use the two following decompositions
\begin{equation} \label{eq:dec_u(3)}
	\mathfrak{u}(3) = i\R I_3 \oplus\mathfrak{su}(3) \quad \text{ and } \quad  \mathfrak{su}(3) = \mathfrak{so}(3) \oplus i \textrm{Sym}_0(\R^3)
\end{equation}
where $\textrm{Sym}_0(\R^3)$ is the space of the $3 \times 3$ real symmetric matrices with zero trace. These decompositions come from the fact that any matrix $M \in \mathfrak{u}(3)$ is of the form
\begin{equation*}
	M = \begin{pmatrix} ib_{11} & a_{12}+ib_{12} & a_{13}+ib_{13} \\
	-a_{12}+ib_{12}  & ib_{22} & a_{23}+ib_{23} \\ -a_{13}+ib_{13} &
	-a_{23}+ib_{23} & ib_{33} \end{pmatrix}\quad \text{with} \quad  a_{ij}, b_{ij} \in \R,
\end{equation*}
and can thus be decomposed as
\begin{equation*}
	\setlength{\arraycolsep}{3.5pt}
	M = i \f{b_{11}+b_{22}+b_{33}}3 I_3 + \begin{pmatrix} 0 & a_{12}  & a_{13}  \\
-a_{12}   & 0 & a_{23} \\ -a_{13} &
-a_{23} & 0 \end{pmatrix} + i \begin{pmatrix} \f{2b_{11}-b_{22}-b_{33}}3  & b_{12} & b_{13} \\ b_{12} &
\f{2b_{22}-b_{33}-b_{11}}3  & b_{23} \\ b_{13} &
	b_{23} & \f{2b_{33}-b_{11}-b_{22}}3 \end{pmatrix}.
\end{equation*}
The second decomposition in~\eqref{eq:dec_u(3)} is a Cartan decomposition of the simple Lie algebra $\mathfrak{su}(3)$. With $\mathfrak{l}:=\mathfrak{so}(3)$  and $\mathfrak{p}:= i \textrm{Sym}_0(\R^3)$, this means $[\mathfrak{l},\mathfrak{l}]\subseteq \mathfrak{l}$, $[\mathfrak{l},\mathfrak{p}]\subseteq \mathfrak{p}$ and $[\mathfrak{p},\mathfrak{p}]\subseteq \mathfrak{l}$.

\medskip

\proofstep{Proof that $\mathfrak{so}(3)\subset \fL$.}
Considering the basis
\begin{equation*}
	A_1:=\begin{pmatrix} 0 & -1 & 0 \\ 1 & 0 & 0 \\ 0 & 0 & 0 \end{pmatrix}, \qquad
A_2:=\begin{pmatrix} 0 & 0 & -1 \\ 0 & 0 & 0 \\ 1 & 0 & 0 \end{pmatrix}, \qquad
A_3:=\begin{pmatrix} 0 & 0 & 0 \\ 0 & 0 & -1 \\ 0 & 1 & 0 \end{pmatrix}
\end{equation*}
of $\mathfrak{so}(3)$ (seen as a vector space), we note that for a diagonal matrix $D=\operatorname{diag}(d_1,d_2,0)$ with $d_1,d_2\in \R$, we have $[-iK,-iD]=\sqrt{2}\tau(d_1 A_2 + d_2A_3)$, so that
\begin{equation*}
	[-iK,-iV] = 2\sqrt{2}\tau( A_2 - A_3), \qquad [-iK,-iU] = \sqrt{2}\tau(U_1 A_2 + U_2A_3).
\end{equation*}
As $\tau > 0$, it follows that $A_2-A_3$ and $U_1 A_2 + U_2A_3$ are in $\fL$. As $U_1+U_2 > 0$, this implies that both $A_2$ and $A_3$ are in $\fL$. Lastly, we get $A_1 =[A_2,A_3]\in \fL$, and conclude that $\mathfrak{so}(3) \subset \fL$.

\medskip

\proofstep{Proof that $\mathfrak p:=i{\rm Sym}_0(\R^3)\subset \fL$.} Let us first show that the adjoint representation
\begin{equation*}
	{\rm ad}: \mathfrak{so}(3) \ni A \mapsto [A,\bullet] \in  {\cL}({\mathfrak p})
\end{equation*}
is irreducible. By Weyl's complete reducibility theorem, this representation can be decomposed as a direct sum of irreducible representations of $\mathfrak{so}(3)$. This leads to a decomposition $\mathfrak{p} = \bigoplus_{k=1}^K V_k$, where $V_1,\dots,V_K$ are ad-invariant subspaces of $\mathfrak{p}$, with
\begin{equation}\label{eq:irreps}
	\sum_{k=1}^K\dim V_k =  \dim \mathfrak{p} = 5.
\end{equation}
Since the real irreducible representations of $\mathfrak{so}(3)$ of dimension at most $5$ have dimension $1$, $3$, $4$ or $5$, the only possibilities to get~\eqref{eq:irreps} are either $K=1$ (in which case ${\rm ad}$ is irreducible), or $K=2$ ($4+1=5$), or $K=3$ ($1+1+3=5$), or $K=5$ ($1+1+1+1+1=5$). To rule out the possibility of ${\rm ad}$ being reducible, it suffices to show that $\mathfrak{p}$ does not contain a one-dimensional ${\rm ad}$-invariant subspace. Assume that there is one, that is that there exists a nonzero matrix $P\in \mathfrak{p}$ such that for every $A\in \mathfrak{so}(3)$, there exists $\mu(A)\in \R$ such that $[A,P]=\mu(A)P$. Then, $\mu:\mathfrak{so}(3)\to \R$ would be a Lie algebra homomorphism, and since $\mathfrak{so}(3)$ is simple, we would have $\mu(\mathfrak{so}(3))=\mu([\mathfrak{so}(3),\mathfrak{so}(3)])=[\mu(\mathfrak{so}(3)),\mu(\mathfrak{so}(3))]=0$. This shows that $P$ would commute with every element of $\mathfrak{so}(3)$. As ${\rm SO}(3)$ is compact and connected, $\exp:\mathfrak{so}(3)\to {\rm SO}(3)$ is surjective so that the matrix $P$ would commute with every rotation. As a consequence, it would be a multiple of the identity. However, since its trace is zero, we would have $P=0$. We reach a contradiction. Therefore ${\rm ad}$ is irreducible.

Now, we observe that $-iV\in \fL\cap \mathfrak{p}$, so that $\mathfrak{p}\cap \fL\neq \{0\}$. Therefore, the vector space $W:=\mathfrak{p}\cap \fL$ is a non-trivial ${\rm ad}$-invariant subspace of $\mathfrak{p}$ and is therefore equal to $\mathfrak{p}$. Therefore, $\mathfrak{p}\subset \fL$.

\medskip

\proofstep{Proof that $iI_3\in \fL$.}
As $iU-i \f{(\tr U)}3 I_3 \in {\mathfrak p} \subset \fL$, $iU\in \fL$, and $\tr(U) = U_1+U_2 > 0$, we finally obtain that $i I_3 \in \fL$. Hence, by~\eqref{eq:dec_u(3)}, $\fL=\mathfrak u(3)$.

\medskip

\noindent
This concludes the proof.
\end{proof}

\end{document}